\documentclass[11pt,a4paper]{article}

\usepackage[margin=2.5cm,top=1.5cm,bottom=1.5cm]{geometry}

\usepackage{natbib,bbm}
\usepackage{mathtools,multirow}
\usepackage{float,caption,physics}
\usepackage{relsize}
\usepackage{makeidx,wasysym,algorithmicx,algorithm,algpseudocode}
\usepackage[maxfloats=80]{morefloats}
\usepackage{amssymb,amsmath,amsthm,color}  
\usepackage{url,pdfsync,graphicx,subfigure}

\usepackage{framed}

\newtheorem{prop}{Proposition}

\def\M{\mathbb{M}}
\def\x{\mathrm{x}}
\def\tx{\tilde{\mathrm{x}}}

\newcommand\given[1][]{\:#1\vert\:}

\newcommand{\one}[1]{\pmb{1}_{#1}}
\newcommand{\1}[1]{\mathbbm{1}\left\lbrace #1 \right\rbrace}

\newcommand{\E}{\mathrm{E}}
\newcommand{\Var}{\mathrm{Var}}

\newcommand{\sm}{\overset{iid}{\sim}}

\def\mathLarge#1{\mbox{\normalsize $#1$}}

\def\d{\mathrm{d}}

\providecommand{\keywords}[1]
{
  \small	
  \textbf{\textit{Keywords---}} #1
}

\title{Bayesian prediction of jumps in large panels of time series data}
\author
{
    A. Alexopoulos
    \thanks{Department of Statistical Science,
        University College London, UK, and MRC Biostatistics Unit, University of Cambridge, UK.
                Email: {\tt angelos@mrc-bsu.cam.ac.uk}.
        }
    \and    
         P. Dellaportas
    \thanks{Department of Statistical Science, University College London, UK, The Alan Turing Institute, London, UK and Department of Statistics, AUEB, Greece.
                Email: {\tt p.dellaportas@ucl.ac.uk}.
        } 
  \and           
       O. Papaspiliopoulos
   \thanks{Department of Decision Sciences, Bocconi University, Milan, Italy.
               Email: {\tt omiros@unibocconi.it}.
        }
   }

\date{{\small \today}}

\begin{document}
\maketitle

\begin{abstract}
We take a new look at the problem of disentangling the volatility and jumps processes of daily stock returns. We first provide a computational framework for the univariate stochastic volatility model with Poisson-driven jumps that offers a competitive inference alternative to the existing tools. This methodology is then extended to a large set of stocks for which we assume that their unobserved jump intensities co-evolve in time through a dynamic factor model. To evaluate the proposed modelling approach we conduct out-of-sample forecasts and we compare the posterior predictive distributions obtained from the different models. We provide evidence that  joint modelling of jumps improves the predictive ability of the stochastic volatility models.
\end{abstract}



\section{Introduction}
\label{sec:intro}

It has been recognised in the financial literature that jumps in asset returns occur clustered in time and affect several stock markets within a few hours or days, see for example \cite{ait2015modeling}. 
We work with a large panel of stocks from several European markets and we aim to identify and predict joint tail risk expressed as probabilities of jumps in their daily returns.  Our modelling assumptions are based on the well-known paradigms of  stochastic volatility (SV) models  \citep{taylor1982financial} combined with Poisson-driven jumps \citep{andersen2002empirical}. The estimation of SV models is preferably conducted by using likelihood-based approaches \citep{harvey1994multivariate} while Bayesian methods, in particular, are notably desirable since they deal efficiently with likelihood intractability problems; see for example  \cite{jarquier1994bayesian}, \cite{chib2002markov} and \cite{kastner2014ancillarity} for detailed discussions. This  motivates us to utilize the Bayesian approach both for modelling and inference purposes. Our purpose is to provide a general modelling approach for the time and cross-sectional dependence of jumps in multiple time series. The resulting Bayesian hierarchical models require careful prior specifications and sophisticated, modern Markov chain Monte Carlo (MCMC) inference implementation strategies. To that end, we build the necessary algorithmic framework by developing robust computational methods. To evaluate the proposed model we focus on forecasting future stock returns. This is a problem of primary interest in financial statistics since it is closely related with risk management, portfolio allocation and asset pricing; see for example \cite{aguilar2000bayesian}, \cite{rapach2013forecasting} and \cite{clements2017forecasting} for more detailed discussions. We test the proposed methods in an out-of-sample forecasting scenario with real data.

\begin{figure}[t]
	\centerline{\includegraphics[width=5in]{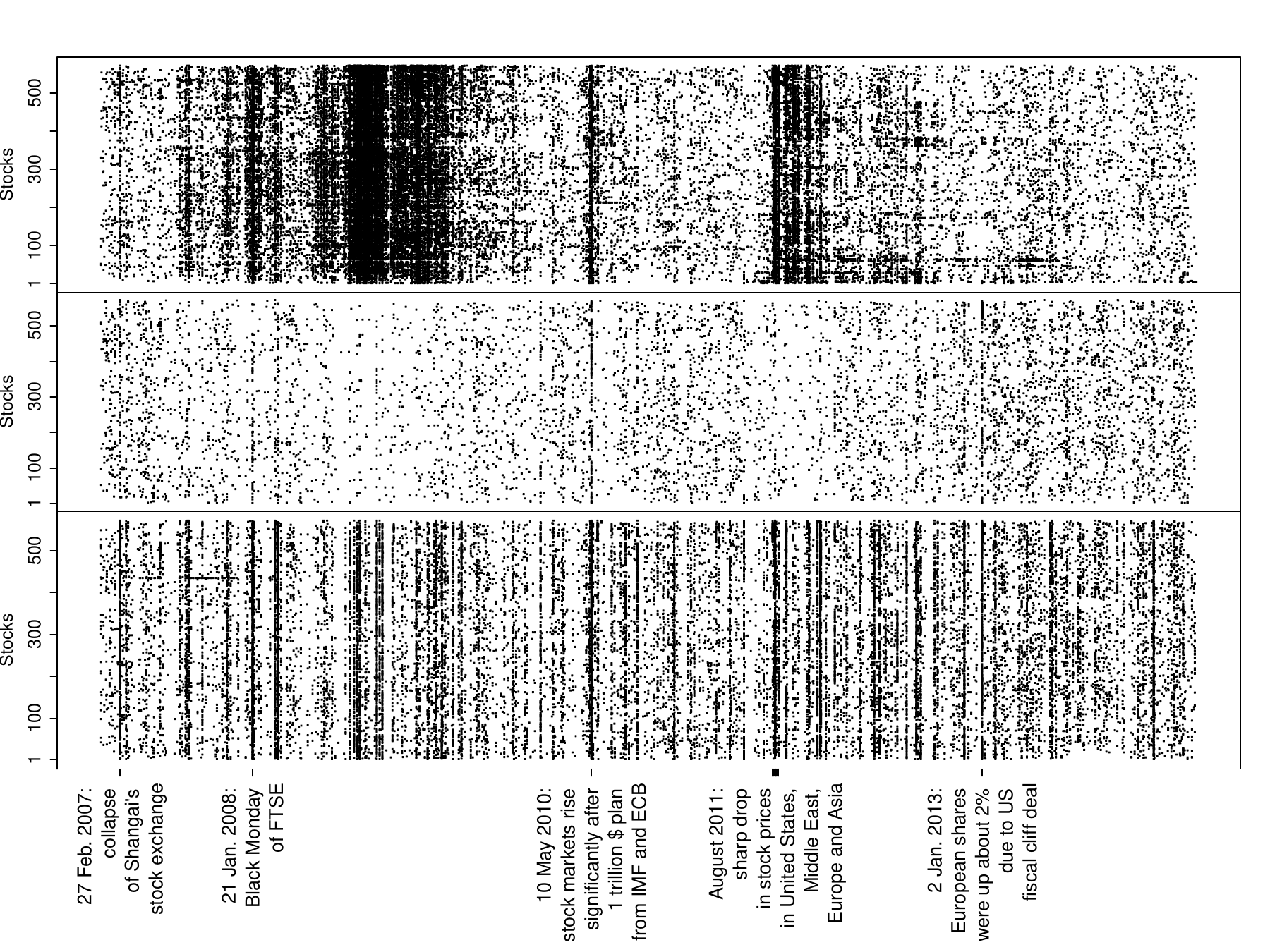}}
	\vspace*{-0.3cm}
	\caption{Each dot indicates a date between $10/1/2007$ to $11/6/2014$ (x-axis) in which at least one jump has been identified in observed daily log-returns of $571$ stocks from the STOXX Europe $600$ Index. The plot in the top panel is constructed by using an empirical method for detection of outliers in the data. The dots in the middle and in the bottom panels represent probabilities of jump greater than $.5$ for models with independent and jointly modelled across stocks intensities respectively.}
	\label{fig:compJumps}
\end{figure}

Our first point of departure is the univariate SV model with jumps.  We provide a new efficient MCMC algorithm combined with careful prior specification that improves upon existing MCMC strategies.  Next, we extend this approach by utilising the panel structure of stock returns so that unobserved jump intensities are assumed to propagate over time through a dynamic factor model. The resulting Bayesian hierarchical model can be viewed as a generalisation of the  \cite{bates1996jumps}, univariate SV model with jumps since it can be  obtained as a special case by just assuming that jump intensities are independent across time and stocks. Furthermore, our computational techniques are relevant to other scientific areas such as, for example, neuroscience, whereas simultaneous recordings of neural spikes are often modelled through latent factor models \citep{buesing2014clustered}. 

To illustrate our motivation, consider the daily returns from the stocks of the STOXX Europe $600$ Index over a period $10/1/2007$--$11/6/2014$. The top panel of Figure \ref{fig:compJumps} has been created by empirically identifying a jump when the return exceeds three standard deviations, where the estimators of mean and variance of each series were robustly estimated as the median and the  robust scale estimator of \citep{rousseeuw1993alternatives} respectively; see the supplementary material for the details of the empirical method used for outlier detection. The middle panel depicts a summary statistic of a SV model with jumps estimated separately for each stock return series: each dot denotes a stock return in which that particular day the probability of a jump in the series has been estimated to be greater than $.5$; by choosing $.5$ as threshold we avoid to over- or underestimate the number of jumps for each stock at each day, see the supplementary material for more details. Compared with the top panel, it provides a more sparse jump process indicating a successful separation of jumps from persistent moves modelled by the volatility process, a topic discussed in detail by, for example, \cite{ait2004disentangling}.  The notable feature that inspired this work is that clustering and inter-dependence of jump intensities is evident in the middle panel and, interestingly,  days in which jumps have been simultaneously identified in a large number of stocks coincide with days of events that affected the financial markets worldwide as pointed out at the x-axis of the plots.  Thus, one may attempt a joint Markovian  modelling of jump intensities across time which might reveal some predictability of jumps. Indeed, the results of our proposed joint modelling formulation provides strong evidence for better out-of-sample predictability with estimated jump probabilities higher than $.5$ depicted in the lower panel of Figure \ref{fig:compJumps}.

We propose a modelling approach for the jump processes of SV models in which the time and cross-sectional dependence of the jump intensities are driven by a latent, common across stocks,  dynamic factor model. We carefully specify informative prior distributions to the parameters of the factor model so that the implied priors for the jump intensities in each stock are in concordance with the well-known \citep{eraker2003impact} a priori expectation for one jump every few months. To avoid identifiability problems of the factor models extra assumptions for the form of the loadings matrix are needed. Nevertheless, it is recognised \citep{ghosh2009default} that MCMC algorithms developed under identifiability constraints exhibit slow convergence to the target distribution of interest. To improve the mixing properties of the proposed MCMC algorithm we follow \cite{bhattacharya2011sparse} and we do not impose the usual  (see, e.g., \cite{aguilar2000bayesian}) identifiability constraints on the factor loadings parameters. We also note that there is a considerable amount of literature dedicated to the modelling of time series count data; see for example \cite{fokianos2009poisson}, \cite{jung2011dynamic}, \cite{pillow2012fully} and \cite{buesing2014clustered}. Our modelling perspective is applied on unobserved count data and, thus, the aim of the proposed latent factor model is twofold. First, since a few jumps are expected to occur during the observation period we borrow information across all the stocks in order to learn the unobserved processes that drive the evolution of jumps across time and stocks. Second, we exploit the dimension reduction achieved by the latent factors to preserve the scalability of our method.

Our Bayesian inference implementation strategy is based on an MCMC algorithm that alternates sampling of the latent volatility process and its parameters with sampling of the jump process and its parameters. Both these steps require drawing of the high-dimensional paths of the latent volatilities and factors from their full conditional distributions. By noting that the target distributions are proportional to the product of intractable, non-linear likelihood functions with Gaussian priors, we utilize the sampler proposed by \cite{titsias2018auxiliary} to simultaneously draw the whole path of each latent process at each MCMC iteration. This is an important feature of our methods for the following reasons. First updating each state of a volatility process separately results in Markov chain samplers with slow mixing \citep{shephard1994comment}. Second, although \cite{chib2002markov} and \cite{nakajima2009leverage} use a mixture of Gaussian distributions approximation to also simultaneously update the whole volatility path, their methods rely on an importance sampling step which is quite problematic, see \cite{johannes2009optimal} for a detailed discussion and Section \ref{sec:chib} of this paper for a simulation that illustrates this issue.

To assess the predictive performance of the different models we compare the corresponding posterior predictive distributions by utilizing proper scoring rules \citep{gneiting2007strictly}.  We provide a full-fledged quantitative evaluation of the obtained forecasts by calculating logarithmic scores, such as predictive log Bayes factors \citep{geweke2010comparing}, interval and continuous ranked probability scores \citep{gneiting2014probabilistic} as well as root mean squared errors. We estimate the posterior predictive distributions by employing sequential Monte Carlo methods such as the particle filters \citep{chopin2020pf} and the annealed importance sampling \citep{neal2001annealed} algorithms.

The structure of the remaining of the paper is as follows. In Section \ref{sec:lit_review} we provide a review of the literature related to univariate and multivariate SV models in econometrics as well as the modelling of jumps in financial applications. Section \ref{sec:Model} presents our proposed modelling framework and Section \ref{sec:Bayes} describes the methods that we develop to conduct Bayesian inference for the proposed model.  In Section \ref{sec:pred} we present the computational techniques that we use to assess the predictive performance of the proposed model.  Section \ref{sec:realdata} presents results and insights from the application of our methods on the real dataset and Section \ref{sec:discuss} concludes with a small discussion.

\section{Related work}
\label{sec:lit_review}

There are two main classes of models that have been considered in the literature for the modelling of financial returns. The first class is based on observation-driven models where the time evolution of the variance of the returns is modelled by using past observations and variances. These models are known as autoregressive conditional heteroscedasticity (ARCH) models and have been developed by \cite{engle1982autoregressive}. A very popular parametrization, the generalized ARCH model, introduced later by \cite{bollerslev1986generalized} while \cite{bollerslev1987conditionally} and \cite{harvey2008beta} presented modifications that account for the heavy tails of the distributions of the observed returns; see also \cite{harvey2013dynamic} and references therein. \cite{creal2011dynamic} extended this modelling approach in the multivariate case; for the Bayesian analysis of these models see, for example, \cite{vrontos2000full} and \cite{vrontos2003full}. The second class is known as parameter-driven models where the time evolution of the variance of the returns is modelled by employing latent stochastic processes.  The models that belong to this class are the SV models which constitute the base of our proposed modelling approach.

The SV model proposed by \cite{taylor1982financial} in order to account for the time evolution of the volatilities in financial time series while \cite{hull1987pricing} and \cite{chesney1989pricing} have employed the SV model as discretization of continuous time SV diffusion processes. As mentioned earlier some early works on the Bayesian estimation of SV models include, but are not limited to, the papers presented by \cite{jarquier1994bayesian}, \cite{kim1998stochastic} and \cite{chib2002markov}. \cite{omori2007stochastic} dealt with efficient Bayesian estimation of the SV model by taking into account the correlation between the observation error and the error of the stochastic variance process; see also \cite{silva2006extended} and \cite{delatola2011bayesian} for other extensions of the basic SV model. A more recent contribution on the Bayesian analysis of univariate SV models has been offered by \cite{kastner2014ancillarity} while \cite{zhang2020stochastic} developed a new class of SV models with autoregressive moving average errors. Considering a jump component to model the empirically observed fat tails of stock returns proposed, in his pioneering paper, by \cite{merton1976option} while \cite{bates1996jumps} and \cite{andersen2002empirical} combined jump diffusions with SV models. Bayesian estimation of SV models with jumps has been considered, among others, by \cite{chib2002markov},  \cite{nakajima2009leverage} and \cite{eraker2003impact} and \cite{johannes2009optimal} where the latter two papers account for jumps in the volatility process as well; see also \cite{bandi2016price} for more recent research on the occurrence of contemporaneous jumps in the returns and volatilities.

To conduct joint modelling of the observed financial time series, multivariate SV models have been considered more than twenty years ago; see for example \cite{harvey1994multivariate} and \cite{asai2006multivariate} for a review. The existing multivariate SV models assume either constant correlations over time or some form of dynamic correlation modelling through factor models with factors being independent univariate stochastic volatility models. The latter approach has been considered, for example, by \cite{chib2006analysis}, \cite{kastner2017efficient} and \cite{kastner2019sparse}. In a similar framework \cite{chan2018bayesian} and \cite{kastner2020sparse} deal with Bayesian estimation of time varying parameter vector autoregressive models with stochastic volatility. An alternative multivariate SV modelling approach proposed recently by \cite{dellaportas2015scalable} where all the elements of the covariance matrix of the observed time series evolve over time. More closely related to the present paper, \cite{bollerslev2008risk}, \cite{jacod2009testing} and \cite{clements2017forecasting}  study the presence of common or not common jumps in multiple time series. \cite{ait2015modeling} model the dependencies of jumps by combining the SV model with a multivariate Hawkes process \citep{hawkes1971spectra} that models the propagation of jumps over time and across assets by introducing a feedback from the jumps to their intensities and back. They employ the generalised version method of moments to estimate the parameters of the model. We close this literature review section by noting that tools for the identification of non-observed events that affect the evolution of time series have been also developed by \cite{hamilton1989new} and \cite{billio2012combination} that deal with the problem of identifying structural changes in time series of economic indices.

\section{Latent jump modelling}
\label{sec:Model}

\subsection{Notation}

We denote the univariate Gaussian distribution with mean $m$ 
and variance $s^2$ by  $\mathcal{N}(m,s^2)$,   and its density evaluated at $x$ by $\mathcal{N}(x|m,s^2)$;  $\mathcal{N}_d(X|M,S)$ denotes the density of the $d$-variate normal distribution with mean $M$ and covariance matrix $S$ evaluated at $X$;  $Gam(\alpha,\beta)$ and $IGam(\alpha,\beta)$ denote the gamma, with mean $\alpha/ \beta$, and inverse gamma, defined as the reciprocal of gamma,  distributions respectively; $U(\alpha,\beta)$ the uniform distribution on interval $(\alpha,\beta)$;  $\Gamma(\cdot )$ denotes the gamma function. Index $i$ is used for stocks and index $t$ for times, e.g., $r_{it}$ is the $i$th stock return at day $t$; when we introduce factor models index $k$ denotes factor.   Upper case letters denote vectors and matrices, and lower case letters denote scalars. For a collection of variables indexed by time or stock and denoted by a lower case  character, the corresponding upper case character denotes their vector or matrix representation, e.g. $R$ denotes the matrix whose elements are $r_{it}$;  $R_i=(r_{i1},\ldots,r_{iT})$, and $R_t = (r_{1t},\ldots,r_{pt})$ are then generic rows and columns of $R$; all vectors are understood as column vectors. A subscript, of the form $(t+1):(t+\ell)$, where $\ell$ is a positive integer, indicates the collection of vectors with subscripts $t+1,\ldots,t+\ell$; e.g. $R_{t+1:t+\ell}$ is the collection of vectors $R_{t+1},\ldots,R_{t+\ell}$. The transpose of a vector or matrix is denoted by a prime; e.g. $R^\top$ is the transpose of $R$; $\mathbbm{1}\{ \mathord{\cdot} \}$ denotes an indicator function that takes value $1$ if the event in brackets is true and $0$ otherwise; for a function $f(x)$ we set $\nabla \eqqcolon ( \partial f/\partial x_1,\ldots,\partial f/\partial x_d  )$.
  
\subsection{Stochastic volatility with jumps}
\label{sec:SV}
We model stock returns as  
\begin{equation}\label{eq:obs}
r_{it} = \exp(h_{it}/2) \epsilon_{it} + \sum_{\kappa=1}^{n_{it}} \xi_{it}^{\kappa}, \,\,\,\, \epsilon_{it} \sim \mathcal{N}(0,1), \,\,\,\, t=1,\ldots,T
\end{equation}
\begin{equation}\label{eq:state}
h_{it} = \mu_i + \phi_i (h_{i,t-1}-\mu_i) + \sigma_{i\eta}\eta_{it}, \,\,\,\, \eta_{it} \sim \mathcal{N}(0,1),  \,\,\,\, t=1,\ldots,T
\end{equation}
with $h_{i0} \sim \mathcal{N}(\mu_i, \sigma_{i\eta}^2/(1-\phi_i^2))$ and $\epsilon_{it}$ and $\eta_{it}$ are independent. 
The number of jumps of the $i$th stock at time $t$ follow a Poisson distribution,
\begin{equation}\label{eq:poisson}
n_{it} \sim Poisson( \Delta_{it}\lambda_{it} ),
\end{equation}
where $\lambda_{it}$ denotes the corresponding jump intensity, $\Delta_{it}$ is a time-increment associated to each return and $ \xi_{it}^{\kappa} \sim \mathcal{N}( \mu_{i\xi}, \sigma^2_{i\xi}  )  $ is the $\kappa$th jump size. We explicitly take into account that stock returns are computed over varying time increments $\Delta_{it}$, such as in-between weekends and holidays, hence $\lambda_{it}$ is the daily jump intensity.

We model the parameters $\mu_i,\phi_i$ and $\sigma_{i\eta}^2$ of each log-volatility process and the parameters $\mu_{i\xi}$ and $\sigma^2_{i\xi}$ of the jump sizes as independent across stocks. For $\mu_i$, $\phi_i$ and $\sigma^2_{i\eta}$ we follow the related literature (\cite{kim1998stochastic}, \cite{omori2007stochastic}, \cite{kastner2014ancillarity}) and take $\mu_i \sim\mathcal{N}(0,10)$, $(\phi_i+1)/2 \sim Beta(20, 1.5)$ and $\sigma^2_{i\eta} \sim Gam(1/2,1/2)$.

For $\mu_{i\xi}$ and $\sigma_{i\xi}^2$ we choose proper priors. In fact, improper priors will lead to improper posterior. This is due to the fact that according to the model there is positive probability that there are no jumps for the whole period of time, i.e., the event $n_{it}=0$ for all $t$, has positive prior probability, and in such an event, $\mu_{i\xi}$ and $\sigma_{i\xi}^2$ are unidentifiable in the likelihood.  Instead, we take $\sigma_{i\xi}^2 \sim IGam(3,\textrm{range}_i^2/18)$, where $\textrm{range}_i = \max_t\{r_{it}\}-\min_t\{r_{it}\}$, and  $\mu_{i\xi} \sim \mathcal{N}(0,5\textrm{range}_i^2)$. In this specification, $\E[\sigma^2_{i\xi}]=\textrm{range}_i^2/36$, hence we match the variance of jump sizes with a quantity that relates to the sample variance of returns. A same reasoning is followed for setting the prior variance of $\mu_{i\xi}$ and we have found that the choice of the multiplier, here taken 5, is not crucial in the analysis.

\subsubsection{Modelling jumps independently across stocks and time}
\label{sec:indepjumpintens}
By modelling $\Lambda_i$ as  independent random vectors, 
we obtain $p$ independent univariate SV with jumps models. Additionally, each vector can be taken as one of independent elements, $\lambda_{it} \sim Gam(\delta,c)$, in which case one obtains the models described, among others, by \cite{bates1996jumps}, \cite{chib2002markov} and \cite{eraker2003impact}. 
It follows that $n_{it}$ has, marginally with respect to $\lambda_{it}$, a negative
binomial distribution with density 
\begin{equation}
  \label{eq:pmf-nbn}
  p(n_{it}) =\frac {\Gamma(\delta+n_{it})}{\Gamma(\delta) n_{it}!} \beta^{\delta} (1-\beta)^{n_{it}},
\end{equation}
where $\beta=c/(c+\Delta_{it})$.
The mean of the distribution is $\delta (1-\beta)/ \beta$ and the variance is
$\delta (1-\beta)/ \beta^2.$
Taking $\delta \leq 1$ ensures that
the  density is monotonically decreasing.  This is a feature we are interested
in: jumps are meant to capture infrequently large price
movements, and a priori we expect that with highest probability there
is no jump and it is  more likely to have less than more jumps, if
any. We
can choose  $c$ such that the probability of no
jump at a given time increment $\Delta_{it}$, is at least $\gamma \in (0,1)$, a user-specified threshold. Taking for example $\delta =1$ and $c = 50$ results in $\gamma = 0.98$ with $\E(\lambda_{it})=0.02$.  These choices are consistent with the prior expectation, in financial applications, for one jump every few months; see for example \cite{eraker2003impact} for a more detailed discussion.

\subsection{Joint modelling of jumps}
\label{sec:jointintens}
To capture the dependence of the jumps over time and across stocks we model $\lambda_{it}$ by using a dynamic factor model which is specified as follows. We first transform $\lambda_{it}$ to $y_{it}$ via 
\begin{equation}\label{eq:intenstrans}
\lambda_{it} = \lambda^*(1+e^{-y_{it}})^{-1},
\end{equation}
which implies that $\lambda_{it}<\lambda^*$. The parameter $\lambda^*$ is, thus, the upper bound for the jump intensities off all the stocks at any time point. This is an important feature of our modelling approach since it allows to incorporate the necessary \citep{chib2002markov, eraker2003impact} prior information about the expected number of jumps which are typically described as large but infrequent moves. Furthermore, as discussed below, by carefully choosing this upper bound we are able to construct a prior for the jump intensities which is similar to the $Gam(1,50)$ distributions that we used in the case of univariate SV models with jumps. We also note that this type of information is not necessary only in Bayesian applications; see for example \cite{honore1998pitfalls} for a similar discussion related to maximum likelihood estimation of models with jumps. Then, we model the time-dependence of jumps via $K$ latent factors $F_t$ modelled as independent autoregressive processes. The cross-sectional dependence of jumps is captured by a $p \times K$ matrix of factor loadings $W$ and a $p\times 1$ vector $B$. The joint model specification is 
\begin{align}
  \label{eq:facdyn1}
    Y_t & = B + W F_t, \\ \label{eq:facdyn2}
    F_{t} & = AF_{t-1} + E_{t}, \quad t=1,\ldots,T \\ \label{eq:facdyn3}
    F_{0} & \sim \mathcal{N}_K(0,\Sigma_F) \quad E_t \sim  \mathcal{N}(0,I)
    \end{align}
where $A$ is a $K \times K$ diagonal matrix with elements $\alpha_1,\ldots,\alpha_K$ and $\Sigma_F$ is a $K \times K$ diagonal matrix with elements $1/(1-\alpha_1^2),\ldots,1/(1-\alpha_K^2)$.

It is known that parameters of latent factor models are only identifiable up to certain transformations, such as orthogonal rotations and sign changes \citep{aguilar2000bayesian}. However, latent factor models are also used 
as low-rank predictive models in which case the out-of-sample prediction is of prior importance; see   \cite{bhattacharya2011sparse}. This is the perspective we adopt here, where we try to borrow strength from the information in large panels of stocks in order to predict with higher accuracy individual jumps by using a parsimonious factor model. Therefore, our prior specification imposes no identifiability constraints on the parameters of the factor model and are taken as $w_{ik} \sm \mathcal{N}(0,\sigma^2_{w})$, $b_i \sm \mathcal{N}(\mu_b, \sigma^2_b)$, and $\alpha_k \sm U(-1,1)$. 

The specification of the hyperparameters of the priors assigned on the parameters $A,B$ and $W$ require extra care because they affect, through \eqref{eq:intenstrans}, the imposed prior on $\lambda_{it}$. The fact that our modelling formulation of Section \ref{sec:indepjumpintens} was based on informative priors $Gam(1,50)$ provides a way to elicit informative prior specification of the hyperparameters $\sigma^2_w, \sigma^2_b$ and $\mu_b$. We first set $\lambda^* =0.15$ which is the $0.0005$ upper quantile of the $Gam(1,50)$ density. As discussed earlier, this an important characteristic of our modelling approach since it penalizes the estimation of a large number of jumps which is not consistent with the prior information for one or two jumps every $100$ days. A small sensitivity analysis, not reported here, confirms that choosing a value considerably greater than $0.15$ does not allow the efficient separation of the jumps from the underlying volatility process of the stock daily returns. Moreover, by using values notably smaller than $0.15$ we estimate quite low jump intensities and this can result in less accurate prediction of future stock returns; in the supplementary material we compare predictions obtained under different prior assumptions for the expected number of jumps. We have concluded that a value for $\lambda^*$ between $0.1$ and $0.2$ does not have a substantial effect on the results presented in the rest of the paper. Then, we choose the hyperparameters $\sigma^2_w$, $\mu_b$ and $\sigma^2_b$ to be such that the mean, the variance and the mode for the resulting prior on each $\lambda_{it}$ in \eqref{eq:intenstrans} are comparable with the corresponding quantities of the $Gam(1,50)$ distribution. Based on calculations presented in the supplementary material we set $\sigma^2_w=0.5$, $\sigma^2_b=1$ and $\mu_b =-5$. In Table \ref{tab:intens_priors} we display the quantities of interest in the two prior distributions. To evaluate the differences between the two priors we also examine the impact of different choices for their hyperparameters in the forecasts of future observations. See the supplementary material for the related results. 

\begin{table}[H]
\begin{center}
\begin{tabular}{rrrr}
\hline
\hline
Prior for $\lambda_{it}$ & Mean & Variance & Mode\\
\hline
$Gam(1,50)$ & $0.020$ & $0.00040$ & 0\\
Factor model &  $0.003$  & $0.00004$ & $0.0001$\\
\hline
\end{tabular}
\end{center}
\caption{Mean, variance and mode of the $Gam(1,50)$ prior distribution assumed for the jump intensities $\lambda_{it}$ in the independent model and for the prior induced by the dynamic factor model.}
\label{tab:intens_priors}
\end{table}

\section{Sampling from the posterior of interest}
\label{sec:Bayes}

Let $\Theta_i = (\mu_i,\phi_i,\sigma_{i\eta}^2,\mu_{i\xi}$,$\sigma^2_{i\xi},b_i)$. Our interest lies on the joint posterior of parameters and latent states $p(H,N,\Xi,W,\Theta,F,A|R)$. To draw samples from the posterior of interest we design an MCMC algorithm which proceeds as follows. At each MCMC iteration we perform stock-specific updates of $H_i$, $N_i$, $\Xi_i$, $W_i$, $\Theta_i$ and then we update the latent factors $F$ and their parameters $A$.

An important characteristic of the designed MCMC algorithm is that before updating the log-volatilities $H_i$ and the number of jumps $N_i$ we integrate out the jump sizes $\Xi_i$. This results in an efficient simultaneous update of the log-volatilities vector and direct sampling of the full conditional of the number of jumps $N_i$. Algorithm \ref{GibbsIdeal} summarizes the steps of the MCMC sampler. The algorithm presents how the full conditional densities of vectors of parameters are sampled by exploiting the conditional independence structure of the hierarchical model defined by equations \eqref{eq:obs}--\eqref{eq:poisson} and \eqref{eq:intenstrans}--\eqref{eq:facdyn3}. In Sections \ref{sec:NoJumpsSizes} and \ref{sec:mcmc_latent_Gaussian} we will describe in detail the more elaborate steps $4,6,7$ and $12$ of Algorithm \ref{GibbsIdeal} whereas the details of the other sampling steps are given in the supplementary material.

\begin{algorithm}[H]
\caption{MCMC that targets $p(H,N,\Xi,W,\Theta,F,A|R)$.}\label{GibbsIdeal}
\begin{algorithmic}[1]
\State Set the number of iterations $S$.
 \For{\texttt{$s = 1,\ldots,S$ }}
 \For{\texttt{$i = 1,\ldots,p$ }}
\State Sample from $p(H_i,\phi_i,\sigma_{i\eta}^2 |\mu_i,\mu_{i\xi},\sigma_{i\xi}^2,N_i,R_i )$ by using Metropolis-Hastings.
\State Sample from $p(\mu_i |H_i,\phi_i,\sigma_{i\eta}^2 )$ which is a Gaussian density.
\State Sample from $p(N_i| H_i,W_i,F,R_i,b_i,\mu_{i\xi},\sigma_{i\xi}^2 )$ by using rejection sampling.
\State Sample from $p(\Xi_i| H_i,N_i,R_i,\mu_{i\xi},\sigma_{i\xi}^2 )$ which is a Gaussian density.
\State Sample from $p(\mu_{i\xi}| N_i,\Xi_i,\sigma_{i\xi}^2 )$ which is a Gaussian density.
\State Sample from $p(\sigma_{i\xi}^2| N_i,\Xi_i,\mu_{i\xi} )$ which is an inverse Gamma density.
\State Sample from $p(b_i, W_i| N_i, F )$ by using Metropolis-Hastings.
\EndFor
\State Sample from $p(F,A| N,B,W )$ by using Metropolis-Hastings.
\EndFor
\end{algorithmic}
\end{algorithm}

We also emphasize that by switching off certain steps of Algorithm \ref{GibbsIdeal} we obtain a novel MCMC sampler for existing models. More precisely, Bayesian inference for univariate SV with jumps models can be conducted if the step in the $10$th line is skipped and the step in the $12$th line replaced with the step of drawing the independent jump intensities of the model. Switching off the steps in lines $6$--$10$ and $12$ results in an MCMC algorithm for univariate SV without jumps models. By removing the steps in lines $4$--$9$ of Algorithm \ref{GibbsIdeal} we obtain a sampler for a Cox model \citep{cox1955some} with intensities driven by latent factors.

\subsection{Sampling the number of jumps and jump sizes}
\label{sec:NoJumpsSizes}

To sample each $n_{it}$ we first integrate out the jump sizes $\xi_{it}$ and then we construct a rejection sampling algorithm to draw from $p(n_{it}|h_{it},\mu_{i\xi},\sigma^2_{i\xi},b_i,W_i,F_t,r_{it})$. This algorithm is based on Proposition \ref{log-conc}. Notice also that a discrete distribution with pmf $p(n)$ is called log-concave if 
$p(n)^2 \geq p(n-1)p(n+1)$ is satisfied for all $n$; see for example \cite{devroye1986non} for more details. For the remaining of this subsection the unnormalized $p(n_{it}|h_{it},\mu_{i\xi},\sigma^2_{i\xi},b_i,W_i,F_t,r_{it})$ is denoted by $\tilde{p}(n)$.
\begin{prop}\label{log-conc}
For each $i=1,\ldots,p$ and $t=1,\ldots,T$ the discrete distribution with probability mass function $p(n_{it} |h_{it},\mu_{i\xi},\sigma^2_{i\xi},b_i,W_i,F_t,r_{it})$ is log-concave for $n_{it} \geq 1$.
\end{prop}
\begin{proof} See the supplementary material. \end{proof}
Due to Proposition \ref{log-conc}, the target distribution is a unimodal distribution \citep{devroye1986non}. Let $m$ be an integer at the right of the mode of the distribution.
We define the probability mass function (pmf) 
$$
q_m(n) = 
\begin{cases} 
\mathLarge{
\frac{\tilde{p}(n)}{c}}, & \mbox{if } n < m \\ \mathLarge{ \frac{\tilde{p}(m)}{c}\left(\frac{\tilde{p}(m+1)}{\tilde{p}(m)}\right)^{n-m}}, & \mbox{if } n \geq m 
\end{cases}
$$
with $c=\sum_{n=0}^{m-1}\tilde{p}(n)+\tilde{p}(m)^2/(\tilde{p}(m)-\tilde{p}(m+1))$. The pmf $q_m(n)$ is proportional to $\tilde{p}(n)$ for $n < m$ and  proportional to the geometric pmf with parameter, $(\tilde{p}(m)-\tilde{p}(m+1))/\tilde{p}(m)$, for $n \geq m$. By noting that if a random variable $X$ follows the exponential distribution with parameter, 
$\log(\tilde{p}(m))-\log(\tilde{p}(m+1))$, 
then $\lfloor X \rfloor$ follows the geometric distribution with parameter  $(\tilde{p}(m)-\tilde{p}(m+1))/\tilde{p}(m)$, we observe that $q_m(n)$, for $n \geq m$,  corresponds to an exponential density that goes through the points $(m,\tilde{p}(m))$ and $(m+1,\tilde{p}(m+1))$.
Log-concavity of the distribution implies that $cq_m(n) \geq \tilde{p}(n)$, for $n \geq m$, since any straight line that goes through the points $(m,\log(\tilde{p}(m)))$ and $(m+1,\log(\tilde{p}(m+1)))$ bounds from above the logarithm of $\tilde{p}(n)$. See Figure \ref{fig:Fig2} for an illustration. A rejection algorithm to sample from the distribution with pmf $p(n_{it} |h_{it},\mu_{i\xi},\sigma^2_{i\xi},b_i,W_i,F_t,r_{it})$ proceeds as summarized by Algorithm \ref{RJ}. Note that Algorithm \ref{RJ} implies the evaluation of $\tilde{p}(0),\dots,\tilde{p}(m+1)$ at each iteration of the proposed MCMC algorithm.

\begin{figure}[t]
	\centerline{\includegraphics[scale=0.4]{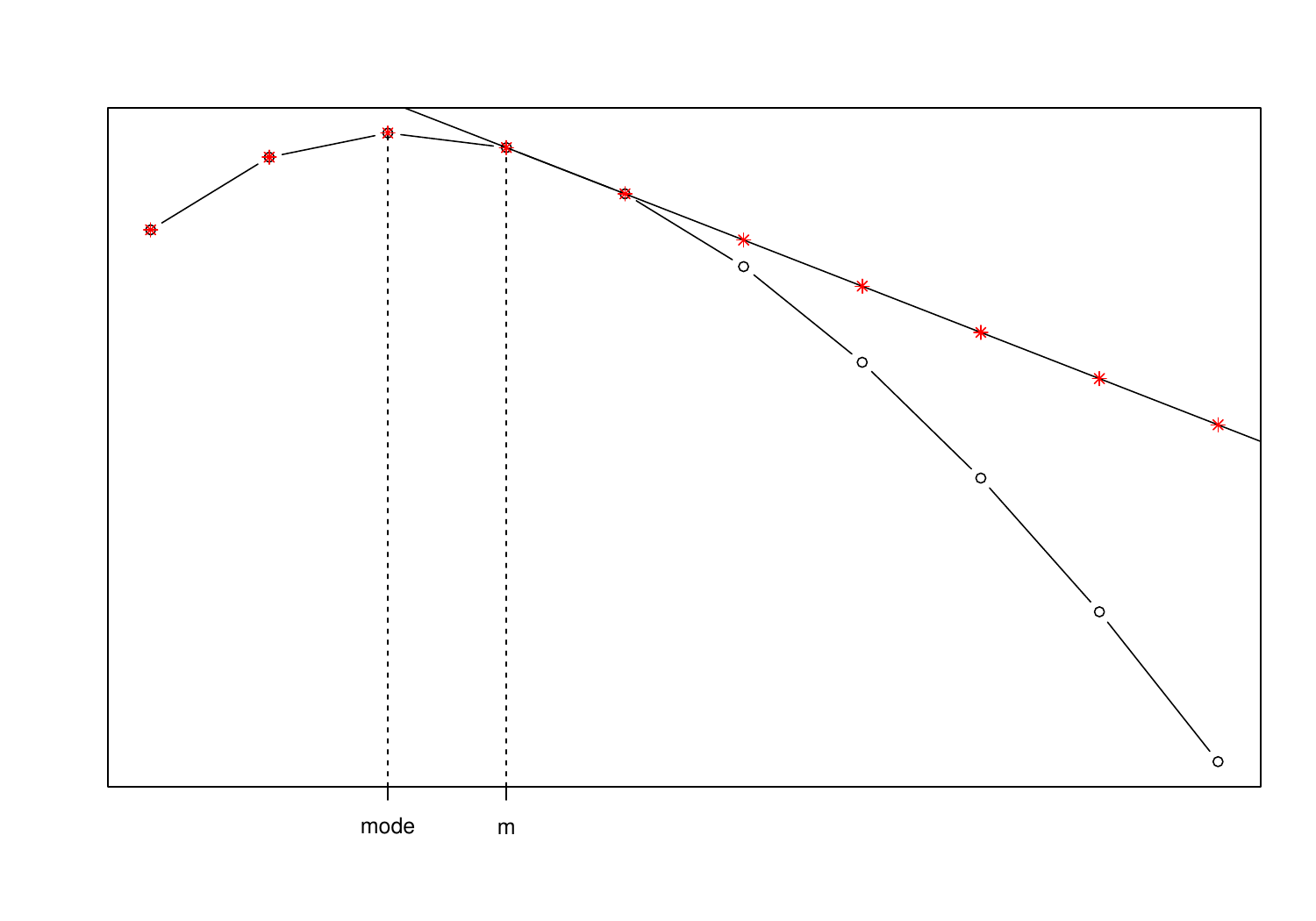}}
	\vspace*{-0.5cm}
	\caption{$\tilde{p}(n)$ and its bounding line in the log-scale. Circles: $\log(\tilde{p}(n))$, stars: $\log(q_m(n))$.}
	\label{fig:Fig2}
\end{figure}

\begin{algorithm}[H]
\caption{Rejection sampler for the number of jumps.}\label{RJ}
\begin{algorithmic}
\State $\bf{Generate}$ $u \sim U(0,1)$
\If {$v \leq \frac{\sum_{n=0}^{m-1}\tilde{p}(n)}{c}$}
  \State $\bf{generate}$ $n$ from $q_m(n)$ truncated from the right at the point $m-1$, by using the inversion method. 
\Else
        \State $\bf{repeat}$:
        \State $\bf{generate}$ $V \sim U(0,1)$
        \State $\bf{generate}$ $n$ from the geometric pmf with parameter $(\tilde{p}(m)-\tilde{p}(m+1))/\tilde{p}(m)$, truncated from the left at the point $m$.
\State $\bf{until}$ $Vcq_m(n) \leq \tilde{p}(n)$
\EndIf\\
\Return $n$ 
\end{algorithmic}
\end{algorithm}

\subsubsection{Jump sizes}

Drawing the jump sizes $\xi_{it}$ given the number of jumps $n_{it}$ is based on the following proposition. 
Let $I_n$ and $\one{n}$ be the $n \times n$ identity matrix and an $n$-dimensional vector with ones respectively.
\begin{prop} According to the model defined by equations \eqref{eq:obs} and \eqref{eq:state} we have that $p(\xi_{it} |n_{it}, h_{it},\mu_{i\xi},\sigma^2_{i\xi}, r_{it})$ is equal to the Gaussian density
\label{prop2}
\begin{equation*}
 \mathcal{N}_{n_{it}} \left
  (\xi_{it} \given[\Big] (\mu_{i\xi} \, \frac{e^{h_{it}}}{e^{h_{it}}+n_{it}\sigma_{i\xi}^2}+\frac{r_{it}}{e^{h_{it}} +n_{it}\sigma_{i\xi}^2}n_{it}\sigma_{i\xi}^2)\one{n_{it}} ,
 \sigma_{i\xi}^2I_{n_{it}}-\frac{\sigma_{i\xi}^4}{e^{h_{it}}+n_{it}\sigma_{i\xi}^2}\one{n_{it}}\one{n_{it}}' \right),
\end{equation*}
for each $i=1,\ldots,p$.
\end{prop}   
\begin{proof} See the supplementary material. \end{proof}

\subsection{Log-volatilities and factors: MCMC for latent Gaussian models}
\label{sec:mcmc_latent_Gaussian}

In the $4$th and $12$th lines of Algorithm \ref{GibbsIdeal} we sample the paths of the log-volatilities $H_i$ and factors $F$ respectively. Moreover, each path is updated jointly with the parameters of the corresponding autoregressive models defined by equations \eqref{eq:state} and \eqref{eq:facdyn3}.  These steps correspond to the joint sampling of latent paths and parameters of latent Gaussian models. A latent Gaussian model is defined in terms of a non-Gaussian likelihood and a latent Gaussian field controlled by some parameters. By denoting with $\mathrm{X}$ a $d$-dimensional random vector that corresponds either to a latent log-volatility or to a factor path, with $\x$ its realization and with $\theta$ any parameters we target the density 
\begin{equation}
\label{eq:lg_target}
p(\x,\theta) \propto \exp\{g(\x)\}\mathcal{N}(\x|0,C_{\theta})\pi(\theta) 
\end{equation}
where $g(\x)$ is the log-likelihood function, $C_{\theta}$ denotes the covariance matrix of the latent Gaussian field, parametrized in terms of $\theta$ and $\pi(\theta)$ denotes the density of the prior for $\theta$. We have assumed, without loss of the generality, that the latent Gaussian field has zero mean. We also note that only throughout the present Section we violate our notation and we use lower-case letters to denote the realizations of a random vector; the dependence on any data has been suppressed.

To draw samples from the target in \eqref{eq:lg_target} two approaches have dominated in the literature. One of them is to alternate sampling from the conditionals $p(\x|\theta)$ and $p(\theta|\x)$. Sampling from $p(\x|\theta)$ is quite challenging since usually corresponds to a high dimensional distribution. To draw the desired samples several well-established methods, based on the Metropolis-Hastings algorithm, have been used, see e.g. \cite{roberts2002langevin}, \cite{murray2010elliptical}, \cite{girolami2011riemann} and \cite{cotter2013mcmc}; in the latter approach the proposal distribution utilizes information both from the likelihood and the Gaussian prior of the latent states and this is an important feature of the proposed sampler \citep{titsias2018auxiliary}. Sampling from $p(\theta|\x)$ is also challenging since the parameters $\theta$ are highly correlated with the Gaussian latent states of the model. See, for example, \cite{papaspiliopoulos2007general}, \cite{murray2010slice} and \cite{yu2011center} for efficient methods that have been developed for the update of $\theta$, conditional on the latent Gaussian filed, in a more general class of Bayesian hierarchical models. In the second approach, joint sampling of the latent states $\x$ and the parameters $\theta$ is conducted. This perspective is considered, among others, by \cite{knorr2002block} where the latent states and the parameters are updated in, carefully selected, blocks and by \cite{titsias2018auxiliary} where a sampler which utilizes auxiliary variables to perform updates that are both prior and likelihood informed has been constructed.  

In the present paper we follow the latter approach and we update jointly the latent states and the parameters of a given model by combining the methods proposed by \cite{yu2011center} and \cite{titsias2018auxiliary}.  Following \cite{yu2011center} we interweave the so-called \citep{papaspiliopoulos2007general} centred and non-centred parametrization of the latent states and by utilizing the algorithm developed by \cite{titsias2018auxiliary} we sample jointly $\x$ and $\theta$ under each one of the two parametrizations. We show that by utilizing the proposed Metropolis-Hastings step we find efficiency gains relative to alternative samplers that can be used to update the latent log-volatilities, the factor paths and their parameters. In the rest of the Section we present the proposed methodology in the case of a general latent Gaussian model and in the supplementary material we provide the details for the application of the method in order to sample the latent factors $F$ as well as each log-volatility path $H_i$ and their parameters.

\subsubsection{Metropolis-Hastings for the joint update of latent paths and parameters}

To construct a Metropolis-Hastings step in order to sample jointly $\x$ and $\theta$ we work as in \cite{titsias2018auxiliary} and we introduce auxiliary variables $\mathrm{u} \sim \mathcal{N}_d(\x, (\delta/2)I_d)$. We consider, thus, the expanded target 
\begin{equation}
\label{eq:lg_target_exp}
p(\x,\theta,\mathrm{u}) \propto \exp\{g(\x)\}\mathcal{N}_d(\x|0,C_{\theta})\mathcal{N}_d(\mathrm{u}|\x, (\delta/2)I_d)\pi(\theta).
\end{equation}
Samples from the expanded target can be drawn by first drawing $\mathrm{u}$ from $p(\mathrm{u}|\x) = \mathcal{N}_d(\mathrm{u}|\x, (\delta/2)I_d)$ and then employing the Metropolis-Hastings algorithm to draw from the intractable $p(\x,\theta|\mathrm{u})$. To perform the latter update we utilize a proposal distribution of the form
\begin{equation}
\label{eq:prop_form}
q(\x',\theta'|\x,\theta,\mathrm{u}) = q(\x'|\x,\theta',\mathrm{u})q(\theta'|\theta).
\end{equation}
To construct $q(\x'|\x,\theta,\mathrm{u})$ we approximate the intractable log-likelihood $g(\cdot)$ by employing a first order Taylor expansion of $g(\cdot)$ around the current value $\x$. We obtain, then, the proposal density
\begin{align}
\nonumber
 q(\x' | \x,\theta',\mathrm{u})  \propto &  \exp\{ g(\x) + \nabla g(\x)^\top
 (\x'-\x) \} 
 \mathcal{N}_d(\x' | 0, C_{\theta'}) \mathcal{N}_d(\mathrm{u} | \x', (\delta/2) I_d)
 \\  
& \propto \mathcal{N}_d \left(\x'  |  \frac{2}{\delta} Q_{\theta'} (
    \mathrm{u} + \frac{\delta}{2} \nabla g(\x) ), Q_{\theta'} \right),
\label{eq:aux1}
\end{align}
where $Q_{\theta'} =  (C^{-1}_{\theta'} + {2 \over \delta}I_d)^{-1}$. To reduce the costs of generating from the proposal in \eqref{eq:aux1} as well as evaluating the resulting acceptance ratio we follow \cite{titsias2018auxiliary} and we define new auxiliary variables $\mathrm{z}$ by using the reparametrization
\begin{equation*}
\mathrm{z} \equiv  \mathrm{u} + (\delta / 2) \nabla g(\x) \sim  \mathcal{N}_d(\x  + (\delta /
2) \nabla g(\x), (\delta/2) I_d ).
\end{equation*}
The expanded target in \eqref{eq:lg_target_exp} becomes now
\begin{equation*}
\label{eq:lg_target_exp_new}
p( \x, \mathrm{z}, \theta)  \propto
\exp\{g(\x) \} \mathcal{N}_d( \x | 0, C_{\theta})  \mathcal{N}_d(\mathrm{z} | \x  + (\delta /
2) \nabla g(\x), (\delta/2) I_d ) \pi(\theta),
\end{equation*}
while the proposal in \eqref{eq:aux1} can be written as 
\begin{equation}
\label{eq:new_aux_prop}
q(\x' |\theta', \mathrm{z})= \mathcal{N}_d\left(\x' |  \tfrac{2}{\delta} Q_{\theta'} \mathrm{z}, Q_{\theta'} \right)=\frac{1}{\mathcal{Z}(\mathrm{z},\theta')}\mathcal{N}_d( \x' | 0, C_{\theta'})\mathcal{N}_d( \mathrm{z} | \x', C_{\theta'}), 
\end{equation}
where $\mathcal{Z}(\mathrm{z},\theta') = \mathcal{N}_d( \mathrm{z} | 0, C_{\theta'}+(\delta/2)I_d)$. Notice that, importantly, the proposed $\x'$ becomes conditionally independent from the
current $\x$ given the auxiliary variable $\mathrm{z}$. In particular, the proposal $q(\x',\theta'|\x,\theta,\mathrm{u})$ in \eqref{eq:prop_form} becomes
$$
q(\x',\theta'|\x,\theta,\mathrm{z}) = q(\x' |\theta', \mathrm{z})q(\theta'|\theta)
$$
and the acceptance ratio of the move is 
\begin{equation}
\label{eq:acc_ratio}
\exp\left \{g(\x') -  g(\x) +\nu(\mathrm{z},\x') - \nu(\mathrm{z},\x) \right \}  \times \frac{\mathcal{Z}(\mathrm{z},\theta') 
\pi(\theta') q(\theta|\theta')}{\mathcal{Z}(\mathrm{z},\theta) \pi(\theta)q(\theta'|\theta)}
\end{equation} 
where $\nu(\mathrm{z},\x) = \left(\mathrm{z}- \x - (\delta/ 4) \nabla g(\x)\right)^\top \nabla g(\x)$. In our application we choose $q(\theta'|\theta)$ to be a simple random walk proposal distribution with step-size $\kappa$. Algorithm \ref{alg:tp} summarizes the proposed Metropolis-Hastings step for the joint update of $\x$ and $\theta$. Finally, we note that following \cite{titsias2018auxiliary} we learn $\kappa$ and $\delta$ during the burn-in period by decomposing step $2$ in Algorithm \ref{alg:tp} in two sub-steps. We first update $\x$ alone and tune $\delta$ in order to achieve an acceptance ratio between $50\%$ and $60\%$. Then, we update $(\x,\theta)$ jointly and we tune $\kappa$ to obtain acceptance ratio in the range of $20\%$ to $30\%$. 

\begin{algorithm}[H]
\caption{Auxiliary gradient-based sampler to draw samples from \eqref{eq:lg_target}}
\begin{algorithmic}[H!]
\State $1$. Sample $\mathrm{z} \sim \mathcal{N}_d(\x  + (\delta /
2) \nabla g(\x), (\delta/2) I_d )$.
\State $2$. Propose $\theta' \sim q(\theta'|\theta)$ and $\x' \sim q(\x'|\theta',z)$ and accept them according to the Metropolis-Hastings ratio in \eqref{eq:acc_ratio}.
\end{algorithmic}
\label{alg:tp}
\end{algorithm}

\subsubsection{Interweaving strategy for the parameters}

Although drawing the latent path $\x$ jointly with the parameters $\theta$ aims to overcome convergence problems that occurring due to their high dependence, it is well-known that the parametrization of latent Gaussian models as well as of more general Bayesian hierarchical models plays a key role in the efficiency of MCMC algorithms \citep{papaspiliopoulos2007general}. To address convergence issues  
that arise from the parametrizations of the latent log-volatility models and from the latent factor model for the jump intensities we employ the, so-called, ancillarity–sufficiency interweaving strategy (ASIS) developed by \cite{yu2011center}; see also \cite{kastner2014ancillarity} for the benefits of the ASIS in the case of univariate SV models and \cite{simpson2017interweaving} where the ASIS has been applied in more general dynamic linear models. In the case of a latent Gaussian model with latent states $\x$ and parameters $\theta$ the application of ASIS is briefly described as follows. After the joint update of $\x$ and $\theta$ with the Metropolis-Hastings step described by Algorithm \ref{alg:tp} we transform the latent path $\x$ to its non-centred parametrization $\tx$ by using an one-to-one transformation. We re-draw then the parameters $\theta$ from $p(\theta|\tx)$ before transforming $\tx$ back to $\x$ by using the inverse transformation. 

To employ ASIS in our application we first note that equations \eqref{eq:obs}-\eqref{eq:state} and \eqref{eq:intenstrans}-\eqref{eq:facdyn3} define the centred parametrizations \citep{papaspiliopoulos2007general} of the corresponding models. The non-centred parametrization for the latent states of a latent Gaussian model can be found by following the reasoning developed by \cite{yu2011center}: under a non-centred parametrization $\x$ has to be an ancillary statistic for $\theta$, i.e., the distribution of $\x$ needs to be free of $\theta$. We note that there are cases where a many-to-one transformation is needed \citep{papaspiliopoulos2003non} but this is not required in our application. It is easy to see that in the case of the latent log-volatilities $h_{it}$ defined in \eqref{eq:state} the variables $\tilde{h}_{it} = (h_{it}-\mu_i)/\sigma_{i\eta}$ are in a non-centred parametrization while in the case of the latent factors defined by \eqref{eq:facdyn2} the transformation $\Gamma_t = F_t -AF_{t-1}$ results in a non-centred factor model for the jump intensities. ASIS in completed by re-drawing the log-volatility parameters $\mu_i, \phi_i$ and $\sigma^2_{i\eta}$ and the factor parameters $A$ under the non-centred parametrizations and by using the inverse transformation to return back to $h_{it}$ and $F_t$. 

In the supplementary material we provide all the details for our implementation of ASIS. To emphasize the gains obtained we conducted simulation studies which are also presented in the supplementary material. By describing briefly the results we note that combining ASIS with the auxiliary gradient-based sampler (Algorithm \ref{alg:tp}) improves substantially the efficiency of the proposed MCMC algorithm (Algorithm \ref{GibbsIdeal}) in drawing samples from the parameters $\mu_i$ and $\sigma^2_{i\eta}$ of the log-volatility processes (see Figures S.$2$ and S.$3$ in the supplementary material) as well as from the parameters in the matrix $A$ of the latent factor model for the jump intensities; see Figure S.$4$ in the supplementary material.

\subsection{Approximation by using a mixture of Gaussian distributions}
\label{sec:chib}

For a given stock, sampling the whole path of the latent log-volatilities $H_i$ at each MCMC iteration could be performed by utilizing the methodology developed by \cite{chib2002markov} and \cite{nakajima2009leverage}. Their methods are based on the idea of \cite{kim1998stochastic} to approximate a SV without jumps model by using a mixture of Gaussian distributions. 

The advantage of their approach is that the update of $H_i$ can be conducted by drawing samples directly from its full conditional distribution which is a $(T+1)$-variate Gaussian with tridiagonal inverse covariance matrix. The jumps sizes and their parameters should be updated after integrating out the mixture component indicators from the likelihood of the approximated model, otherwise convergence of the corresponding MCMC algorithm will be slow \citep{nakajima2009leverage}. Integrating out the mixture component indicators results in a partially collapsed Metropolis-Hastings within Gibbs sampler which has to be implemented with care since permuting the order of the updates can change the stationary distribution of the chain \citep{van2015metropolis}. If, for example, the update of the mixture indicators is conducted in the step which is intermediate between the update of the latent log-volatilities and the update of their parameters, then the Markov chain will not converge to the desired stationary distribution.

More importantly, by using the methods developed in \cite{chib2002markov} and \cite{nakajima2009leverage}, samples from the posterior of the parameters of the exact model are obtained with an importance sampling step. In this step, weights that are equal to the ratio between the likelihood of the exact and the approximated model, are assigned to the MCMC samples obtained from the approximated model. Although in the case of SV models without jumps the described importance sampling step works without problems, it is known \citep{johannes2009optimal} that in the case of models with jumps such weights may have large variance since a few or only one of them could be much larger than the others. Thus, the output of the importance sampling step will not be suitable for any Monte Carlo calculation. For an illustration of the problem we conducted a simulation experiment in which we calculated the required importance sampling weights for SV models with and without jumps. In particular, we simulated $p=1$ time series with $T=1,500$ log-returns from each model by utilizing equations \eqref{eq:obs}--\eqref{eq:state} and omitting the jump component in \eqref{eq:obs} to simulate from the model without jumps. In the case of the SV model with jumps we assumed independent jump intensities over time following the modelling approach presented in Section \ref{sec:indepjumpintens}. We chose values for the parameters of the log-volatility process which are common in the financial literature \citep{kim1998stochastic}; $\mu_1=-0.85, \phi_1=0.98$ and $\sigma_{1\eta} =0.15$ and we simulated the sizes of possible jumps from the Gaussian distribution with $\mu_{1\xi}=0$ and $\sigma_{1\xi}=3.5$ \citep{eraker2003impact}. Then, we drew $3,000$ (thinned) posterior samples of the parameters and latent states of the approximated models by using the corresponding MCMC algorithms.

Figure \ref{fig:chibw} displays the log-normalised importance sampling weights plus the logarithm of their number. The variance of the normalized importance sampling weights should become smaller as the importance sampling approximation improves. This implies that the histograms in Figure \ref{fig:chibw} should be centered at zero in the case of an accurate enough approximation; see the corresponding plots in \cite{kim1998stochastic} and \cite{omori2007stochastic}. The left panel of Figure \ref{fig:chibw} indicates that there is little effect of the mixture approximation in the case of the model without jumps, but the right panel points out that this is not true in the case of the SV with jumps model. The bad performance of the approximation strategy in models with jumps is due to the fact that the mixture model has been proposed by \cite{kim1998stochastic} and improved by \cite{omori2007stochastic} for SV models without jumps. By adding a jump component in the model the mixture approximation has to take into account the uncertainty in the estimation of the jump process and this is not the case in the approximation used by \cite{chib2002markov} and \cite{nakajima2009leverage}.

\begin{figure}[H]
	\centerline{\includegraphics[width=4in]{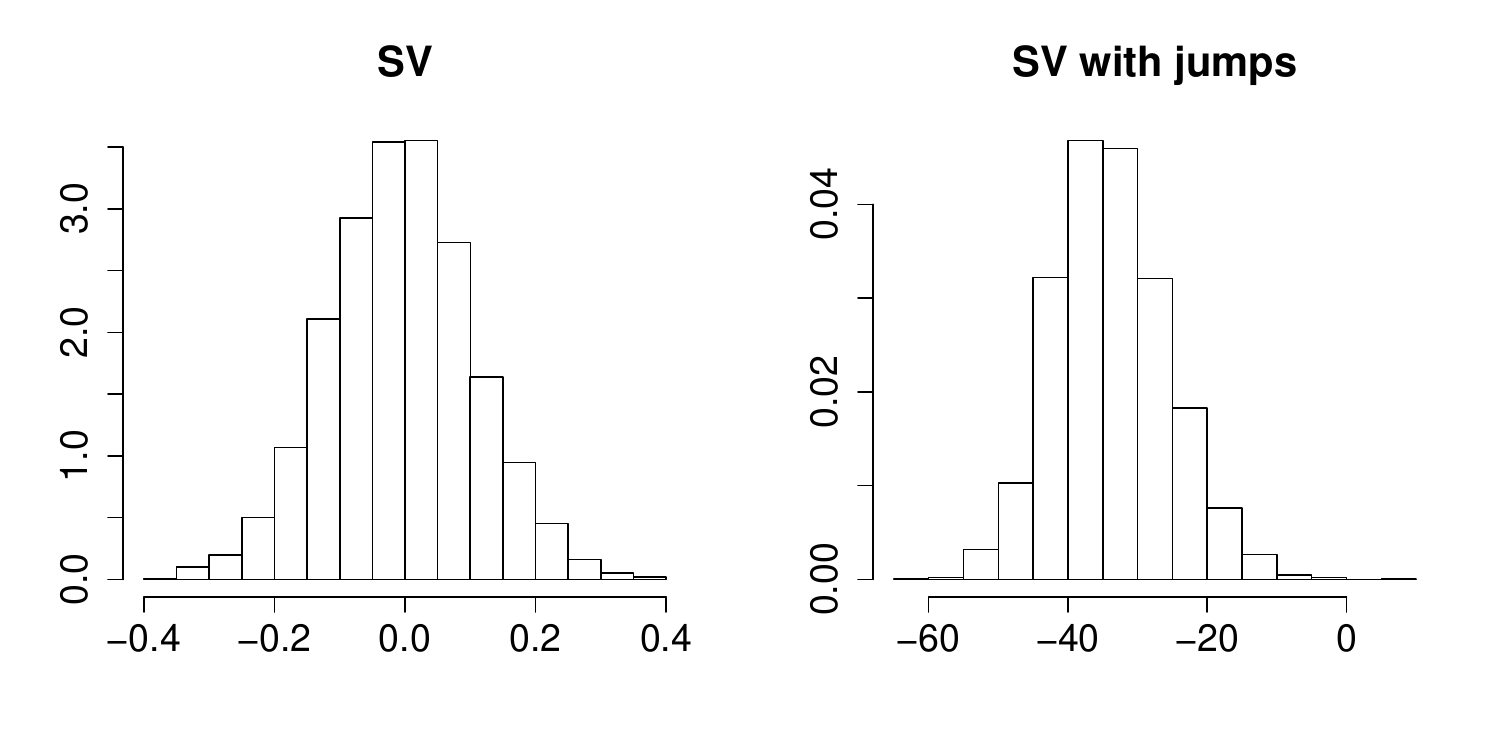}}
	\vspace*{-0.5cm}
	\caption{Histograms of normalized weights multiplied by their number, in the log-scale, calculated as the ratio between the likelihood of the exact and the approximated, by a mixture of Gaussian distributions, model in the case of stochastic volatility (SV) models with and without jumps. 
} 
	\label{fig:chibw}
\end{figure}

\section{Out-of-sample forecasting} 
\label{sec:pred}

To compare the predictive performance of different SV models we develop a computational framework for out-of-sample forecasting. Based on $T$ in-sample and $\ell$ out-of-sample observations we assess the forecasts obtained from the various models by calculating proper scoring rules \citep{gneiting2007strictly}. A scoring rule assigns a numerical score to a probabilistic forecast by taking into account the posterior predictive distribution and the realized observation and it is proper if the expected score is maximized when the estimated predictive distribution coincides with the true distribution. Proper scoring rules allow for the joint evaluation of the calibration (compatibility of forecasts and observations) and the sharpness (concentration of the predictive distributions); see, for example, \cite{geweke2010comparing} and \cite{kruger2020predictive} for more detailed discussions. We focus on three proper scoring rules: (i) logarithmic scores, and particularly, the predictive log Bayes factors which compare the models with respect to their marginal likelihoods, (ii) continuous ranked probability scores which are employed for the assessment of the predictive cumulative distribution functions (CDFs)  and (iii) interval scores which are utilized for the evaluation of the prediction credible intervals. We complete a full-fledged assessment of the predictive performance of the different models by calculating root mean squared errors of the obtained forecasts. 

To calculate these summary measures we draw samples from the posterior predictive distributions of interest and we estimate their densities by utilizing sequential Monte Carlo methods \citep{chopin2020introduction}. In the rest of the Section we present the computational framework that we built in order to conduct out-of-sample forecasts with univariate SV models with and without jumps as well as by using the proposed joint modelling approach for the jump intensities. To evaluate the accuracy of the obtained estimators we calculate the effective sample size of the samples drawn from the posterior predictive distributions. The effective sample size of a given sample is commonly employed \citep{neal2001annealed,chopin2013smc2} to assess the variance of estimators obtained from weighted rather than independent samples; see in the supplementary material for more details. For the remaining of this section we assume that the parameters $\Theta$, $A$ and $W$ are fixed to their in-sample estimated posterior means and we omit them from our notation.

\subsection{Prediction with independent SV models}
\label{sec:compuni}
Let $\M_{\mathrm{svj}}$ and $\M_{\mathrm{sv}}$ be the model that consists of $p$ independent, univariate SV models with and without jumps respectively. Let also $R_{1:T}$ and $R_{T+1:T+\ell}$ denote in-sample and out-of-sample observations. To calculate predictive Bayes factors and the continuous ranked probability and interval scores we consider the posterior predictive distributions with densities $p(r_{it}|r_{i,1:t-1})$, $t=T+1,\ldots,T+\ell$. The predictive log Bayes factors are defined as 
\begin{equation}\label{eq:mvBFuni}
\log(p(R_{T+1:T+j}|R_{1:T},\M_{\mathrm{svj}}))-\log(p(R_{T+1:T+j}|R_{1:T},\M_{\mathrm{sv}})),\,\,\,\, j=1,\ldots,\ell,
\end{equation}
where for both $\M_{\mathrm{svj}}$ and $\M_{\mathrm{sv}}$ we have that
\begin{equation*}\label{eq:margLik}
p(R_{T+1:T+j}|R_{1:T}) = \prod_{i=1}^{p}\prod_{t=T+1}^{T+j} p(r_{it}|r_{i,1:t-1}).
\end{equation*}
Moreover, if $l_{it}$ and $u_{it}$ denote the $\alpha/2$ and $1-\alpha/2$ quantiles of a $(1-\alpha)\times 100\%$ prediction interval for the value of $r_{it}$ then its interval score is given by the formula
\begin{equation}
\label{eq:interval_score}
(u_{it}-l_{it}) + \frac{2}{\alpha}(l_{it}-r_{it}) \mathbbm{1}\{r_{it} < l_{it}\} 
 + \frac{2}{\alpha}(r_{it}-u_{it}) \mathbbm{1}\{r_{it} > u_{it}\}.
\end{equation}
We note, finally, that the calculation of the continuous ranked probability scores that we employ for the evaluation of predictive CDFs is also based on samples from $p(r_{it}|r_{i,1:t-1})$; see the supplementary material for details.

We draw the desired samples and we estimate the densities $p(r_{it}|r_{i,1:t-1})$ based on the formula 
\begin{equation}
\label{eq:filteredUni}
p(h_{i,0:t} | r_{i,1:t} ) =  \dfrac{p(h_{i,0:t-1} |r_{i,1:t-1} )p(r_{it}| h_{it})p(h_{it}|h_{i,t-1})}{p(r_{it}| r_{i,1:t-1} )},
\end{equation}
and by utilizing sequential Monte Carlo methods also known as particle filters \citep{chopin2020pf}. From the output of a particle filter algorithm that targets the posterior in \eqref{eq:filteredUni} we are able to draw samples from the posterior predictive distribution of $r_{it}$ and to estimate the normalizing constants $p(r_{it}| r_{i,1:t-1} )$. See in the supplementary material for more details and for the pseudocode of the particle filter algorithm.

\subsection{Prediction by using joint modelled jump intensities}
\label{sec:compmulti}
Let $\M_{\mathrm{mvj}}$ be our proposed SV with jumps model, defined by equations \eqref{eq:obs}--\eqref{eq:poisson} and \eqref{eq:intenstrans}--\eqref{eq:facdyn3}. To compare the predictive performance of $\M_{\mathrm{mvj}}$ with the performance of competing models we need to estimate the predictive densities $p(R_{T+j}|R_{1:T+j-1},\M_{\mathrm{mvj}})$, for each $j=1,\ldots,\ell$. This could be achieved with a particle filter algorithm, constructed similarly to the case of independent SV models, as described in section \ref{sec:compuni}. However, in the case of $\M_{\mathrm{mvj}}$, particle filter algorithms deliver poor estimations of the marginal likelihoods since the importance weights, calculated at each of their steps, have large variance. 

We develop an alternative sequential importance sampling algorithm to obtain the desired estimates which is based on the annealed importance sampling (AIS) proposed by \cite{neal2001annealed}. AIS utilizes the method of simulated annealing \citep{kirkpatrick1983optimization} to construct an importance sampling algorithm for drawing weighted samples from an unnormalized target distribution. The proposal distribution of the algorithm is based on a sequence of auxiliary distributions and Markov transition kernels that leave them invariant. As in any importance sampling algorithm, the sample mean of the importance weights is an estimation of the ratio between the normalizing constants of the target and proposal distributions. Moreover, the output of AIS can be used for the estimation of expectations with respect to any of the auxiliary distributions used for the construction of the algorithm. Here we exploit these features of AIS to draw samples from the posterior predictive distributions of interest and to estimate their densities. In the rest of the section we provide a detailed description of how we apply AIS, $\M_{\mathrm{mvj}}$ is suppressed throughout except when it is necessary. 

Let $g_0,\ldots,g_{\ell}$ be the required sequence of auxiliary distributions. By noting that each $g_j$ must be proportional to a computable function and satisfies $g_j \neq 0$ wherever $g_{j-1} \neq 0$ we set $g_{\ell}(H_{0:T+\ell},F_{0:T+\ell}) = p(H_{0:T+\ell},F_{0:T+\ell}|R_{1:T+\ell})$ and  
\begin{align*}
\label{eq:ais_seq_target}
g_{j}(H_{0:T+\ell},F_{0:T+\ell})  & = p(H_{0:T+\ell},F_{0:T+\ell}|R_{1:T+j})  \\
&  = p(H_{0:T+j},F_{0:T+j}|R_{1:T+j})\prod_{t=T+j+1}^{T+\ell} p(H_t|H_{t-1})p(F_t|F_{t-1}),
\end{align*}
for $j=0,\ldots,\ell-1$.
Samples from $g_j$, $j \geq 1$, are obtained by drawing samples from $g_{j-1}$ using a few iterations of Algorithm \ref{GibbsIdeal} and the transition densities $p(H_{T+j}|H_{T+j-1})$ and $p(F_{T+j}|F_{T+j-1})$. To sample from $g_0$ we utilize samples from $p(H_{0:T},F_{0:T}|R_{1:T})$. We compute the importance weights of the obtained samples by first noting that $g_j$ are proportional to the computable functions
\begin{equation*}
g^*_{j}(H_{0:T+\ell},F_{0:T+\ell}) = p(H_{0:T+j},F_{0:T+j},R_{1:T+j})\prod_{t=T+j+1}^{T+\ell} p(H_t|H_{t-1})p(F_t|F_{t-1}),\,\, j < \ell,
\end{equation*}
and $g^*_{\ell}(H_{0:T+\ell},F_{0:T+\ell}) = p(H_{0:T+\ell},F_{0:T+\ell},R_{1:T+\ell})$.
Therefore, the importance weight for every sampled point is
\begin{equation*}
\label{eq:aisw2}
\Omega^s_{T+j} = \Omega^s_{T+j-1} \dfrac{g^*_{j}(H^s_{0:T+\ell}, F^s_{0:T+\ell} )}{g^*_{j-1}(H^s_{0:T+\ell}, F^s_{0:T+\ell} )} =  \Omega^s_{T+j-1}p(R_{T+j} | H_{T+j}^s,F_{T+j}^s),\,\,j\geq2,
\end{equation*} 
and $\Omega^s_{T+1} = p(R_{T+1} | H_{T+1}^s,F_{T+1}^s)$, $s=1,\ldots,S$, where $S$ is the number of samples obtained by AIS and 
\begin{align}
\label{eq:lik_ais}
p(R_{t} | H_{t}^s,F_{t}^s) &= \prod_{i=1}^p p(r_{it} | h_{it}^s,F_{t}^s)\nonumber \\& = \prod_{i=1}^p\sum_{n_{it}=0}^{\infty} \mathcal{N}(r_{it}|n_{it}\mu_{i\xi},\exp(h^s_{it})+n_{it}\sigma_{i\xi}^2)p(n_{it}|\Delta_{it} \lambda^s_{it}).
\end{align}
We also note that the infinite sum in \eqref{eq:lik_ais} can be truncated without affecting the results of the algorithm \citep{johannes2009optimal}. We refer to the supplementary material for the pseudocode that we used to construct the described AIS algorithm. 

The sample means of the weights calculated at each step of AIS provide estimators of the marginal likelihoods $p(R_{T+1:T+j}|R_{1:T})$. However, the resulting weights have large variance and the corresponding estimates will be inaccurate. Nevertheless, and in contrast with a particle filter that targets $p(H_{0:T+\ell},F_{0:T+\ell}|R_{1:T+\ell})$, we can use the samples $\{H^s_{0:T+j},F^s_{0:T+j}\}_{s=1}^S$ to estimate the quantities $p(r_{i,T+1:T+j}|R_{1:T})$ for each stock separately. We propose the following sampling procedure. First, note that  
\begin{align*}
\label{eq:marg_i}
p(r_{it}|R_{1:t-1}) &= \int p(r_{it}|h_{it},F_{t})p(h_{it}|h_{i,t-1})p(F_{t}|F_{t-1})\\ & \times p(h_{i,0:t-1},F_{0:t-1}|R_{1:t-1})\d{ h_{i,0:t}} \d{F_{0:t}}.
\end{align*}   
This implies that by using only the $i$th element of the product in \eqref{eq:lik_ais} we obtain an estimator of the $i$th marginal of $p(R_{T+1:T+j}|R_{1:T})$ as follows. We set
$\hat{p}(r_{i,T+1:T+j}|R_{1:T}) = \sum_{s=1}^S \omega^s_{i,T+j}/S $, where 
\begin{equation*}
\label{eq:ais_marg_wts}
\omega^s_{it}  = \omega^s_{i,t-1}p(r_{it}|h_{it}^s,F_{t}^s ),
\end{equation*}
and $\omega^s_{i,T+1} = 1$. Thus, $p(r_{i,T+1:T+j}|R_{1:T},\M_{\mathrm{mvj}})$ is estimated by $\hat{p}(r_{i,T+1:T+j}|R_{1:T})$ accurately since the latter is based on importance sampling weights with reduced variance and can be used for the estimation of log Bayes factors of the form
\begin{equation}
\label{eq:margmargratio}
\log(p(r_{i,T+1:T+j}|R_{1:T},\M_{\mathrm{mvj}}))-\log(p(r_{i,T+1:T+j}|R_{1:T},\M_{\mathrm{svj}})),
\end{equation}
where $p(r_{i,T+1:T+j}|R_{1:T},\M_{\mathrm{svj}})$ can be estimated with the methods presented in Section \ref{sec:compuni}.
We also note that by assuming that the joint predictive density of all returns is estimated by the product of their marginal predictive densities, as in the independent across stocks SV with jumps models, we estimate the approximate log Bayes factors
  \begin{equation}
\label{eq:margmargratio2}
\sum_{i=1}^p \log(p(r_{i,T+1:T+j}|R_{1:T},\M_{\mathrm{mvj}}))-   \sum_{i=1}^p \log(p(r_{i,T+1:T+j}|R_{1:T},\M_{\mathrm{svj}})).    
  \end{equation}
Finally, $\{h^s_{i,0:T+j},F^s_{0:T+j}\}_{s=1}^S$ can be utilized in order to draw samples from the posterior predictive distributions needed to calculate the continuous ranked probability and the interval scores; see the supplementary material for details.

\section{Real data application}
\label{sec:realdata}

Here we present results from the application of the developed methodology on real data. In the supplementary material we present results from simulation studies used to test our methodology.  Further specifics of the data analysis, such as parameter estimates for the log-volatility processes, can be also found in the supplementary material as well as in \cite{alexopoulos2017gaussian}.

\subsection{Data and MCMC details} 

We applied the developed methods on time series consisted of daily log-returns from stocks of the STOXX Europe $600$ Index downloaded from Bloomberg between $10/1/2007$ to $11/6/2014$. We removed $29$ stocks by requiring each stock to have at least $1000$ traded days and no more than $10$ consecutive days with unchanged price. The final dataset consisted of $p = 571$ stocks and $1947$ traded days. We used daily log-returns observed in the first $T =1917$ days as in-sample observations to estimate all models described through our article and the last $\ell=30$ observations to test their predictive ability. In our dynamic factor model we used $K=3$ factors. We ran all the MCMC algorithms for $250,000$ iterations using a burn-in period of $100,000$ sampled points and we collected $1,000$ (thinned) posterior samples. In the supplementary material we give the computing times for the real and simulated data analysis that we have implemented and we evaluate the accuracy of the sequential Monte Carlo methods that we used to conduct the assessment of the predictive performance of the different models.

\subsection{Results}
\label{sec:est_results}

\subsubsection{Bayesian inference for parameters and latent states}
\label{sec:pred_results}

Figure \ref{fig:factors} depicts the posterior means of the factor paths along with the posterior distributions of their persistent parameters $\alpha_1$, $\alpha_2$ and $\alpha_3$. The posterior densities of $\alpha_1$ and $\alpha_2$ indicate strong evidence of autocorrelation in the first two factor processes while that of $\alpha_3$ suggests that
a model with only $K=2$ factors adequately captures the dependence of the jumps over times and across stocks. Therefore, all results in the following are based on $K=2$ factors.

By noting that we have applied our methods on a very high-dimensional dataset we present posterior distributions for summary statistics related to the latent volatility and jump processes. Figure \ref{fig:vol_post} shows, for each stock, the $95\%$ credible interval of the temporal average of the volatility process. From the visual inspection of the Figure we note that the stocks can be separated in those with low, medium and high volatilities. More precisely, the posterior temporal average of the volatilities for the majority of the stocks is estimated lower than $5$, for a few number of stocks the temporal average ranges between $5$ and $15$ while for less than $10$ stocks the posterior temporal average has $95\%$ credible interval with lower limit greater than $15$. See also the supplementary material for additional results regarding Bayesian inference about the log-volatility processes and their parameters. Figure \ref{fig:jumps_sectors} presents the $95\%$ credible intervals for the posterior distribution of the temporal sum of jumps for each stock with respect to the economic sector that it belongs to. The Figure indicates that for each stock have been identified $100$ jumps on average between $10/1/2007$ and $11/6/2014$. We also note that the number of stocks for which almost no jumps have been estimated is very small and such stocks mainly belong to the financial, the industrial and the consumer- cyclical and non-cyclical, the energy and the communications sectors. In the supplementary material we present more results related to the Bayesian analysis of the jump processes of the SV models while Figure \ref{fig:compJumps} offers a visual illustration of the estimated jumps underlying the time evolution of the stock returns.

As already noted in the introduction of the paper the top panel of Figure \ref{fig:compJumps}, which has been created by using an empirical method for jump identification, visualizes a dense estimation of the underlying jump process. The middle panel which corresponds to independent modelling of jumps depicts a very sparse jump process whereas the bottom panel indicates less number of jumps than the empirical method but more than the independent SV models. To evaluate the estimation of the jump processes in the three panels we first emphasize that the true jumps in the stock prices are unobserved events. Therefore, drawing conclusions for the accuracy of the jumps estimation by only visually inspecting the Figure is difficult. We note, however, that the dense jump structure obtained from the empirical method can be explained by the fact that the time evolution of the volatility of the stock returns is not taken into account; resulting in a possible overestimation of the number of jumps.  The latter combined with the more narrow credible intervals for the log-volatilities delivered by SV models with jumps (see Figure S$.12$ in the supplementary material) indicates that an efficient jump identification results in accurate estimation of the log-volatility process without increasing unnecessarily its level due to the presence of large price movements and, thus, more precise forecasts of future returns are expected.  See also the following Section where we compare the predictive performance of SV models with and without jumps. To assess the jump structures estimated by employing independent and joint models for the jump intensities of SV models we rely on their forecasting performance. In particular, as also described below, the proposed joint modelling approach for the jump intensities results in a more accurate prediction of future returns. We conclude, therefore, that the estimation of the underlying jumps process depicted by the third panel of Figure \ref{fig:compJumps} is closer to the real jump activity of the stock prices allowing further improvement in the estimation of the log-volatilities and more accurate prediction of future stock returns.

\begin{figure}[H]
\begin{center}
\includegraphics[scale=0.35]{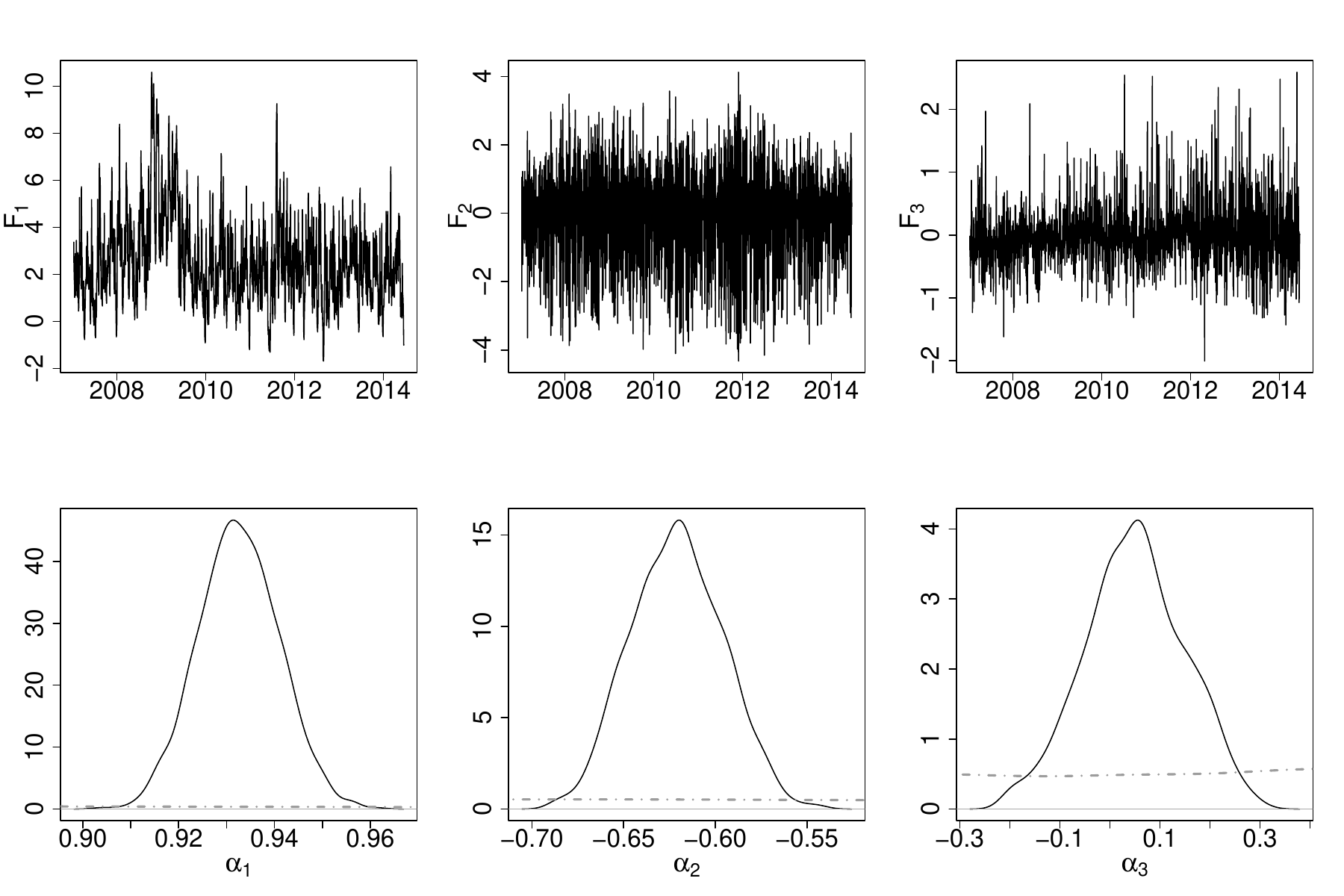}
\vspace*{-0.3cm}
\caption{First row: posterior mean of the paths of the three autoregressive factors $F_1$, $F_2$ and $F_3$ that we used to model the evolution of the jump intensities across stocks and across time. Second row: posterior (solid lines) and prior (dotted lines) distributions of the persistent parameters $\alpha_1$, $\alpha_2$ and $\alpha_3$.}
\label{fig:factors}
\end{center}
\end{figure}

\begin{figure}[H]
\begin{center}
\includegraphics[scale=0.45]{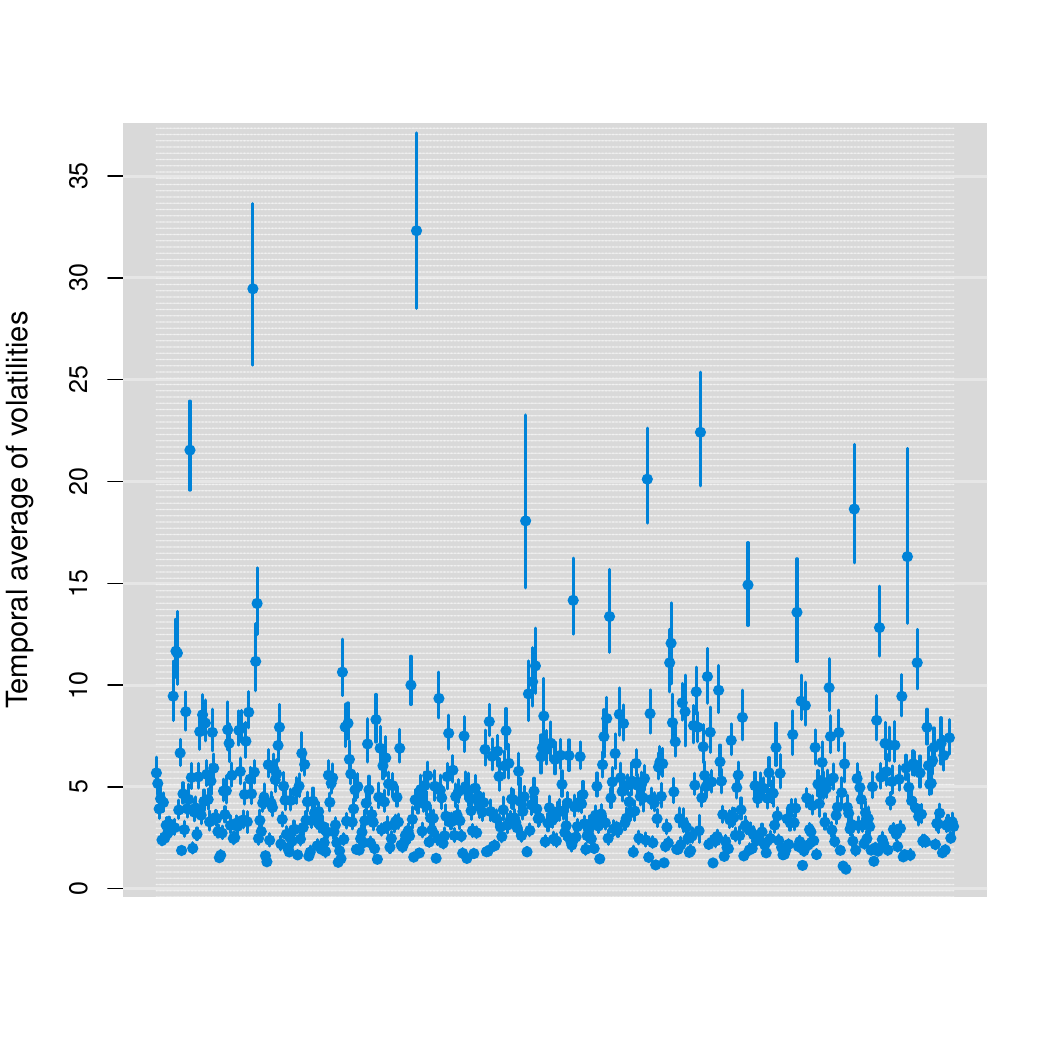}
\vspace*{-0.5cm}
\caption{$95\%$ credible intervals of the posterior distributions of the temporal averages of the volatility processes. Each dotted blue line corresponds to a different stock.}
\label{fig:vol_post}
\end{center}
\end{figure}

\begin{figure}[H]
\begin{center}
\includegraphics[scale=0.47]{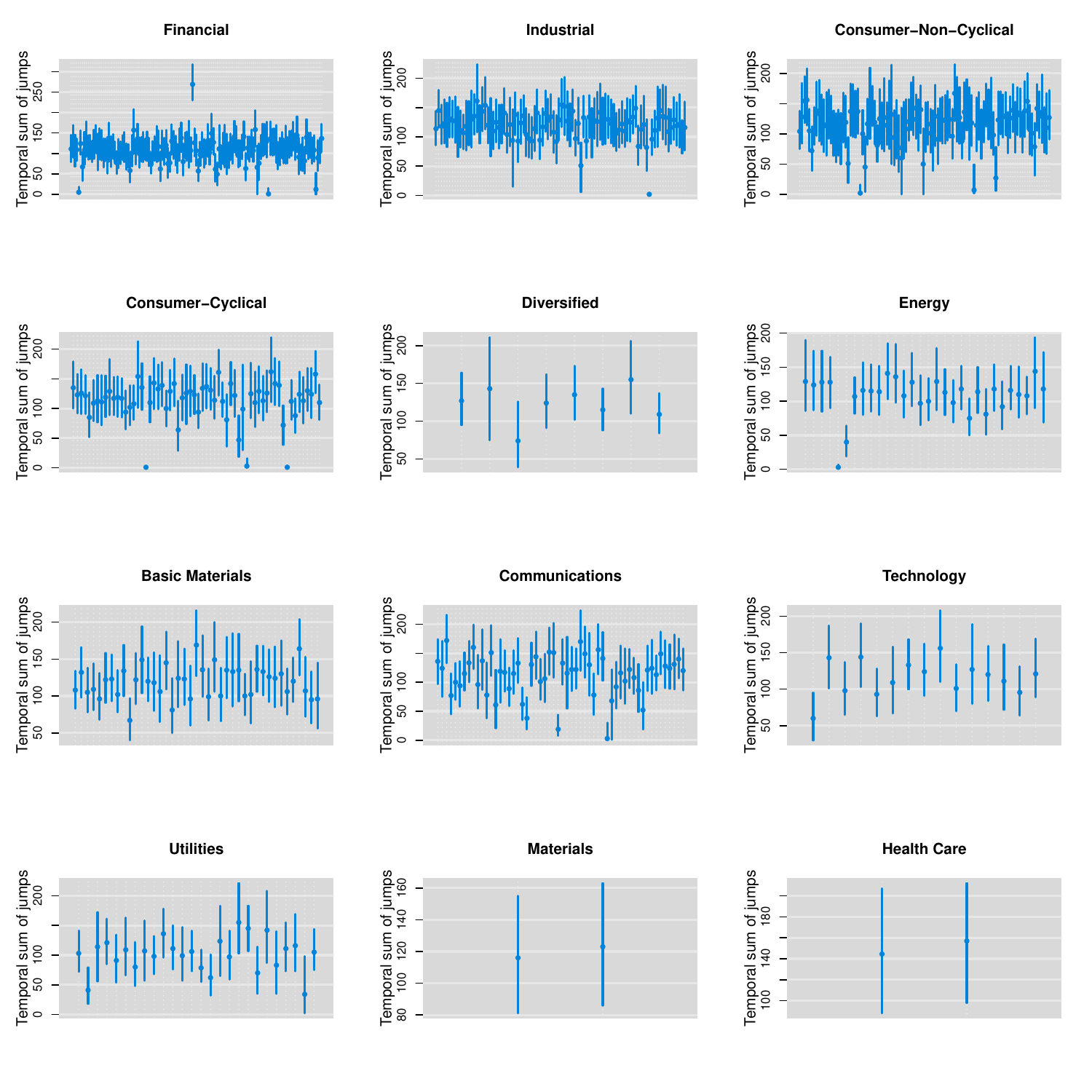}
\vspace*{-0.5cm}
\caption{$95\%$ credible intervals of the posterior distributions of the temporal sum of jumps. Each dotted blue line corresponds to a different stock and each plot is consisted of credible intervals for the stocks that belong to the economic sector indicated by its title.}
\label{fig:jumps_sectors}
\end{center}
\end{figure}

\subsubsection{Predictive performance}

Here we evaluate the predictive performance of the proposed SV model with jointly modelled jump intensities as well as the corresponding performance of univariate SV models with and without jumps which are also specified by equations \eqref{eq:obs}-\eqref{eq:poisson} by making suitable assumptions. In the supplementary material we compare the forecasting performance of the proposed model with alternative SV models which account for the heavy tails of the distribution of the observed returns.

The left panel of Figure \ref{fig:uni_mv_BF} compares the predictive performance of independent, univariate SV models with and without jumps by presenting predictive log Bayes factors. The Figure is constructed by considering the log Bayes factors defined in \eqref{eq:mvBFuni} for the $\ell = 30$ out-of-sample observations of our dataset. Thus, positive log Bayes factor indicates that the predictions obtained with univariate SV models with jumps are more accurate than those of SV models without jumps. According to \cite{kass1995bayes} a value of the log Bayes factor greater than $5$ provides strong evidence in favour of the model with jumps. The increasing nature of the log Bayes factor reassures that as more data are collected and as more jumps are identified in stock returns, the evidence that the model with jumps outperforms the model without jumps increases. In the right panel of Figure \ref{fig:uni_mv_BF} we compare the predictive performance of our proposed SV with jumps model in which the jump intensities are modelled by using a dynamic factor model with SV models in which the jump intensities are independent over time and across stocks. The Figure compare the two models by presenting the logarithms of the approximate log Bayes factors in \eqref{eq:margmargratio2}.  Clearly, there is strong evidence that our proposed model outperforms the independent modelling of stock returns with a SV with jumps formulation. 

\begin{figure}[t]
\begin{center}
\includegraphics[scale=0.4]{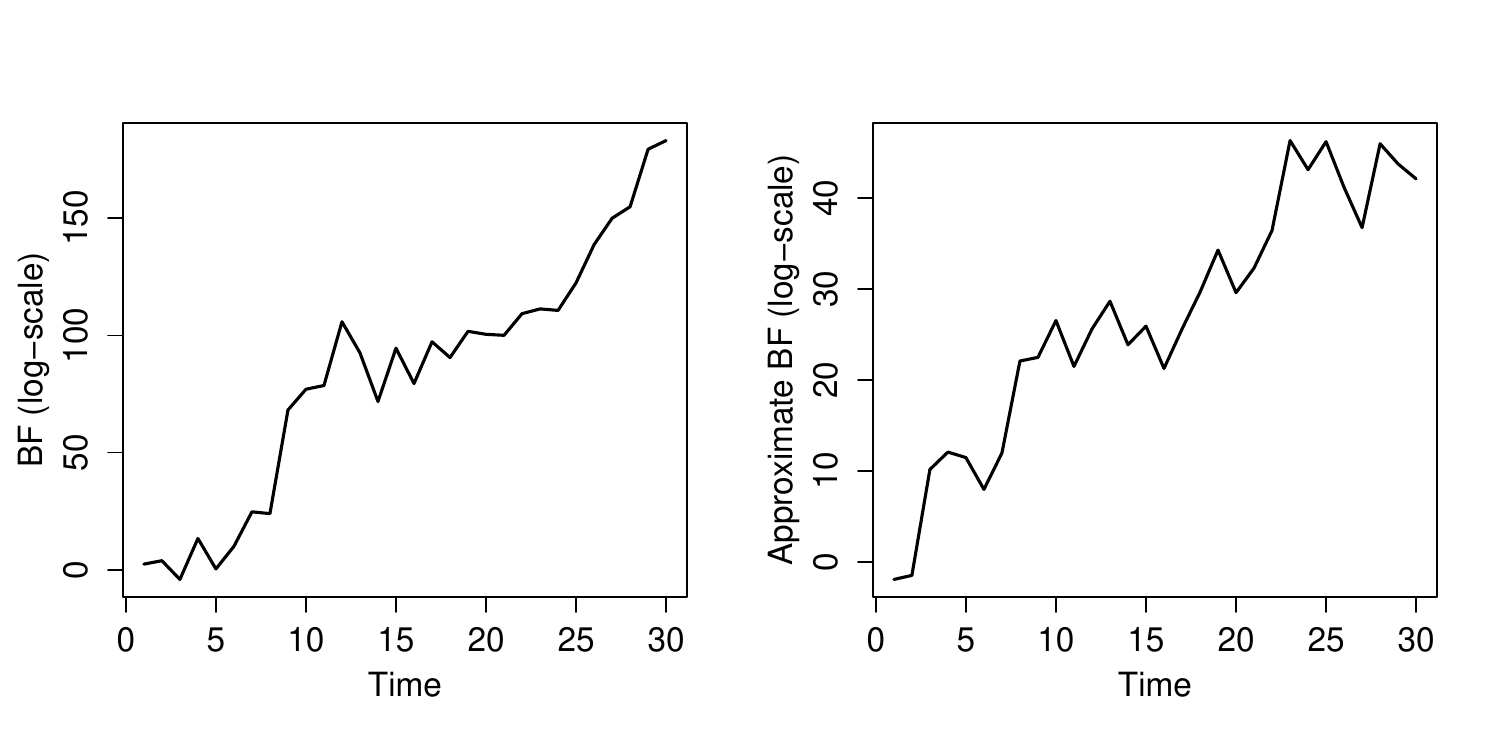}
\vspace*{-0.5cm}
\caption{Left: Log Bayes factors (BF) in favour of univariate SV models with jump intensities modelled independently across time and stocks against univariate SV without jumps models; see \eqref{eq:mvBFuni} for the corresponding formula. Right: Logarithm of approximate BF for SV models with jointly (numerator) and independently (denominator) modelled jumps intensities; see \eqref{eq:margmargratio2} for the corresponding formula. The x-axis in both plots refers to the $\ell=30$ out-of-sample time points.}
\label{fig:uni_mv_BF}
\end{center}
\end{figure}

Figure \ref{fig:CI_eval} compares the forecasts of the proposed model with those obtained from SV models without jumps and with independent jumps respectively by employing their probabilistic evaluation \citep{gneiting2007strictly}.  Joint modelling of the intensities leads to $95\%$ prediction credible intervals with the lowest interval scores at each one of the out-of-sample points for the vast majority of the stocks. The definition of the interval score in \eqref{eq:interval_score} rewards the forecaster that delivers narrow prediction intervals and when the observation is outside the predicted interval a penalty proportional to the significance level of the interval is incurred. It is thus indicated that the proposed SV model delivers the most accurate prediction intervals compared to those obtained from univariate SV models with and without jumps. The fact that Figure \ref{fig:CI_eval} displays the median instead of the mean of the $571$ interval scores emphasizes an important characteristic of our modelling approach. The efficient identification of the in-sample jumps has resulted in estimating lower volatilities for the stock returns; see Figures S$.11$ and S$.12$ in the supplementary material. This is depicted in the conducted out-of-sample exercise by more narrow prediction intervals for a given coverage probability. Moreover, since the prediction of future jumps is based on common across the stocks, autoregressive factors, we expect that jumps that occur suddenly in a small number of stocks at a given day are not easily predictable. On the other hand, the large volatility intervals delivered from SV models without jumps or models with less accurate jump identification can accidentally predict more efficiently these type of movements. The median is chosen to avoid taking into account such predictions. The superiority of the proposed approach is also emphasized by Figure \ref{fig:CI_eval} by the lower continuous ranked probability scores that correspond to the proposed SV model than those of the univariate SV models. This implies that the posterior predictive distribution of the proposed model are more concentrated than those obtained from univariate SV models with and without jumps while deliver predictions that are not easily distinguishable from the materialized observations. This is a desired property of probabilistic forecasts described as the main goal of probabilistic forecasting; see for example \cite{gneiting2014probabilistic}.
Finally, Figure S$.14$ in the supplementary material which presents root mean squared errors for the predictions obtained from the different models indicates that joint modelling of the jump intensities improves slightly the point forecasts. We conclude that the proposed modelling approach offers a more accurate identification of jumps and consequently a more precise estimation of the returns volatility compared to the one conducted by using SV models without or with independent jumps. As a result, the predictive performance of the proposed model outperforms the other models. This makes the suggested model a competitive alternative for practitioners who are interested in obtaining accurate predictions for the stocks of large portfolios.

\section{Discussion}
\label{sec:discuss}

We have developed a general modelling framework together with carefully designed MCMC algorithms to perform Bayesian inference for SV models with Poisson-driven jumps.  We have shown that for the data we applied our models there is evidence that, with respect to predictive Bayes factors,  (i) univariate SV models with jumps outperform univariate SV models without jumps and (ii) models that jointly model jump intensities with a dynamic factor model outperform SV models with jumps that are applied independently across stocks.  We feel that (ii) is an interesting result that adds considerable insight to the growing literature of modelling financial returns  with jumps, adding to the observation by \cite{ait2015modeling} that there is indeed predictability in jump intensities.

There are various issues that have not been addressed in this paper. The proposed MCMC algorithms can be easily extended in order to conduct inference for the more realistic \citep{harvey1996estimation} class of SV models which allow for correlation between the error terms $\epsilon_{it}$ and $\eta_{it}$ in \eqref{eq:obs} and \eqref{eq:state}, a phenomenon which is commonly known as leverage effect. In this case 
samples from their full conditional distributions can be obtained by using the auxiliary gradient-based sampler \citep{titsias2018auxiliary} to draw each correlation parameter jointly with $\phi_i$, $\sigma_{i\eta}^2$ and $H_i$.
Other modelling aspects include the relaxing of independence of $E_t$ in (\ref{eq:facdyn3}) or the assumption that $A$ is lower diagonal as in the models proposed in \cite{dellaportas2012cholesky}.  

Finally, a series of interesting financial questions can be addressed with our models by exploiting the fact that panel of stock returns carry additional, possibly important, information.  For example, in our dataset one can explore the effects of country and sector effects or  could investigate whether the joint tail risk dependence does or does not change before, during and after the financial crisis in Europe. We feel that our proposed models together with our algorithmic guidelines will serve as useful tools for such future research pursuits in financial literature.

\begin{figure}[H]
\begin{center}
\includegraphics[scale=0.4]{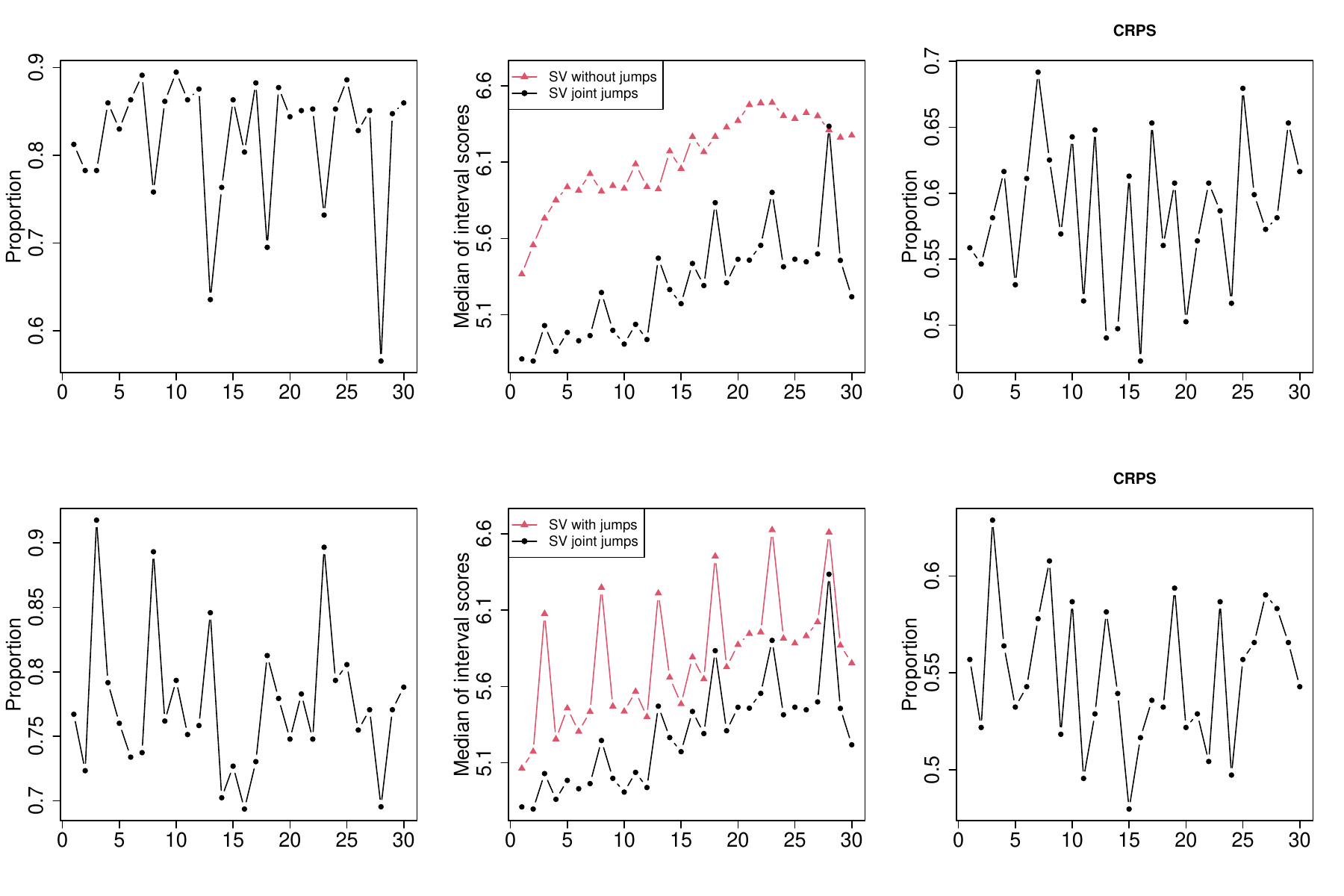}
\vspace*{-0.5cm}
       \caption{First column: proportion of stocks for which the interval scores for the $95\%$ prediction intervals of the SV model with jointly modelled jumps is lower than the corresponding score for the SV models without (top) and with independent (bottom) jumps. Second column: medians of interval scores across the $571$ stocks. Third column: proportion of stocks for which the continuous ranked probability score (CRPS) for the SV model with jointly modelled jumps is lower than the the corresponding score for the SV models without (top) and with independent (bottom) jumps. The x-axis in all plots refers to the $\ell=30$ out-of-sample time points.}
\label{fig:CI_eval}
\end{center}
\end{figure}

\bibliographystyle{Chicago}
\bibliography{refs}

\newpage

\begin{center}
{\LARGE Supplementary material}
\end{center}

\appendix

\section*{Empirical detection of jumps}

Let $R_i=(r_{i1},\ldots,r_{iT})$ be the vector with log-returns of the $i$th stock. The detection of jumps is based on the statistic 
\begin{equation}
\label{eq:mad}
\dfrac{| r_{it} -median (R_i)  | }{MAD}
\end{equation}
where
$
MAD = 1.48 \times median (| r_{it} -median (R_i)  | )
$
is the Mean Absolute Deviation from the median and it is a robust estimator for the scale parameter of a Gaussian distribution \citep{rousseeuw1993alternatives}. A log-return is considered a jump if for this log-return the statistic in \eqref{eq:mad} has value greater than $3$; see \cite{rousseeuw1993alternatives} for a detailed discussion about the use of the MAD estimator in outlier detection problems.

\section*{Posterior jump probabilities}

Figure $1$ in the main paper was constructed by considering a price movement as jump if the jump probability for the given stock at that particular day was estimated greater than $0.5$. Figure \ref{jump_probs} indicates that the majority of the posterior jump probabilities have been estimated to be close to zero. In particular, $98\%$ of them are lower than $0.5$ and from those $95\%$ is less than $0.2$.

\begin{figure}[H]
	\centerline{\includegraphics[scale=0.55]{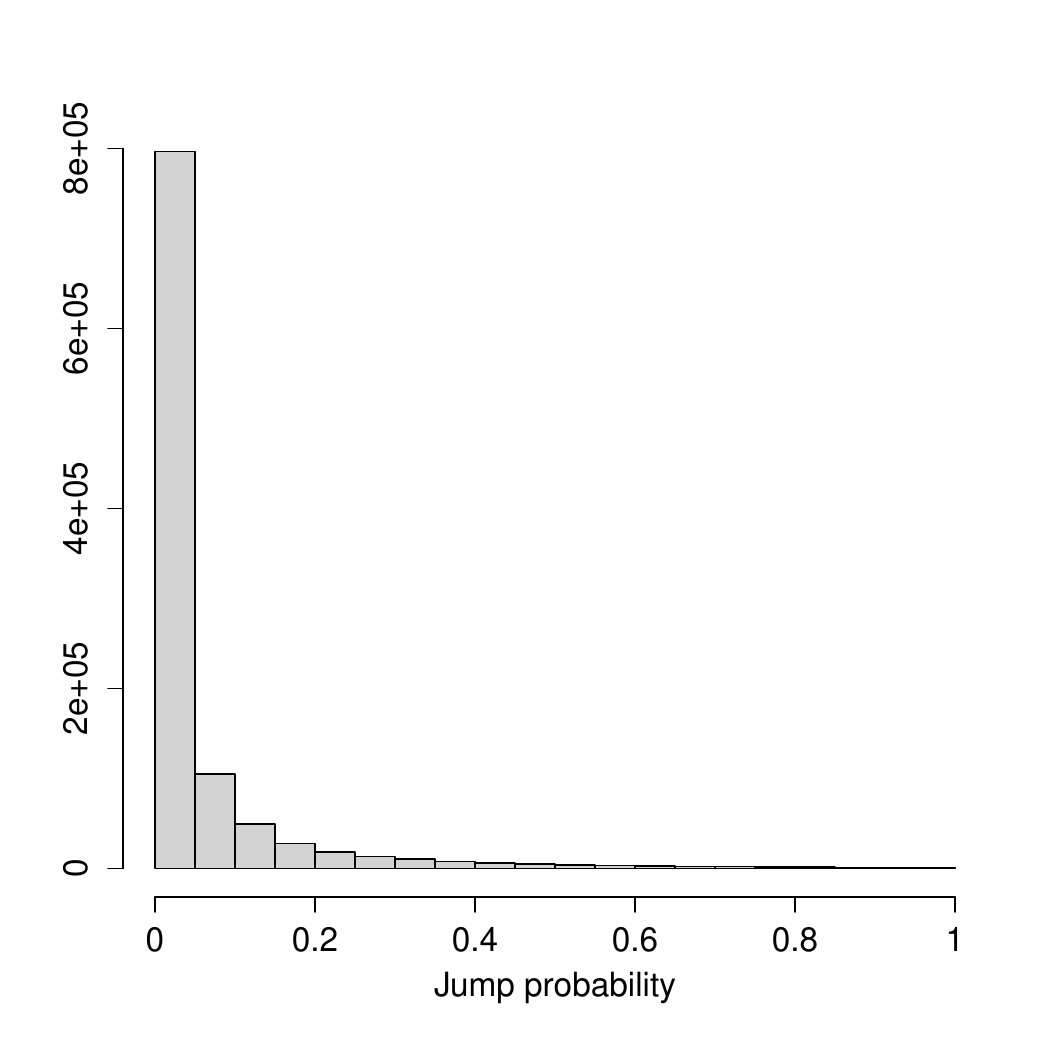}}
	\caption{Histogram of posterior probability of at least one jump for each stock and each day.}
	\label{jump_probs}
\end{figure}

\section*{Choosing hyperparameters for the factor model}
We have assumed that $b_i \sm \mathcal{N}(\mu_b,\sigma^2_b)$, $w_{ik} \sm \mathcal{N}(0,\sigma_w^2)$ and $\alpha_k \sm U(-1,1)$, $i=1,\ldots,p$ and $k=1,\ldots,K$. Our aim is to specify the hyperparameters $\mu_b$, $\sigma_b^2$ and $\sigma_w^2$ to be such that the resulting prior on $\lambda_{it}$ has mode, mean and variance close to the corresponding values of the $Gam(1,50)$ distribution which was used as prior in the case of modelling the jump intensities independently over time and across stocks. 

From equations $(3.5)$ and $(3.6)$ of the main paper we have that 
\begin{equation}\label{eq:new_intens_trans}
\quad y_{it}= \log \big( \lambda_{it}/(\lambda^*-\lambda_{it}) \big) \sim \mathcal{N}(\mu_y, \sigma^2_y),\,\,\, \lambda_{it} \in (0,\lambda^*),
\end{equation}
where we have set, see Section $3.3$, $\lambda^* = 0.15$. To specify $\mu_y$ and $\sigma_y^2$ we note that equations $(3.6)$--$(3.8)$ in the main paper imply 
$$
\E(y_{it}|\mu_b)=\mu_b\,\,\, \text{and}\,\,\,  \Var(y_{it}|A,\sigma^2_w,\sigma^2_b) = \sigma_{w}^2\sum \limits_{k=1}^K (1-\alpha_k^2)^{-1} + \sigma^2_b.
$$
Thus, by assuming that $\alpha_k = \E(\alpha_k) =0$ we have
\begin{equation}\label{eq:new_intens_moments}
\mu_y= \mu_b\,\,\, \text{and}\,\,\, \sigma^2_y = K\sigma^2_w + \sigma^2_b.
\end{equation}
From \eqref{eq:new_intens_trans} we have that the density of $\lambda_{it}$ is 
\begin{equation}\label{eq:new_density}
\dfrac{\lambda^*\exp \left\lbrace -\frac{1}{2\sigma_y^2}\Big( \log( \frac{\lambda_{it}}{\lambda^*-\lambda_{it}} ) -\mu_y  \Big)^2 \right\rbrace  }{\sigma_y\lambda_{it}(\lambda^*-\lambda_{it})\sqrt{2\pi}},\,\,\, \lambda_{it} \in (0,\lambda^*).
\end{equation}
Therefore, the location of its mode satisfies the equation 
\begin{equation}\label{eq:mode}
\log( \frac{\lambda_{it}}{\lambda^*-\lambda_{it}} ) = \frac{\lambda^*\mu_y+\sigma_y^2(2\lambda_{it}-\lambda^*)}{\lambda^*},  
\end{equation}
which can be solved, for given values of $\lambda^*$, $\mu_y$ and $\sigma^2_y$, with numerical methods. Here we use the R-package \texttt{rootSolve} \citep{rootSolve} to find all the roots of the equation in \eqref{eq:mode}.

Since our aim is to match $\mu_y$, $\sigma^2_y$ and the mode of the density in \eqref{eq:new_density} with the corresponding quantities implied from the $Gam(1,50)$ distribution we first note that if $\lambda_{it} \sim Gam(1,50)$, right truncated at $\lambda^*$, then for $\lambda^* =0.15$ we have that 
\begin{equation}\label{eq:uni_moments}
\E\big[ \log \big( \lambda_{it}/(\lambda^*-\lambda_{it}) \big) \big] \approx -2.4\,\, \text{and}\,\, \Var\big[ \log \big( \lambda_{it}/(\lambda^*-\lambda_{it}) \big) \big] \approx 2,
 \end{equation}
and that the mode of the $Gam(1,50)$ distribution is located at zero. By noting that in our applications we have used models with $K=2$ and $K=3$ latent factors and taking into account \eqref{eq:new_intens_moments} we set $\sigma_w^2 =0.5$ and $\sigma^2_b=1$. Then, from \eqref{eq:new_intens_moments} we have that $\sigma^2_y$ is close to the variance that we obtain for $y_{it}$ in the univariate case; see equation \eqref{eq:uni_moments}. To specify $\mu_b$ except from matching the means in \eqref{eq:new_intens_moments} and \eqref{eq:uni_moments} we are also trying to set the mode of the density in \eqref{eq:new_density} close to zero. Table \ref{tab:intens_priors_supp} displays the mean, the variance and the mode of \eqref{eq:new_density} for three different values of $\mu_b$. Note that by choosing $\mu_b =-2.4$ we have that the mean and the variance of the density in \eqref{eq:new_density} are quite close to the corresponding moments of the $Gam(1,50)$ distribution while by taking smaller values for $\mu_b$ the mode in \eqref{eq:new_density} becomes closer to zero.

\begin{table}[H]
\centering
\begin{tabular}{|c c c c|}
\hline
 & Mean & Variance & Mode\\
\hline
$\mu_b=-2.4$ & $0.021$ & $0.00054$ & $0.002$\\
$\mu_b=-5$ &  $0.003$  & $0.00004$ & $0.0001$\\
$\mu_b=-10$ &  $2 \times 10^{-5} $  & $4 \times 10^{-8}$ & $10^{-6}$\\
\hline
\end{tabular}
\caption{Mean, variance and mode of the density in \eqref{eq:new_density} for different values of $\mu_b$.}
\label{tab:intens_priors_supp}
\end{table}

\section*{Proof of Proposition 1}

\begin{proof}
A discrete distribution with pmf $p(n)$ is called log-concave when the inequality
$p(n)^2 \geq p(n-1)p(n+1)$
is satisfied for all $n$. 
Let $\sigma^2_{it} = \exp (h_{it})$, in the case of joint modelled jump intensities we have that $\lambda_{it} = \lambda^*/\big(1+\exp(-b_i-W'_iF_t)\big)$ and in the case of independent modelling $\lambda_{it}$ is a constant.
We denote with $p(n)$ the pmf $p(n_{it}=n|h_{it},\mu_{i\xi},\sigma^2_{i\xi},b_i,W_i,F_t,r_{it})$. By integrating out the jump sizes $\Xi_i$ from the likelihood we have that
\begin{align*}
p(n) & \propto p(r_{it}| n,h_{it},\mu_{i\xi},\sigma^2_{i\xi})p(n | \lambda_{it}\Delta_{it})\\
& =(2\pi(\sigma^2_{it}+n\sigma^2_{i\xi}))^{-1/2}\exp\left( -\frac{1}{2}\frac{(r_{it}-n\mu_{i\xi})^2}{\sigma^2_{it}+n\sigma^2_{i\xi}}-\lambda_{it} \Delta_{it} \right)\frac{(\lambda_{it} \Delta_{it})^{n}}{n!},\,\,\,  n=0,1,2,...
\end{align*}
To prove log-concavity for $n=1,2,...$, we need to show that for $n=2,3,...$
$$
2\log(p(n))-\log(p(n-1))-\log(p(n+1)) \geq 0
$$
or equivalently that,
\begin{multline*}
\log\left((\sigma_{it}^2+n\sigma^2_{i\xi})^2-\sigma^4_{i\xi}\right)-\log(\sigma^2_{it}+n\sigma^2_{i\xi})^2+\frac{2\left(\sigma^2_{i\xi}(r_{it}-n\mu_{i\xi})^2+\mu_{i\xi}(\sigma^2_{it}+n\sigma^2_{i\xi})\right)^2}{(\sigma^2_{it}+n\sigma^2_{i\xi})\left((\sigma^2_{it}+n\sigma^2_{i\xi})^2-\sigma^4_{i\xi}\right)} \\+\log(n+1)^2-\log(n)^2\geq 0.
\end{multline*}
or that
$$
\frac{((\sigma_{it}^2+n\sigma^2_{i\xi})^2-\sigma^4_{i\xi})(n+1)^2}{n^2(\sigma^2_{it}+n\sigma^2_{i\xi})^2}e^L \geq 1
$$
where $L=\frac{2\left(\sigma^2_{i\xi}(r_{it}-n\mu_{i\xi})^2+\mu_{i\xi}(\sigma^2_{it}+n\sigma^2_{i\xi})\right)^2}{(\sigma^2_{it}+n\sigma^2_{i\xi})\left((\sigma^2_{it}+n\sigma^2_{i\xi})^2-\sigma^4_{i\xi}\right)}$.
The last inequality holds; $L>0$ and it is readily shown that the numerator is greater than the denominator
for all $ n=2,3...$.
\end{proof}

\section*{Proof of Proposition 2}

\begin{proof}
Let $\one{n}'$ be an $n$-dimensional column vector with ones and $I_{n}$ the $n \times n $ identity matrix. 
For each $i=1,\ldots,p$ and $t=1,\ldots,T$ we have that $\xi_{it}^{\kappa} \sm N(\mu_{i\xi},\sigma^2_{i\xi})$, $\kappa=1,\ldots,n_{it}$.
Let $\sigma^2_{it}=\exp ( h_{it} )$ then from the canonical form of the Gaussian distribution we have that
\begin{align*}\label{eq:xiprod}
p(\xi_{it}| n_{it}, r_{it},h_{it},\mu_{i\xi},\sigma^2_{i\xi}) & \propto \mathcal{N}\left(r_{it}\bigg|\sum_{\kappa=1}^{n_{it}} \xi_{it}^{\kappa},\sigma_{it}^2\right)
\prod_{\kappa=1}^{n_{it}}\mathcal{N}\big(\xi^{\kappa}_{it}|\mu_{i\xi}, \sigma^2_{i\xi}\big) \\
& \propto \mathcal{N}_{n_{it}}\big(\xi_{it}|Q_{i\xi}^{-1}c, Q_{i\xi}^{-1} \big)
\end{align*}
where $c= (\frac{\mu_{i\xi}}{\sigma^2_{i\xi}}+\frac{r_{it}}{\sigma_{it}^2})\one{n_{it}}$ and
$Q_{i\xi} = Q_1 + Q_2$, with
$Q_1 =\sigma^{-2}_{i\xi}I_{n_{it}} $ and $Q_2 = \sigma_{it}^{-2}\one{n_{it}}\one{n_{it}}'$.
From \cite{miller1981inverse} we have that
\begin{equation*}
\label{eq:miller}
Q^{-1}_{i\xi}=(Q_1+Q_2)^{-1} = Q_1^{-1} - \frac{1}{1+\nu}Q_1^{-1}Q_2Q_1^{-1}
\end{equation*}
where $\nu = tr(Q_2Q_1^{-1})$ or
$(1+\nu)^{-1}= \dfrac{\sigma^2_{it}}{\sigma^2_{it}+n_{it}\sigma^2_{i\xi}}$. Therefore
$$
Q_{i\xi}^{-1}=\sigma^2_{i\xi}I_{n_{it}}- \frac{\sigma^4_{i\xi}}{\sigma_{it}^{2}+n_{it}\sigma^2_{i\xi}}\one{n_{it}}\one{n_{it}}'
\,\,\,
\text{and}
\,\,\,
Q_{i\xi}^{-1}c = (\dfrac{\mu_{i\xi}}{\sigma^2_{i\xi}}+\dfrac{r_{it}}{\sigma^2_{it}})\dfrac{\sigma^2_{i\xi}\sigma^2_{it}}{\sigma^2_{it}+n_{it}\sigma^2_{i\xi}}\one{n_{it}}.
$$
\end{proof}

\section*{Auxiliary gradient-based sampler}

Here we describe how we apply the auxiliary gradient-based sampler \citep{titsias2018auxiliary} to draw, at each MCMC iteration, the whole latent paths of the log-volatilities $H_i$ jointly with the parameters $\phi_i$ and $\sigma_{i\eta}^2$ and the whole path of the latent factors $F$ jointly with the parameters $A$.

\subsection*{Sampling latent log-volatilities and parameters}

To sample, for a fixed $i$, $H_i$ jointly with $\phi_i$ and $\sigma_{i\eta}^2$, we need to draw samples from the distribution with density 
\begin{equation}\label{eq:logvoltarget}
p(H_i,\phi_i,\sigma_{i\eta}^2 |\mu_i,\mu_{i\xi},\sigma_{i\xi}^2,N_i,R_i ) \propto \pi(H_i|\phi_i,\sigma_{i\eta}^2,\mu_i)\pi(\phi_i)\pi(\sigma_{i\eta}^2)\prod_{t=1}^T \mathcal{N}(r_{it}|n_{it}\mu_{i\xi},n_{it}\sigma^2_{i\xi}+e^{h_{it}})
\end{equation}
where the the jump sizes $\Xi_i$ have been integrated out, $\pi(H_i|\phi_i,\sigma_{i\eta}^2,\mu_i)$ denotes the prior of the latent autoregressive log-volatility process defined in $(3.2)$ and $\pi(\phi_i)$ and $\pi(\sigma_{i\eta}^2)$ denote the densities of the priors for the parameters $\phi_i$ and $\sigma^2_{i\eta}$ as described in Section $3.2$. To draw samples from \eqref{eq:logvoltarget} we employ Algorithm $3$ presented in Section $4.2$ of the main paper where $\mathrm{x}$ is the latent volatility vector $H_i$ and $\theta$ is consisted of the parameters $\phi_i$ and $\sigma^2_{i\eta}$; it is also easy to see that $C_{\theta}^{-1}$ has a tridiagonal form and both sampling from the proposal and calculating the acceptance ratio in $(3.6)$ can be conducted, efficiently, based on the Cholesky decomposition of tridiagonal matrices. Finally, we note that although Algorithm $3$ corresponds to a latent Gaussian model in which the Gaussian field has zero mean, each log-volatility vector $H_i$ has mean $\mu_i$; we take this into account and we present below how the steps of Algorithm $3$ of the main paper are adjusted in order to draw the desires samples $(H_i,\phi_i,\sigma^2_{i\eta})$. In particular, we iterate through the following steps:

\begin{enumerate}
\item[(i)] Sample $Z_i \sim \mathcal{N}_{T+1}\big(Z_i|H_i+(\delta /2)\nabla g(H_i|R_i,N_i,\mu_{i\xi}, \sigma^2_{i\xi}), (\delta/2)I_{T+1}\big)$
\item[(ii)] Sample from $p(H_i,\phi_i,\sigma_{i\eta}^2|Z_i,\mu_i,\mu_{i\xi},\sigma^2_{i\xi},N_i,R_i)$ by using a Metropolis-Hastings 
in which we transform $\hat{\phi}_i = \log\big( (\phi_i+1)/(1-\phi_i)  \big)$ and $\hat{\sigma}_{i\eta}^2 = 2\log \sigma_{i\eta}$ and
we propose new values $(H_i^*, \hat{\phi}_i^*,\hat{\sigma}_{i\eta}^{2*})$ 
by first drawing $(\hat{\phi}_i^*,\hat{\sigma}_{i\eta}^{2*})$ from the bivariate Gaussian density centred at $(\hat{\phi}_i,\hat{\sigma}_{i\eta}^{2})$ with covariance matrix $\kappa_iI_2$ and then draw $H_i^*$ from
\begin{equation}\label{eq:tp_prop}
q\big( H_i^*\big|\phi_i^*,\sigma_{i\eta}^{2*},Z_i\big) = \mathcal{N}\bigg(H_i^*\big|\frac{2}{\delta}Q( Z_i + \frac{\delta}{2}M^*_i ),Q\bigg) =\dfrac{\mathcal{N}(H_i^* \big|\mu_i,C_{\theta}^*)\mathcal{N}(Z_i \big| H_i^*,(\delta /2)I_{T+1} )}{\mathcal{Z}(Z_i,\mu_i,\phi_i^*,\sigma_{i\eta}^{2*})} ,
\end{equation}
with $Q = ( C_{\theta}^{*-1} + \frac{2}{\delta}I_{T+1}  )^{-1}$, $M^*_i = \dfrac{\mu_i(1-\phi^*_i)}{\sigma_{i\eta}^{2*}} \big(1,\,\,1-\phi^*_i,\ldots,1-\phi^*_i,\,\,1 \big)$ and 
$$
\mathcal{Z}(Z_i,\mu_i,\phi_i^*,\sigma_{i\eta}^{2*}) = \mathcal{N}_{T+1}(Z_i|\mu_i, C_{\theta}^* + (\delta/2)I_{T+1}  ),
$$
where $\phi^*_i = (e^{\hat{\phi}_i^*}-1)/(e^{\hat{\phi}_i^*}+1)$ and $\sigma_{i\eta}^{2*} = e^{\hat{\sigma}_{i\eta}^{2*}}$. The values $(H_i^*, \hat{\phi}_i^*,\hat{\sigma}_{i\eta}^{2*})$ are accepted or rejected according to the ratio
\begin{multline}
\exp \big\{ g\big(H_i^*|N_i,\mu_{i\xi},\sigma^2_{i\xi}\big)-g\big(H_i|R_i,N_i,\mu_{i\xi},\sigma^2_{i\xi}\big) + \nu\big(Z_i,H_i^*\big) - \nu\big(Z_i,H_i\big)  \big\} \\ \times
\dfrac{\mathcal{Z}(Z_i,\mu_i,\phi_i^*,\sigma_{i\eta}^{2*}\big)\pi \big(\phi^*_i\big)\pi \big(\sigma_{i\eta}^{2*}\big)J(\hat{\phi}^*_i,\hat{\sigma}_{i\eta}^{2*})}{\mathcal{Z}\big(Z_i,\mu_i,\phi_i,\sigma_{i\eta}^{2}\big)\pi\big(\phi_i\big)\pi \big(\sigma_{i\eta}^2\big)J(\hat{\phi}_i,\hat{\sigma}_{i\eta}^{2})}
\end{multline}
where $\nu\big(Z_i,H_i \big) = \big( Z_i-H_i - \frac{\delta}{4} \nabla g\big(H_i\big|R_i,N_i,\mu_{i\xi},\sigma^2_{i\xi}\big)  \big)'\nabla g\big(H_i\big|N_i,\mu_{i\xi},\sigma^2_{i\xi}\big)$ and $J(\cdot,\cdot)$ denotes the Jacobian of the inverse of the transformation $(\phi_i,\sigma_{i\eta}^2) \rightarrow (\hat{\phi}_i,\hat{\sigma}_{i\eta}^{2})$.
\end{enumerate}
After drawing $(H_i,\phi_i,\sigma_{i\eta}^2)$ we sample $\mu_i$ from its full conditional 
$\mathcal{N}( m_i, s_i^2 )$, where 
$$m_i =\frac{s_i^2\big[(1-\phi_i^2)h_{i0}+(1-\phi_i)\sum_{t=1}^T(h_{it}-\phi_ih_{i,t-1})\big]}{\sigma_{i\eta}^2}$$ and
$$s_i^2 = \frac{\sigma_{i\eta}^2}{\sigma_{i\eta}^2/10+ (1-\phi_i^2) + T(1-\phi_i)^2}.
$$

\subsubsection*{Ancillarity-sufficiency interweaving strategy for log-volatility parameters}

In order to implement the ancillarity-sufficiency interweaving strategy \citep{yu2011center} for the parameters $\mu_i$, $\phi_i$ and $\sigma_{i\eta}^2$ we first consider the model defined, in the main paper, by equations $(3.1)$ and $(3.2)$ with the non-centred parametrization of the log-volatilities. This is specified, for each $i=1,\ldots,p$ and $t=1,\ldots,T$, by 
\begin{equation*}
r_{it} = \exp(\frac{\mu_i+\sigma_{i\eta}\tilde{h}_{it}}{2})\epsilon_{it} + \sum_{\kappa=1}^{n_{it}} \xi_{it}^\kappa, \,\,\,\, \epsilon_{it} \sim N(0,1)
\end{equation*}
\begin{equation*}
\tilde{h}_{it} = \phi_i \tilde{h}_{i,t-1} + \eta_{it}, \,\,\,\, \eta_{it} \sim N(0,1), \,\,\,\, \tilde{h}_{i0} \sim N(0,1/(1-\phi_i^2)),
\end{equation*}
where $\epsilon_{it}$ and $\eta_{it}$ are independent random variables and
$\tilde{h}_{it} = (h_{it}-\mu_i)/\sigma_{i\eta}$. Then, we utilize random walk Metropolis-Hastings steps to draw $\mu_i$ and $\sigma_{i\eta}^2$ from their full conditionals obtained from
$$
p(\mu_i,\sigma_{i\eta}^2|\mu_{i\xi},\sigma^2_{i\xi},\tilde{H}_i,R_i ) \propto \pi(\mu_i)\pi(\sigma_{i\eta}^2)\prod_{t=1}^T \mathcal{N}\big(r_{it}\big|n_{it}\mu_{i\xi},n_{it}\sigma^2_{i\xi}+\exp\{\mu_i+\sigma_{i\eta}\tilde{h}_{it}\}\big),
$$
and we sample $\phi_i$ from 
$$
p(\phi_i|\tilde{H_i}) \propto \pi(\phi_i)\mathcal{N}\big(\tilde{h}_{i0}|0,1/(1-\phi_i^2)\big)\prod_{t=2}^T\mathcal{N}\big(\tilde{h}_{it}| \phi_i\tilde{h}_{i,t-1},1\big),
$$
by using a Metropolis-Hastings step with proposal distribution
$$
\mathcal{N}\bigg( \dfrac{\sum_{t=1}^T\tilde{h}_{it}\tilde{h}_{i,t-1}}{\sum_{t=1}^T\tilde{h}_{i,t-1}},\dfrac{1}{\sum_{t=1}^T\tilde{h}_{i,t-1}} \bigg).
$$

\subsubsection*{Simulation experiments for samplers with and without interweaving}

Figures \ref{098} and \ref{090} evaluate the efficiency of the auxiliary gradient-based sampler in drawing the parameters of the latent log-volatility path of univariate SV models. We compare three sampling schemes; the auxiliary gradient-based sampler with interweaving, without interweaving, and a scheme in which instead of interweaving we iterate twice the auxiliary sampler. To evaluate the different sampling schemes we compare the effective sample size (ESS) of the samples drawn from the posterior distributions of interest. The ESS of $S$ samples drawn by using an MCMC algorithm can be estimated as $v^2S/\nu_0$ where $v^2$ is the sample variance of the posterior samples and $\nu_0$ is an estimation of the spectral density of the Markov chain at zero. We used the R-package \texttt{coda} \citep{rcoda} to estimate the ESS of the posterior samples obtained from the samplers under comparison. 

\begin{figure}[H]
	\centerline{\includegraphics[scale=0.45]{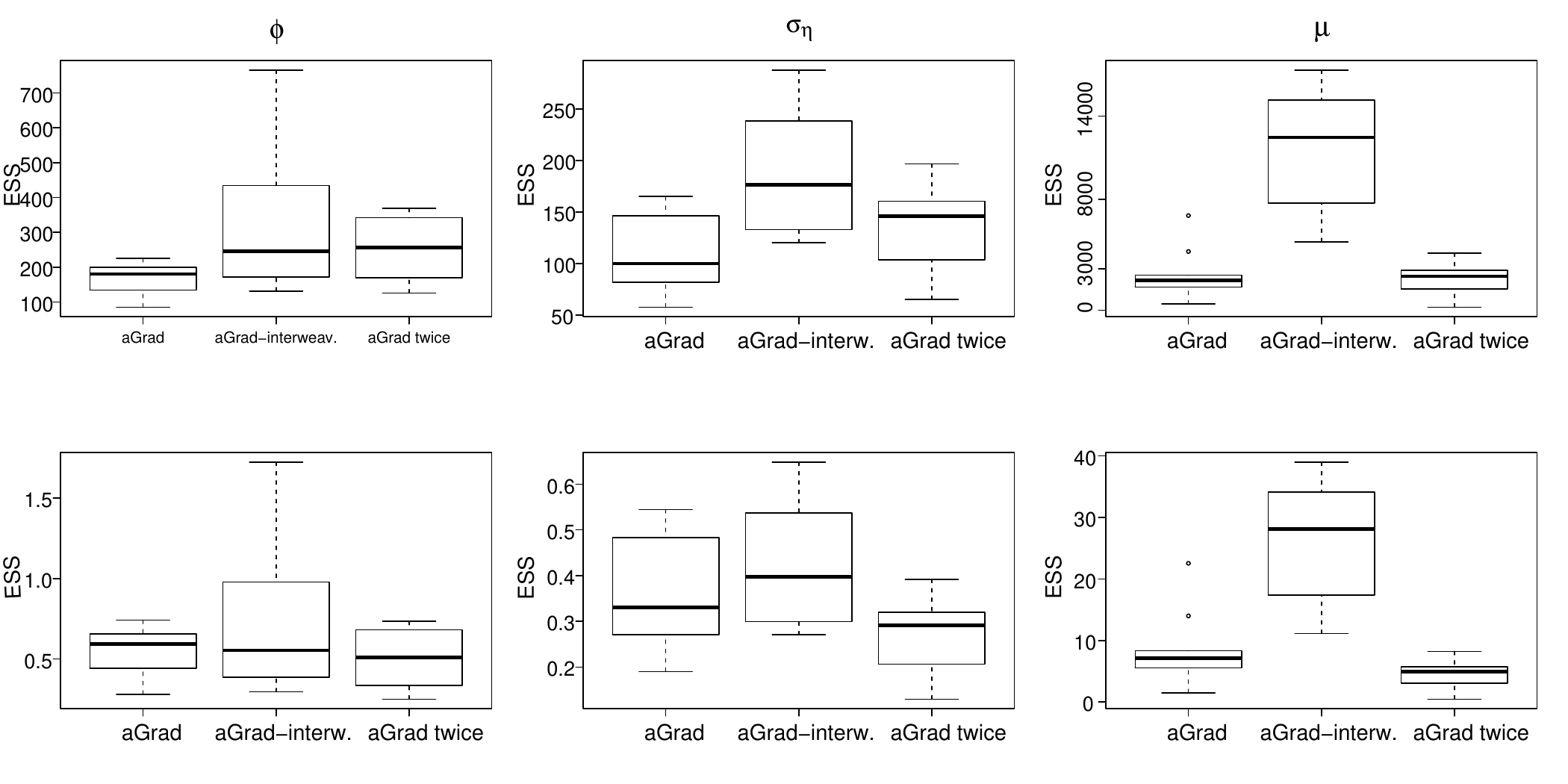}}
	\caption{Top: the effective sample size (ESS) of the posterior samples of the parameters $\mu$, $\sigma_{\eta}$ and $\phi$ obtained by applying the auxiliary gradient-based (aGrad) sampler with and without interweaving in $10$ replicates of a dataset consisted of $1,500$ log-returns simulated from the univariate SV model with $\mu=-0.85$, $\sigma_{\eta}=0.15$ and $\phi=0.98$. Bottom: the corresponding ESS divided by the number of seconds required for its sampling scheme.}
	\label{098}
\end{figure}

\begin{figure}[H]
	\centerline{\includegraphics[scale=0.45]{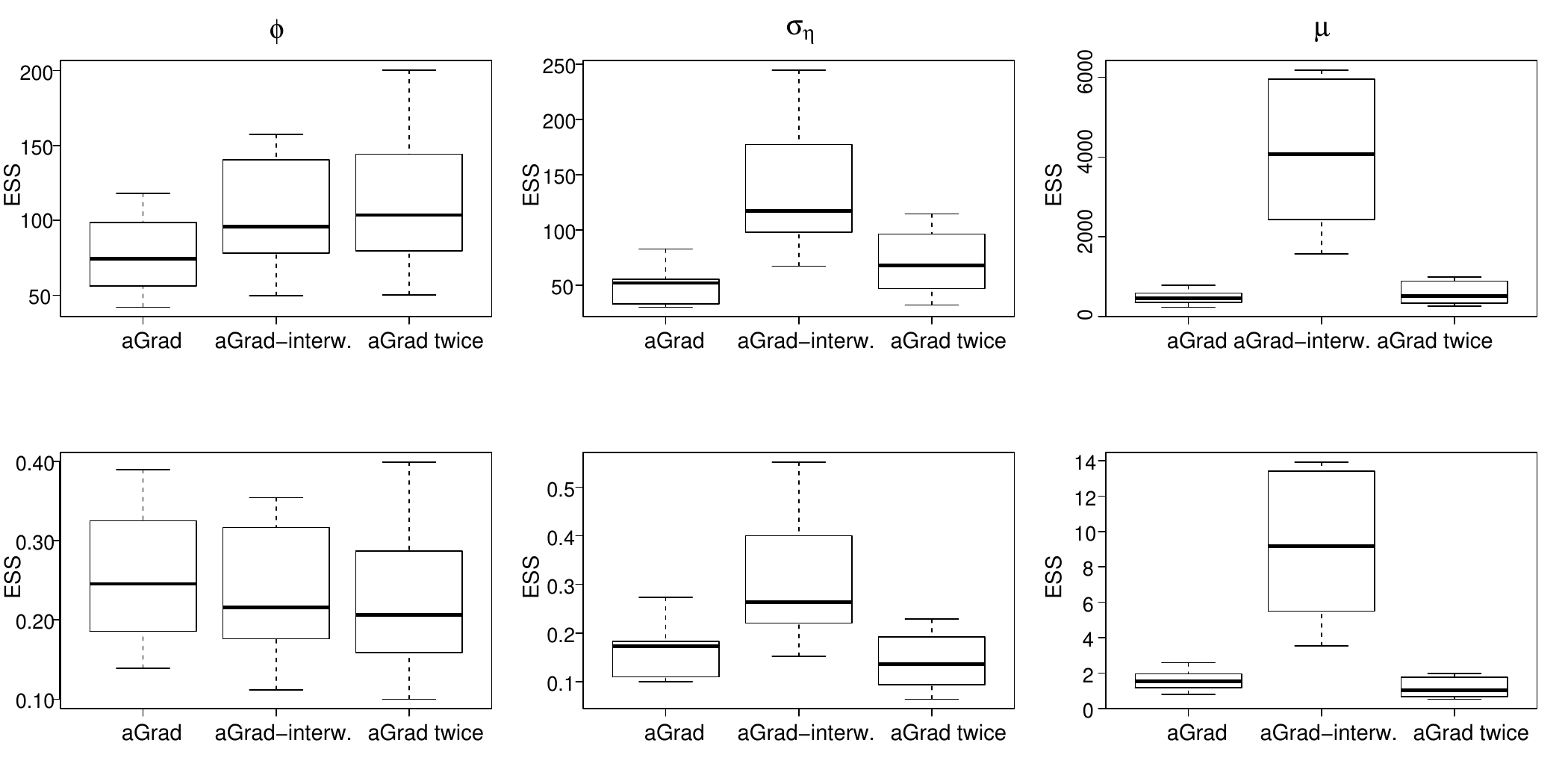}}
	\caption{Top: the effective sample size (ESS) of the posterior samples of the parameters $\mu$, $\sigma_{\eta}$ and $\phi$ obtained by applying the auxiliary gradient-based (aGrad) sampler with and without interweaving in $10$ replicates of a dataset consisted of $1,500$ log-returns simulated from the univariate SV model with $\mu=-0.85$, $\sigma_{\eta}=0.15$ and $\phi=0.90$. Bottom: the corresponding ESS divided by the number of seconds required for its sampling scheme.}
	\label{090}
\end{figure} 

The comparisons of the three sampling schemes are based on $10$ replicates of two datasets consisted of $T=1,500$ log-returns. These were simulated from the univariate SV model without jumps which is defined by $(3.1)$ and $(3.2)$ of the main paper, with $p=1$, and by omitting the jump component in $(3.1)$. For the remaining of this subsection we omit the subscript $i$ from our notation; e.g. $\mu \equiv \mu_i$ and $H \equiv H_i$. We set $\mu=-0.85$, $\sigma_{\eta}=0.15$ to generate both datasets and $\phi=0.98$ in the first and $\phi=0.90$ in the second. In each replicate of the datasets we applied the sampling schemes under comparison in order to draw samples for the latent log-volatility path $H$ and its parameters $\mu$, $\sigma_{\eta}$ and $\phi$.  We update $H$ jointly with $\phi$ and $\sigma^2_{\eta}$ by using the auxiliary gradient-based sampler to draw from \eqref{eq:logvoltarget} (with $n_t=0$) and then we draw $\mu$ from its full conditional. In the case of combining the sampler with the ancillarity-sufficiency interweaving strategy we follow the procedure described in the previous section in order to update the parameters $\mu$, $\sigma_{\eta}$ and $\phi$ under the non-centered parametrization of the log-volatility process. We ran all the MCMC algorithms for $30,000$ iterations and we discarded the first $10,000$ as burn-in period obtaining, thus, $20,000$ posterior samples for the log-volatility path and its parameters for each replicate of the two simulated datasets. In the Figures \ref{098} and \ref{090} we present, for the datasets with $\phi=0.98$ and $\phi=0.90$ respectively, the boxplots with the ESS of the posterior samples of the parameters $\mu$, $\sigma_{\eta}$ and $\phi$ obtained from the $10$ replicates. The two figures indicate that by combining the auxiliary gradient-based sampler with the ancillarity-sufficiency interweaving strategy the ESS of $\sigma_{\eta}$ and $\mu$ 
is clearly improved, while the three schemes have almost the same efficiency in drawing $\phi$.

\subsection*{Sampling the latent factors and their parameters}

To draw the latent factors $F$ jointly with the parameters $A$ we utilize the auxiliary gradient-based sampler summarized by Algorithm $3$ in the main paper and we target the distribution with density
\begin{equation*}\label{eq:facpost}
p(F,A|N,B,W) \propto \pi(A)\mathcal{N}_{K(T+1)}(F|0,C_A)\prod_{t=1}^T\prod_{i=1}^n e^{-\lambda_{it}}\lambda_{it}^{n_{it}},
\end{equation*}
where $\pi(A)$ denotes the prior density of $A$ and $C_A$ is the covariance matrix (which depends on the values of the parameters $A$) of the autoregressive process defined by $(3.7)$ and $(3.8)$.
Since the latent factor model defined by equations $(3.6)$--$(3.8)$ corresponds to a latent Gaussian model with zero mean Gaussian field, the application of Algorithm $3$ of the main paper is straightforward. The latent path $\mathrm{x}$ in Algorithm $3$ corresponds to the path of the latent factors $F$ and the parameters $\theta$ are the parameters $\alpha_1,\ldots,\alpha_K$ of the matrix $A$. We note that the parameters $\alpha_k$, $k=1,\ldots,K$, are defined on the interval $(-1,1)$. Therefore, we apply the transformation $\hat{\alpha}_k = \log \big( (\alpha_k+1)/(1-\alpha_k) \big)$ in order to obtain parameters in the real line. In the acceptance ratio in $(4.6)$ of the main paper we need thus to account for the Jacobian of the inverse transformation.

\subsubsection*{Using the interweaving strategy to improve sampling of $A$}
We combine the auxiliary gradient-based sampler (Algorithm $3$ in the main paper) with the ancillarity-sufficiency interweaving strategy in order to reduce the autocorrelation of the parameters in $A$. More precisely, after sampling jointly the paths $F$ of the latent factors and the parameters $A$ we set $\Gamma_0 = DF_0$, where $D$ is a $K \times K$ diagonal matrix with elements $(1-\alpha_k^2)^{1/2}$, $k=1,\ldots,K$, and $\Gamma_t = F_t - AF_{t-1}$, for $t=1,\ldots,T$. Then we use a random walk Metropolis-Hastings step that targets the density
\begin{equation*}\label{eq:facpost2}
p(A|\Gamma,N,B,W) \propto \pi(A)\prod_{t=1}^T\prod_{i=1}^n e^{-\lambda_{it}}\lambda_{it}^{n_{it}},
\end{equation*}
with  $\lambda_{it} = \lambda^*/(1+\exp(-b_i-W'_iF_t))$. We note that $F_t$ and $\Gamma_t$ are related by $F_0 = D^{-1}\Gamma_0$ and $F_t = AF_{t-1} + \Gamma_t$, $t \geq 1$.

Here, we compare the efficiency of three different MCMC schemes for the update of the parameters $A$. These are the auxiliary gradient-based sampler with interweaving, without interweaving, and a scheme in which instead of interweaving we iterate twice the auxiliary sampler. To simplify our experiment we omit equations $(3.1)$ and $(3.2)$ of the main paper and we simulate observations from the model defined by equations $(3.3)$ and $(3.5)$--$(3.8)$. This a Cox model with intensities driven by latent dynamic factors and
Bayesian inference for this model is conducted by switching off certain steps (see Section $4$) of Algorithm $1$. To compare the three sampling schemes we simulated $p=100$ time series consisted of $T = 1,500$ observations $n_{it}$, $t=1,\ldots,T$. We used $K=2$ autoregressive latent factors for the jump intensities, we set $\alpha_1 =0.8$ and $\alpha_2 =0.4$ and we generated $w_{ik} \sm N(0,1)$. We ran the MCMC algorithms, that correspond to the three different schemes for the update of $F$ and $A$, for $30,000$ iterations and we discarded the first $10,000$ as burn-in period. Figure \ref{acfA} presents the estimated autocorrelation functions of the parameters in $A$ for each one of the three sampling schemes.

\begin{figure}[H]
	\centerline{\includegraphics[scale=0.65]{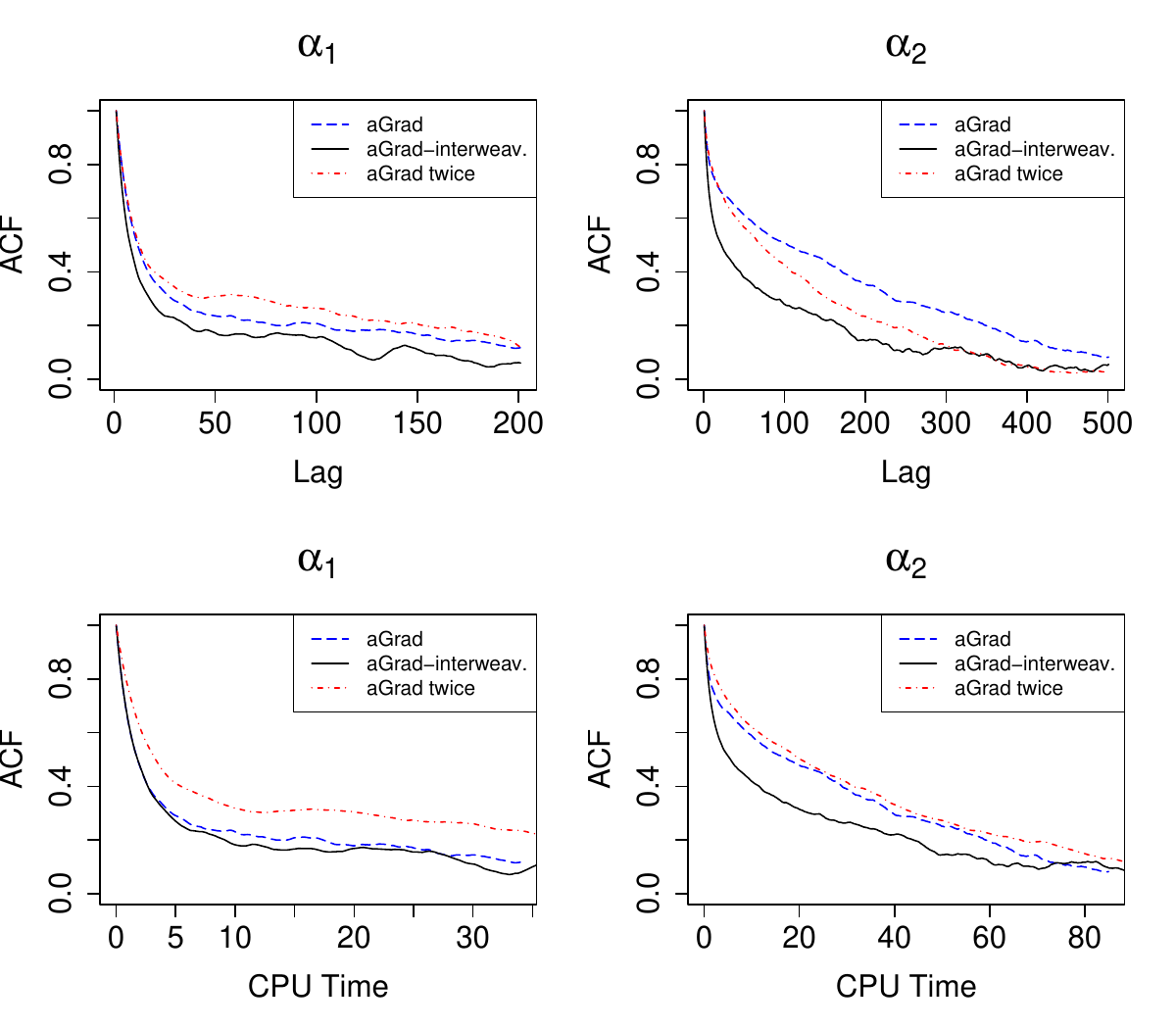}}
	\caption{Estimated autocorrelation functions (ACF) of the parameters $\alpha_1$ and $\alpha_2$ of the latent autoregressive factors used for the intensities of the simulated time series from the model defined by equations $(3.3)$ and $(3.5)$--$(3.8)$ of the main paper. The three sampling schemes under comparison are based on the auxiliary gradient-based (aGrad) sampler. Top: ACF against time-lag, bottom: ACF against CPU time.}
	\label{acfA}
\end{figure}

\subsection*{Sampling $\mu_{i\xi}$ and $\sigma^2_{i\xi}$}

The full conditional distribution of $\mu_{i\xi}$ is 
\begin{equation*}
N\left( \frac{5\textrm{range}_i^2\sum_{t=1}^T \sum_{\kappa=1}^{n_{it}} \xi_{it}^{\kappa}}{5\textrm{range}_i^2\sum_{t=1}^{T}n_{it}+ \sigma^2_{i\xi}}, \frac{5\textrm{range}_i^2\sigma^2_{i\xi}}{5\textrm{range}_i^2\sum_{t=1}^{T}n_{it}+ \sigma^2_{i\xi}}   \right),\,\,\, i=1,\ldots,p
\end{equation*}
and the full conditional distribution of $\sigma^2_{i\xi}$ is
\begin{equation*}
IGam \left( 3+ 0.5\sum_{t=1}^Tn_{it}, \textrm{range}_i^2/18 + 0.5\sum_{t=1}^T\sum_{\kappa=1}^{n_{it}} (\xi_{it}^{\kappa}-\mu_{i\xi})^2 \right),\,\,\, i=1,\ldots,p,
\end{equation*}

\section*{Particle filter and annealed importance sampling}
Here we present the pseudocode for the implementation of the particle filter and annealed importance sampling algorithms that we constructed in order to estimate predictive Bayes factors.
A superscript $s$ in Algorithm \ref{alg:pf} means that the operation is performed for all $s = 1,...,S$.
The infinite sum in \eqref{eq:uniwts} is truncated at $10$ without affecting the performance of the algorithm, see also \cite{johannes2009optimal} for a discussion. The step in the $9$th line of Algorithm \ref{alg:pf}  is conducted by using multinomial resampling, see \cite{doucet2009tutorial} for a detailed presentation of resampling methods. The \texttt{for} loop that begins in the $2$nd line of Algorithm \ref{alg:ais} is performed in parallel. In the $11$th line of Algorithm \ref{alg:ais} we use a few iterations of Algorithm $1$ initialized with the samples obtained from the $(j-1)$th step.

\begin{algorithm}[H]
\caption{Particle filter algorithm that targets $(5.3)$ for $t =T+1,\ldots, T+ \ell$.}\label{alg:pf}
\begin{algorithmic}[1]
\State \textbf{Set} the number of desired samples $S$ and give as \textbf{inputs}; in sample and out of sample observations $r_{i,1:T} $ and $r_{i,T+1:T+\ell}$ from the $i$th stock and $S$ samples from $p(h_{i,1:T}|r_{i,1:T})$.
\For {\texttt{$t = T+1,\ldots,T+\ell$ }}
\If{$t=T+1$}
\State \textbf{Sample}  $h_{it}^s \sim N( \mu_i + \phi_i (h^s_{i,t-1}-\mu_i), \sigma_{i \eta}^2    )$
\State \textbf{Set} 
\begin{equation}
\label{eq:uniwts}
\omega_{t}^s = 
\small
p(r_{it}| h_{it}) = \left\{
\begin{array}{ll}
      \sum_{n_{it}=0}^{\infty} \mathcal{N}(r_{it}|n_{it}\mu_{i\xi},e^{h_{it}}+n_{it}\sigma_{i\xi}^2)p(n_{it}), & \text{if SV with jumps} \\
      \mathcal{N}(r_{it}| 0 , e^{h_{it}} ), &\hspace*{-0.5cm} \text{if SV without jumps}\\
\end{array} 
\right\}
\end{equation}
\normalsize
\State \textbf{Set} $\hat{p}(r_{it}|r_{i,1:t-1}) = \sum_{s=1}^S \omega^s_{t} /S$
\Else
\State  \textbf{Compute} $ESS_{t-1} = \dfrac{ (\sum_{d=1}^s \omega_{t-1}^s)^2  }{ \sum_{s=1}^S (\omega_{t-1}^s)^2 }.$
\State \textbf{If $ESS_{t-1} < S/2$} 
obtain ancestor variables $o_{t-1}^{s}$ and set $\omega^s_{t-1} = 1$ \textbf{else } set  $o_{t-1}^{s} =s$.
\State \textbf{Sample} $h_{it}^s \sim N( \mu _i+ \phi_i (h^{o_{t-1}^{s}}_{i,t-1}-\mu_i), \sigma_{i \eta}^2    )$ and \textbf{set} $\omega_{t}^s = \omega^s_{t-1} p(r_{it}|h^s_{it})$
\State \textbf{If} $ESS_{t-1}<S/2$ set $\hat{p}(r_{it}|r_{i,1:t-1}) = \sum_{s=1}^S \omega^s_{t} /S$ \textbf{else} set  $\hat{p}(r_{it}|r_{i,1:t-1}) = \sum_{s=1}^S \omega^s_{t} / \sum_{s=1}^S \omega^s_{t-1}$

 \EndIf      
\EndFor \\
\Return  $\hat{p}(r_{i,T+1}| r_{i,1:T}),\ldots,\hat{p}(r_{i,T+\ell}|r_{i,1:T+\ell-1})$
\end{algorithmic}
\end{algorithm}

\begin{algorithm}[H]
\caption{Annealed importance sampling algorithm that targets $p(H_{0:T+\ell},F_{0:T+\ell}|R_{1:T+\ell})$.}\label{alg:ais}
\begin{algorithmic}[1]
\State \textbf{Set} the number of desired samples $S$ and give as \textbf{inputs}; $S$ samples from $p(H_{0:T},F_{0:T}|R_{1:T})$ and the out of sample observations $R_{T+1:T+\ell}$.
 \For{\texttt{$s = 1,\ldots,S$ }}
    \For{\texttt{$j = 1,\ldots,\ell$ }}
      \If {$ j=1$ }
 
      \State \textbf{Sample} $F^s_{T+1} \sim N(AF_T^s,I_{K})$
       \For{\texttt{$i = 1,\ldots,p$ }}
       
       \State \textbf{Sample} $h^s_{i,T+1} \sim N(\mu_i + \phi_i(h^s_{iT}-\mu_i), \sigma^2_{i\eta}   )$  
       \State \textbf{Set} $\omega^s_{i,T+1} = p(r_{i,T+1} | h_{i,T+1}^s,F_{T+1}^s)$; the $i$th element of the product in $(4.3)$
     \EndFor
     \Else
      \State \textbf{Sample} $(F^s_{T+j-1},H_{T+j-1}^s) \sim p(F_{0:T+j-1},H_{0:T+j-1}|R_{1:T+j-1})$ 
       \State \textbf{Sample} $F^s_{T+j} \sim N(AF_{T+j-1}^s,I_{K})$
        \For{\texttt{$i = 1,\ldots,p$ }}
       
       \State \textbf{Sample} $h^s_{iT+j} \sim N(\mu_i + \phi_i(h^s_{i,T+j-1}-\mu_i), \sigma^2_{i\eta}   )$  
       \State \textbf{Set} $\omega^s_{i,T+j} = \omega^s_{i,T+j-1}p(r_{i,T+j} | h_{i,T+j}^s,F_{T+j}^s)$
     \EndFor

\EndIf
   \EndFor
   \EndFor\\
     \Return  $\{\omega^s_{it}\}_{s=1}^S$ for each $i=1,\ldots,p$ and $t=T+1,\ldots,T+\ell$
\end{algorithmic}
\end{algorithm}

\section*{Approximation with a mixture of Gaussian distributions}

The approximation of the model defined, in the main paper, by equations $(3.1)$--$(3.3)$ with a mixture of Gaussian distributions relies on the transformation 
\begin{equation*}\label{eq:chibtrans}
r_{it}^* = \log(r_{it} - \sum_{\kappa=1}^{n_{it}}\xi_{it}^{\kappa})^2 = h_{it} + \log (\epsilon_{it}^2) ,\,\,\ t=1,\ldots,T,
\end{equation*}
and on the assumption that $\log (\epsilon_{it}^2 )|u_{it} \sim \mathcal{N}(m_{u_{it}},s^2_{u_{it}})$ where $u_{it} \in \{ 1,\ldots,10 \} $ denotes a mixture component indicator. 
The means and the variances of the mixture components as well as the probabilities $p(u_{it}=j)$, $j=1,\ldots,10$, can be chosen as suggested by \cite{omori2007stochastic}. See \cite{chib2002markov} and \cite{nakajima2009leverage} for details.

\section*{Continuous ranked probability score and mean root squared error}

\subsection*{Continuous ranked probability score (CRPS)}

Let $F_r$ be the cdf of the posterior predictive distribution of $r_{it}$.  The CRPS \citep{matheson1976scoring} is defined in terms of the predictive $F_r$ and is given by
\begin{align}
\label{eq:crps_true}
\operatorname{CRPS}(F_r,r_{it}) &= \int_\mathbb{R} (F_r(z) - \1{r_{it} \leq z})^2 \, \mathrm{d} z,  \\ \label{eq:crps_true2}
& =  \mathrm{E}_{F_r} | \tilde{r}_{it}^1 - r_{it} | - \frac{1}{2} \mathrm{E}_{F_r} | \tilde{r}_{it}^1 - \tilde{r}_{it}^2 |,
\end{align}
where $\tilde{r}_{it}^1$ and $\tilde{r}_{it}^2$ are independent random variables with distribution $F_r$ and it is assumed that the first moment of $F_r$ is finite. The CRPS defined by equation \eqref{eq:crps_true} is negatively oriented and, thus, a small value indicates a more accurate forecast. Moreover, the negative oriented CRPS reduces to the mean absolute error in the case of a point forecast; see for example \cite{gneiting2007strictly} for more details. Since $F_r$ is not available in an analytic form we need to estimate \eqref{eq:crps_true2} by utilizing the $S$ samples $\{\tilde{r}^s_{it}\}_{s=1}^S$ drawn from the posterior predictive distribution. To achieve this we first compute the empirical cdf 
\begin{equation}
    \label{eq:emp_cdf}
    \hat{F}_r(z) = \frac{1}{S} \sum_{s=1}^S \1{\tilde{r}_{it}^s \leq z}.
\end{equation}
Then, an estimation for the CRPS is given by
\begin{equation}
\label{eq:crps_est}
\operatorname{CRPS}(\hat F_r,r_{it}) = \frac{1}{S} \sum_{s=1}^S |\tilde{r}_{it}^s - r_{it}| - \frac{1}{2 S^2} \sum_{s=1}^S \sum_{q=1}^S |\tilde{r}_{it}^s - \tilde{r}_{it}^q|.
\end{equation}
We note that in order to compute \eqref{eq:crps_est} we utilize the r-package \texttt{scoringRules} \citep{jordan2017evaluating} where the computation is conducted efficiently by employing a representation of \eqref{eq:crps_est} which is based on order statistics.

\subsection*{Root mean squared error (RMSE)}

To assess point predictions obtained from the different models we calculate the RMSE of the draws from the posterior predictive distribution of the true log-returns $r_{it}$ as

\begin{align*}
    \operatorname{RMSE}_{i}&=\sqrt{\frac{1}{\ell}\sum_{t=T+1}^{T+\ell}\bigg( \frac{1}{S}\sum_{s=1}^S (\tilde{r}^s_{it}-r_{it})^2 \bigg)}\\ 
    & =\sqrt{\frac{1}{\ell} \sum_{t=1}^{T+\ell}\bigg[ (\hat{\tilde{r}}_{it}-r_{it})^2 + \frac{1}{S}\sum_{s=1}^S(\tilde{r}_{it}^{s}-\hat{\tilde{r}}_{it})^2 \bigg]},\,\,\, i=1,\ldots,p
\end{align*}
where $\hat{\tilde{r}}_{it}$ denotes the mean of the posterior predictive distribution.

\section*{Results on simulated data}

\subsection*{Results for independent SV with jumps models}

By using the approach described in Section $3.2$ of the main paper to model the jump intensities in equation $(3.3)$ we obtain $p$ independent univariate SV with jumps models. Bayesian inference for their parameters and latent states can be conducted with the MCMC algorithm derived from Algorithm $1$ by omitting the $10$th step and by substituting the $12$th step (sampling latent factors) with the step of sampling the independent jump intensities from their inverse Gamma full conditionals. 

To test our methods we simulated $T = 1,500$ log-returns from $p=4$ independent univariate models as defined by equations $(3.1)$--$(3.3)$ with $Gam(1,c)$ priors for the jump intensities. We used $c=50$ to simulate the first and the second and $c=100$ to simulate the third and the fourth time series in order to represent observed log-returns with more and less jumps respectively.
Following the related literature (see for example \cite{chib2002markov} and \cite{eraker2003impact}) we set $\mu_i=-0.85$, $\phi_i=0.98$ and $\sigma_{i\eta} = 0.12$ for each $i=1,2,3,4$ and $\mu_{1\xi}=\mu_{3\xi}=-3$ and $\mu_{2\xi}=\mu_{4\xi}=0$. We ran the MCMC algorithm for $80,000$ iterations from which we discarded the first $20,000$ iterations as burn-in period and we were storing the sampling outcome of every $60$th iteration. Figure \ref{uniobs} presents the simulated log-returns along with the times in which jumps have been generated and times with estimated posterior probability of at least one jump greater than $50\%$. Figure \ref{univols} displays the simulated paths for the corresponding log-volatilities and the $95\%$ credible intervals of their posterior distributions. By the visual inspection of the figures it is clear that with the proposed methods we have separated the jumps from the log-volatility process in all the simulated frameworks.

In order to check the ability of the proposed methods to identify jumps in time series with varying volatility of the  log-volatility process we conducted the following experiment. We simulated $p=4$ time series, from the model defined by equations $(3.1)$--$(3.3)$, with $T=1,500$ observations in each one by using different values for $\sigma_{i\eta}$, $i=1,\ldots,4$; $\sigma_{1\eta}=0.1$, $\sigma_{2\eta}=0.15$, $\sigma_{3\eta}=0.25$ and $\sigma_{4\eta}=0.35$. We used $Gam(1,50)$ priors for the jump intensities, we set $\mu_{1\xi}=\mu_{2\xi}=0$ and the same values as in the previous experiment for all the other parameters. The same MCMC algorithm as in the first experiment was used to obtain posterior samples for the parameters and the latent states of the model. Figure \ref{univolssens} presents the simulated log-returns for each stock and indicates the times in which we estimated posterior probabilities of at least one jump greater than $50\%$, while Figure \ref{univolssensprobs} displays the posterior probabilities of at least one jump for each one of the $1,500$ time points of each stock. The figures demonstrate that using the proposed methods we are able to disentangle the volatility from the jump process in all the different scenarios for the log-volatility process, since in the majority of the simulated jumps have been identified and only jumps with size that is in the range of the volatility process are not detected.

\begin{figure}[H]
	\centerline{\includegraphics[scale=0.7]{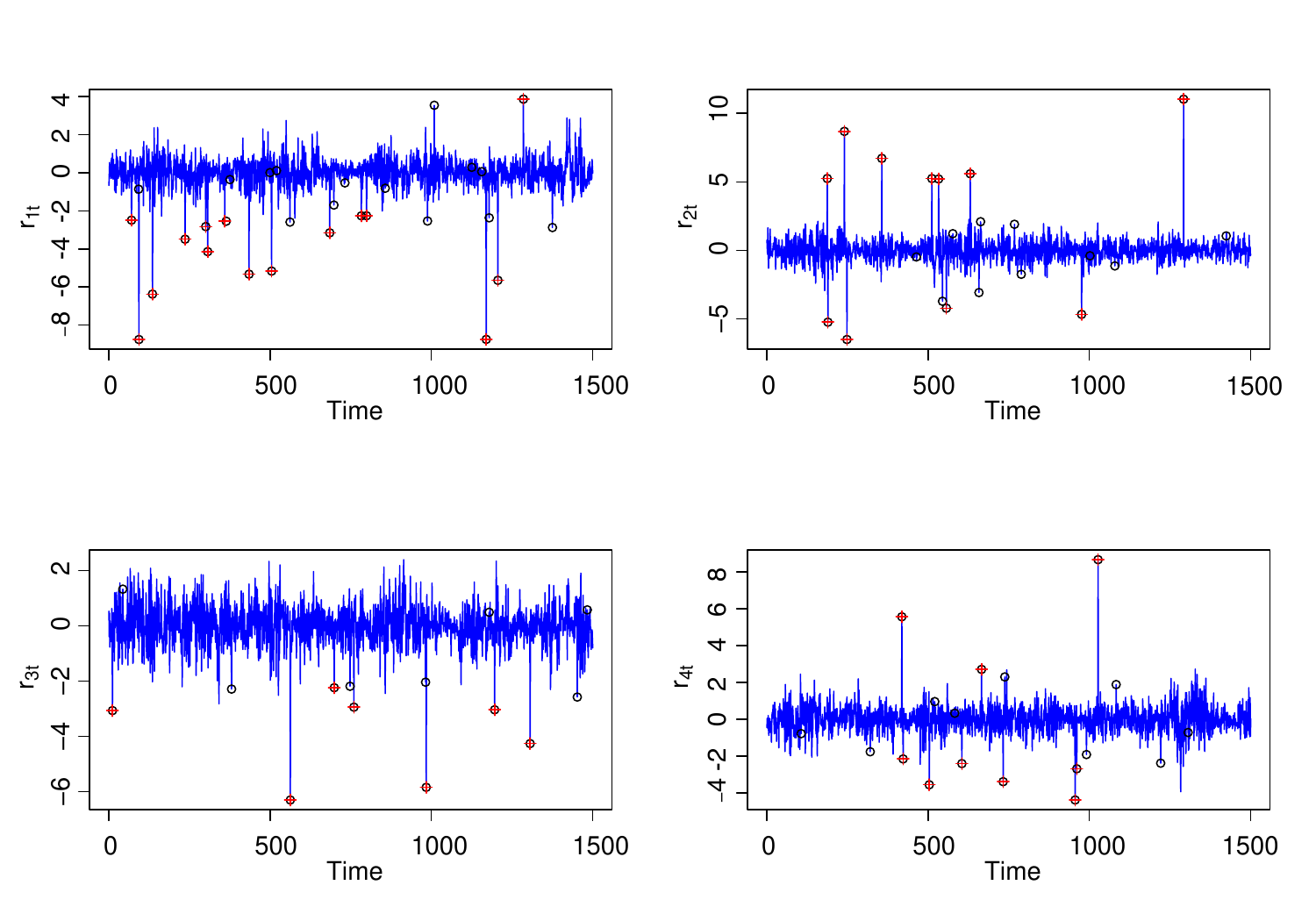}}
	\caption{Simulated log-returns from $4$ independent univariate SV with jumps models defined by equations $(3.1)$--$(3.3)$ of the main paper with $Gam(1,50)$ (top) and $Gam(1,100)$ (bottom) prior for the jump intensities with $\mu_{1\xi}=\mu_{3\xi}=-3$ (left column) and $\mu_{2\xi}=\mu_{4\xi}=0$ (right column). We set $T=1,500$, $\mu_i =-0.85$, $\phi_i=0.98$, $\sigma_{i\eta} =0.12$, $\sigma_{i\xi}=3.5$, $i=1,2,3,4$. Black circles indicate times in which a jump has been simulated. Red crosses indicate times with estimated posterior probabilities of at least one jump greater than $50\%$. }
	\label{uniobs}
\end{figure} 

\begin{figure}[H]
	\centerline{\includegraphics[scale=0.7]{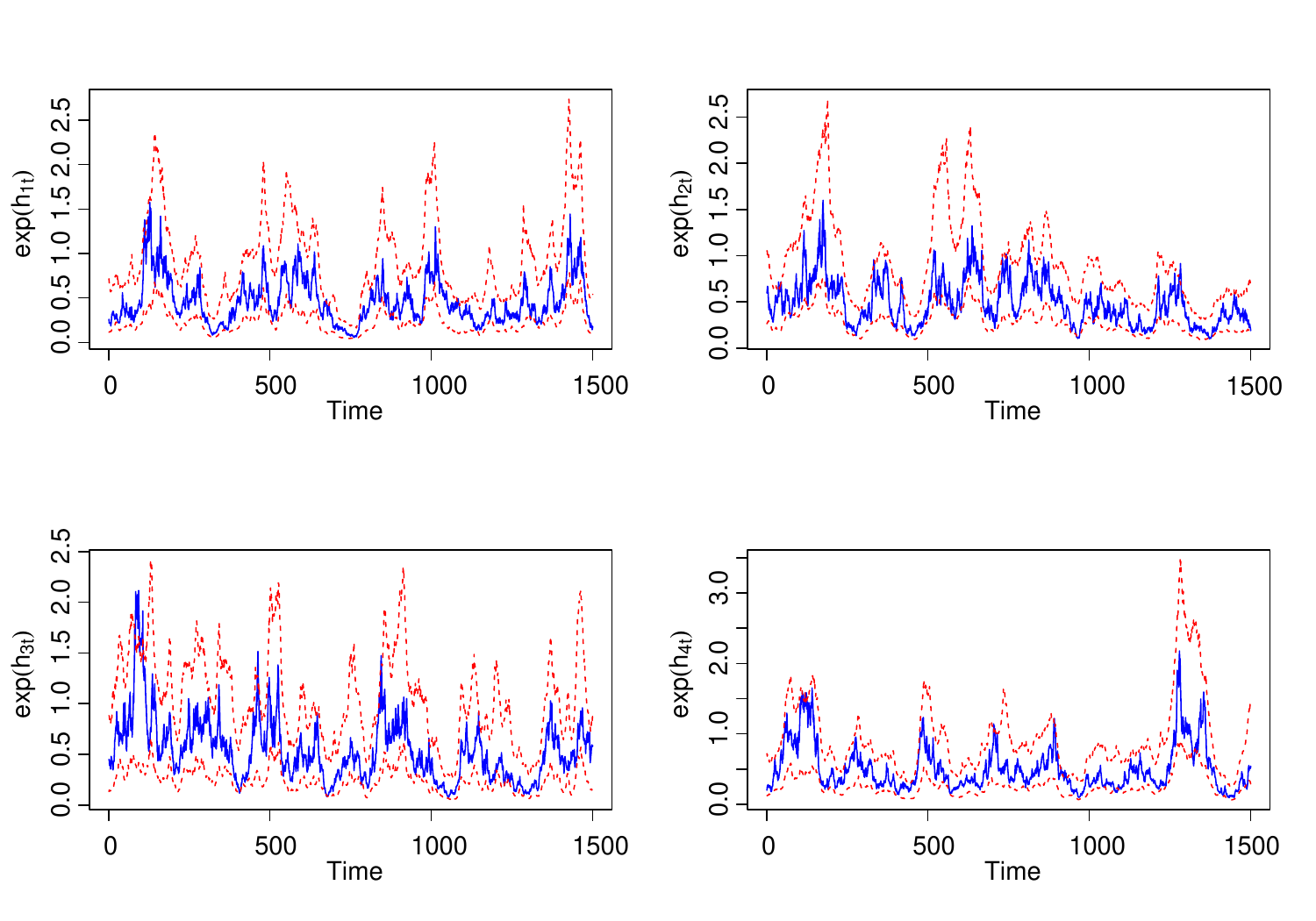}}
    \caption{Simulated paths of the volatility process (solid lines) and $95\%$ posterior credible intervals (dashed lines). The log-returns are simulated from $4$ independent univariate SV with jumps models defined by equations $(3.1)$--$(3.3)$ of the main paper with $Gam(1,50)$ (top) and $Gam(1,100)$ (bottom) prior for the jump intensities and with $\mu_{1\xi}=\mu_{3\xi}=-3$ (left column) and $\mu_{2\xi}=\mu_{4\xi}=0$ (right column). We set $T=1,500$, $\mu_i =-0.85$, $\phi_i=0.98$, $\sigma_{i\eta} =0.12$, $\sigma_{i\xi}=3.5$, $i=1,2,3,4$.}
	\label{univols}
\end{figure}

\begin{figure}[H]
	\centerline{\includegraphics[scale=0.7]{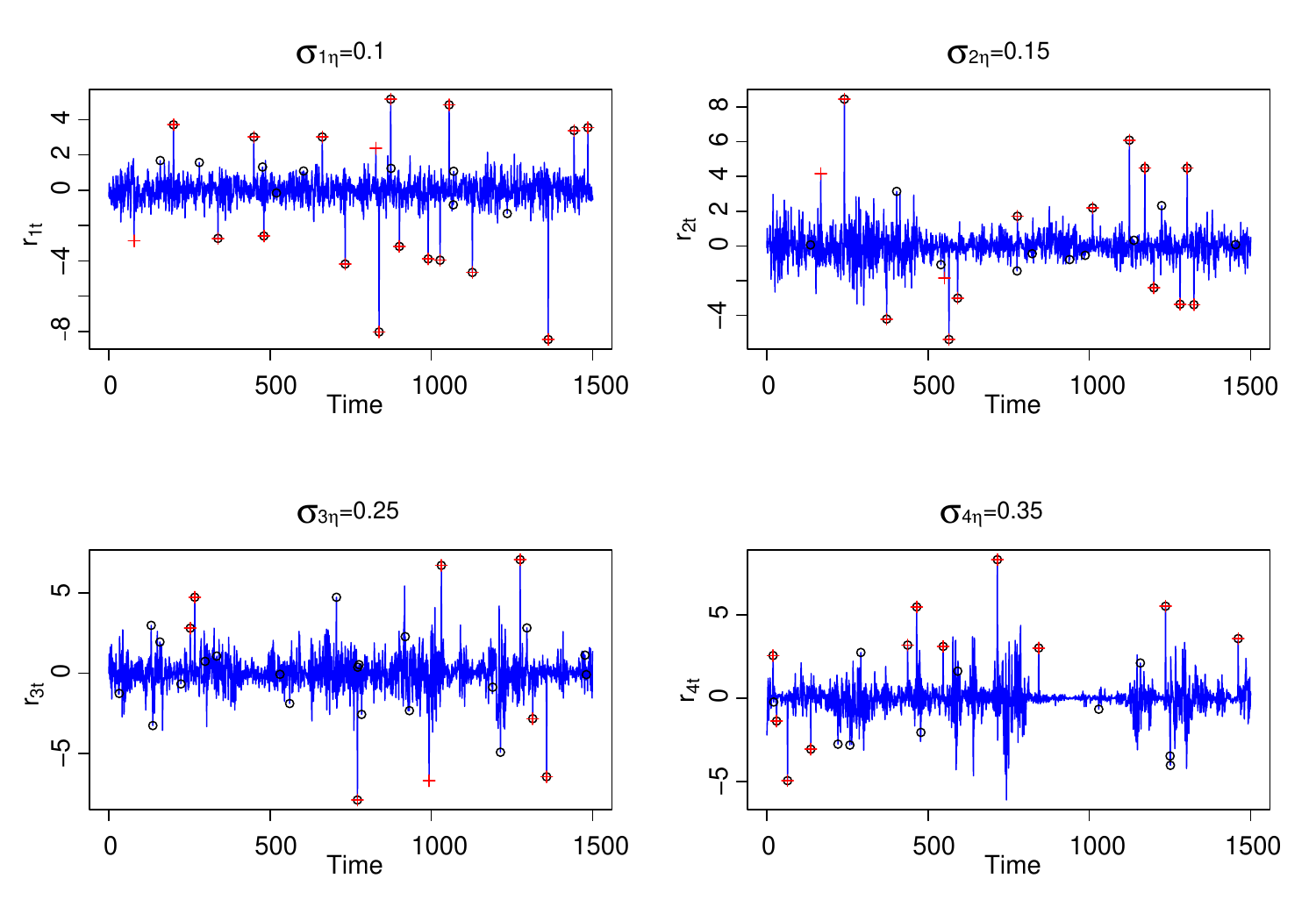}}
	\vspace*{-0.5cm}
	\caption{Simulated log-returns from $p=4$ independent univariate SV with jumps models defined by equations $(3.1)$--$(3.3)$ of the main paper with $Gam(1,50)$ priors for the jump intensities. We set $T=1,500$, $\mu_i =-0.85$, $\phi_i=0.98$, $\mu_{i\xi}=0$ and $\sigma_{i\xi}=3.5$ and we used different values for $\sigma_{i\eta}$, $i=1,2,3,4$. Black circles: times in which a jump has been simulated. Red crosses: times with estimated posterior probabilities of at least one jump greater than $50\%$.}
	\label{univolssens}
\end{figure} 

\begin{figure}[H]
	\centerline{\includegraphics[scale=0.7]{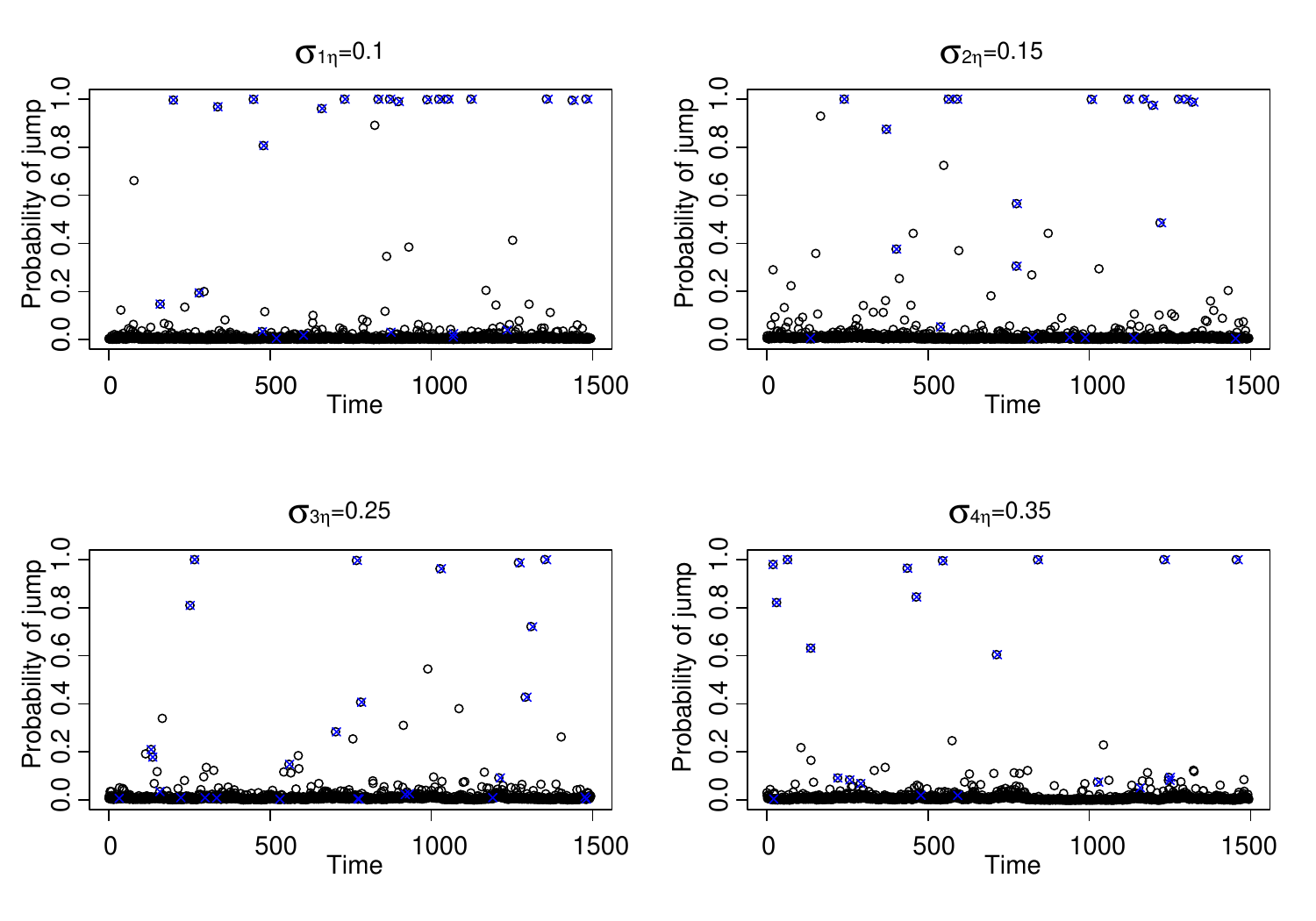}}
	\vspace*{-0.5cm}
	\caption{Black circles: Posterior probabilities for at least one jump in simulated log-returns from independent univariate SV with jumps models defined by equations $(3.1)$--$(3.3)$ of the main paper with $Gam(1,50)$ priors for the jump intensities. We set $T=1,500$, $\mu_i =-0.85$, $\phi_i=0.98$, $\mu_{i\xi}=0$ and $\sigma_{i\xi}=3.5$ and we used different values for $\sigma_{i\eta}$, $i=1,2,3,4$. Blue crosses: times of simulated jumps.}
	\label{univolssensprobs}
\end{figure} 

\subsection*{Results for joint modelled jump intensities}

We present results from the application of our methods on $p=100$ time series simulated from the model defined by equations $(3.1)$--$(3.3)$ and $(3.5)$--$(3.8)$ in the main paper. We used $K=3$ latent factors with $\alpha_1 = 0.9$, $\alpha_2 =0.65$ and $\alpha_3=-0.65$ to simulate the paths of the latent factors. The factor loadings in $W$ were simulated from standard normal distributions and we simulated $b_i \sim \mathcal{N}(-2.45,1)$, $i=1,\ldots,p$. For the parameters of the latent log-volatility processes we used $\phi_i=0.98$, $\sigma_{i\eta} = 0.12$ and $\mu_i = -0.85$ and we set $\mu_{i\xi}=0$ and $\sigma_{i\xi}=3.5$ for each $i=1,\ldots,p$. Then we ran the MCMC algorithm summarized by Algorithm $1$ for $100,000$ iterations from which we discarded the first $30,000$ as burn-in period and we were saving the outcome of every $70$th iteration.

Each line with black circles in Figure \ref{jumpMVsim} refers to a stock while the circles indicate times in which a jumps has been simulated in the price of the stock. The red crosses indicate the times where the posterior probability of at least one jump has estimated to be greater than $50\%$.  Figure \ref{Apost} presents the posterior distributions for the parameters $\alpha_1$, $\alpha_2$ and $\alpha_3$.

\begin{figure}[H]
	\centerline{\includegraphics[scale=0.5]{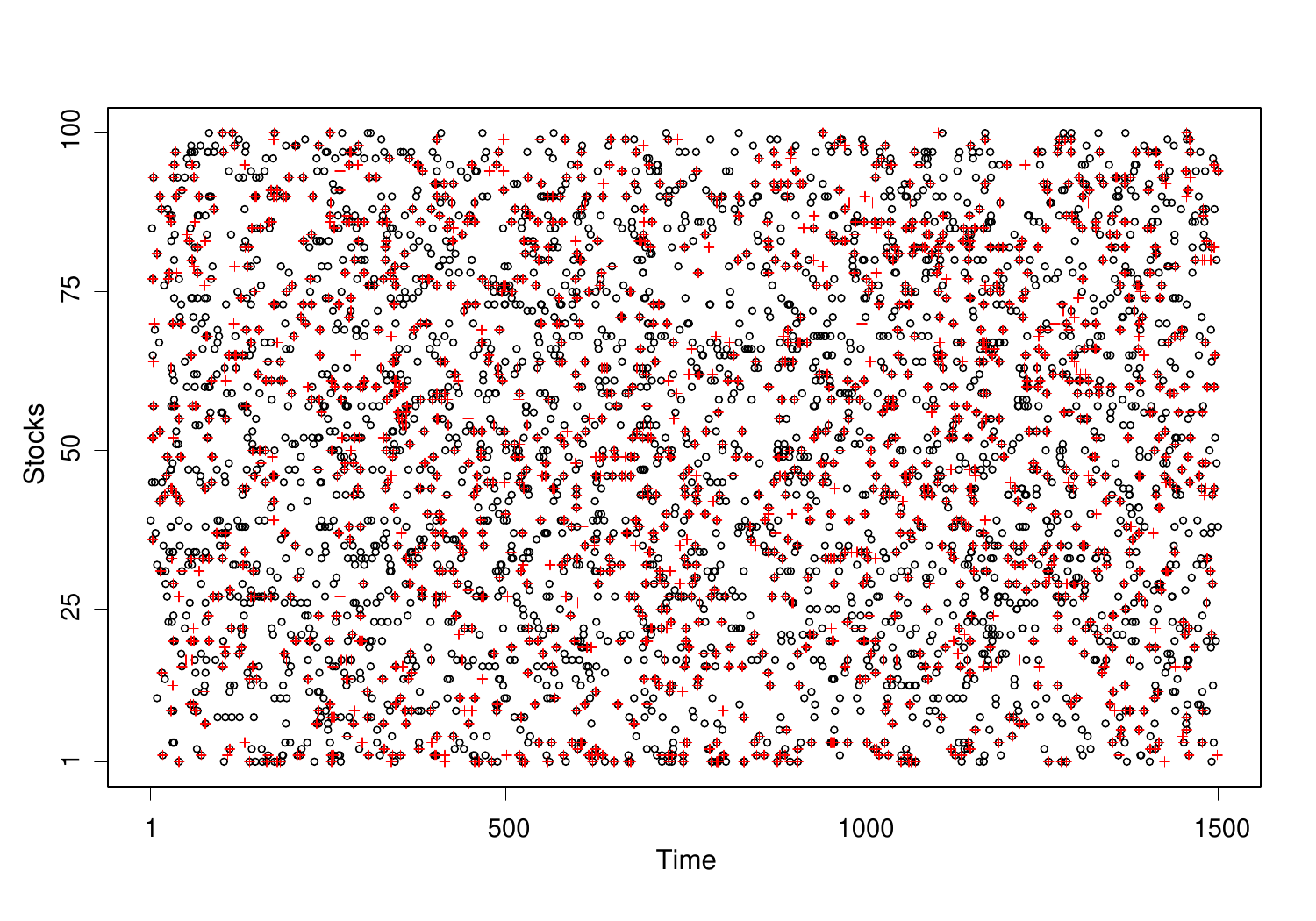}}
	\vspace*{-0.3cm}
	\caption{Simulated jump times (black circles) for each of the $p = 100$ simulated time series with $T=1,500$ log-returns simulated from the model defined by the equations $(3.1)$--$(3.3)$ and $(3.5)$--$(3.8)$ in the main paper. The red crosses indicate times with posterior probability of having at least one jump greater than $50\%$.}
	\label{jumpMVsim}
\end{figure} 

\begin{figure}[H]
	\centerline{\includegraphics[scale=0.55]{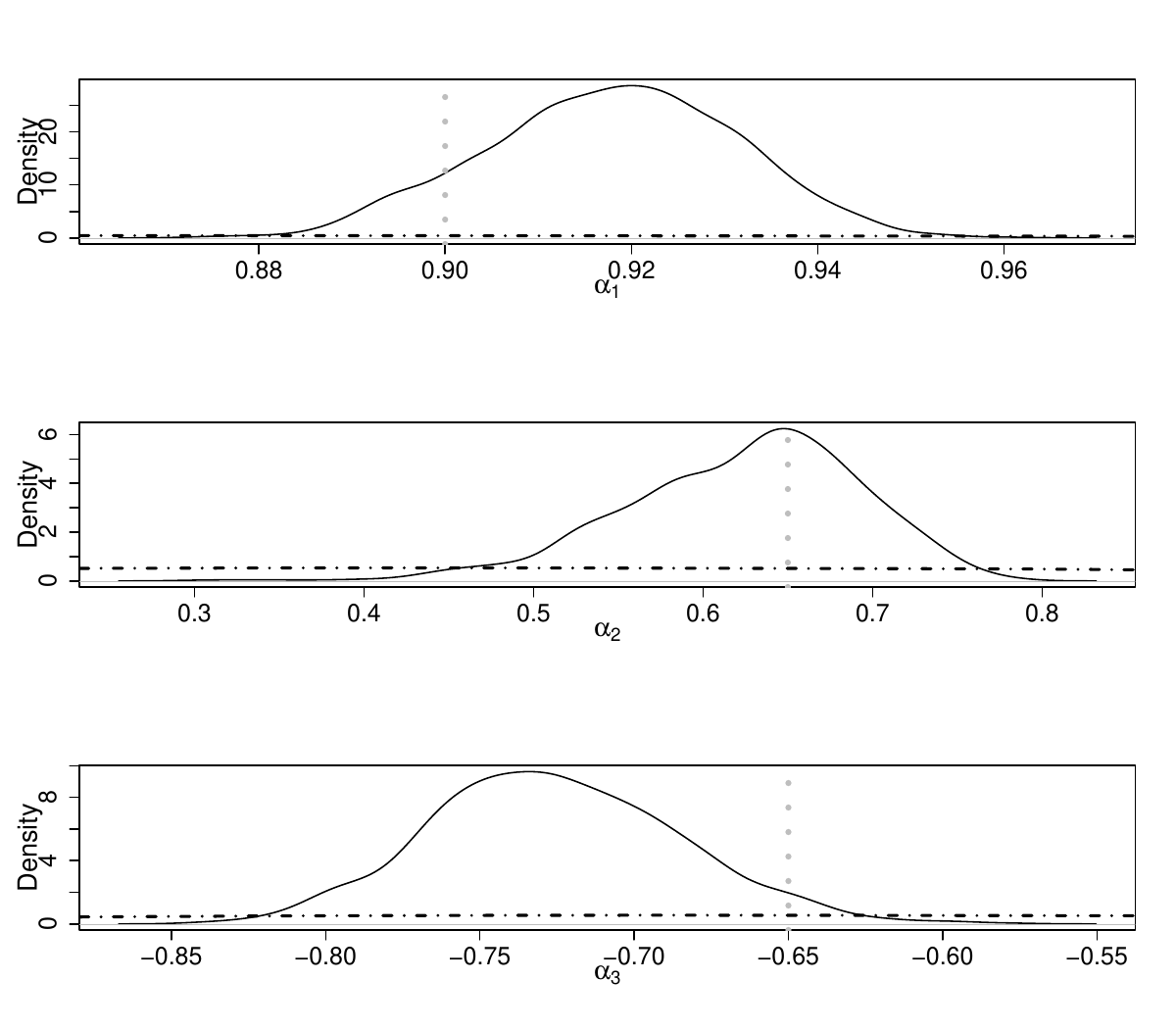}}\vspace*{-0.3cm}
	\caption{Posterior (black solid lines) and prior (black dashed lines) distributions for the parameters $\alpha_1$, $\alpha_2$ and $\alpha_3$ of the dynamic factor model used for the joint modelling of the jump intensities. The vertical gray dotted lines indicate the true values.}
	\label{Apost}
\end{figure} 

\section*{Additional results from the real data analysis}

\subsection*{Bayesian analysis of log-volatilities and jumps}

Figures \ref{fig:vol_pars} and \ref{fig:vol_CIs} compare the Bayesian inference conducted for the parameters and the latent log-volatility processes of SV models without jumps, with independent jumps and with jointly modelled jump intensities. From the visual inspection of the Figures it is evident that the SV models with jumps, either independent or jointly modelled, deliver lower log-volatility means $\mu_i$ and variances $\sigma^2_{i\eta}$ and higher persistent parameters $\phi_i$ than SV models without jumps while the proposed joint modelling of the jump intensities delivers the most narrow intervals and the SV models without jumps the widest. Figure \ref{fig:weekdays} depicts the posterior probabilities of the number of jumps for each weekday. This is clearly related with the parameter $\Delta_{it}$ in $(3.3)$ which takes into account the time distance of two successive observations.

\begin{figure}[H]
\begin{center}
\includegraphics[scale=0.5]{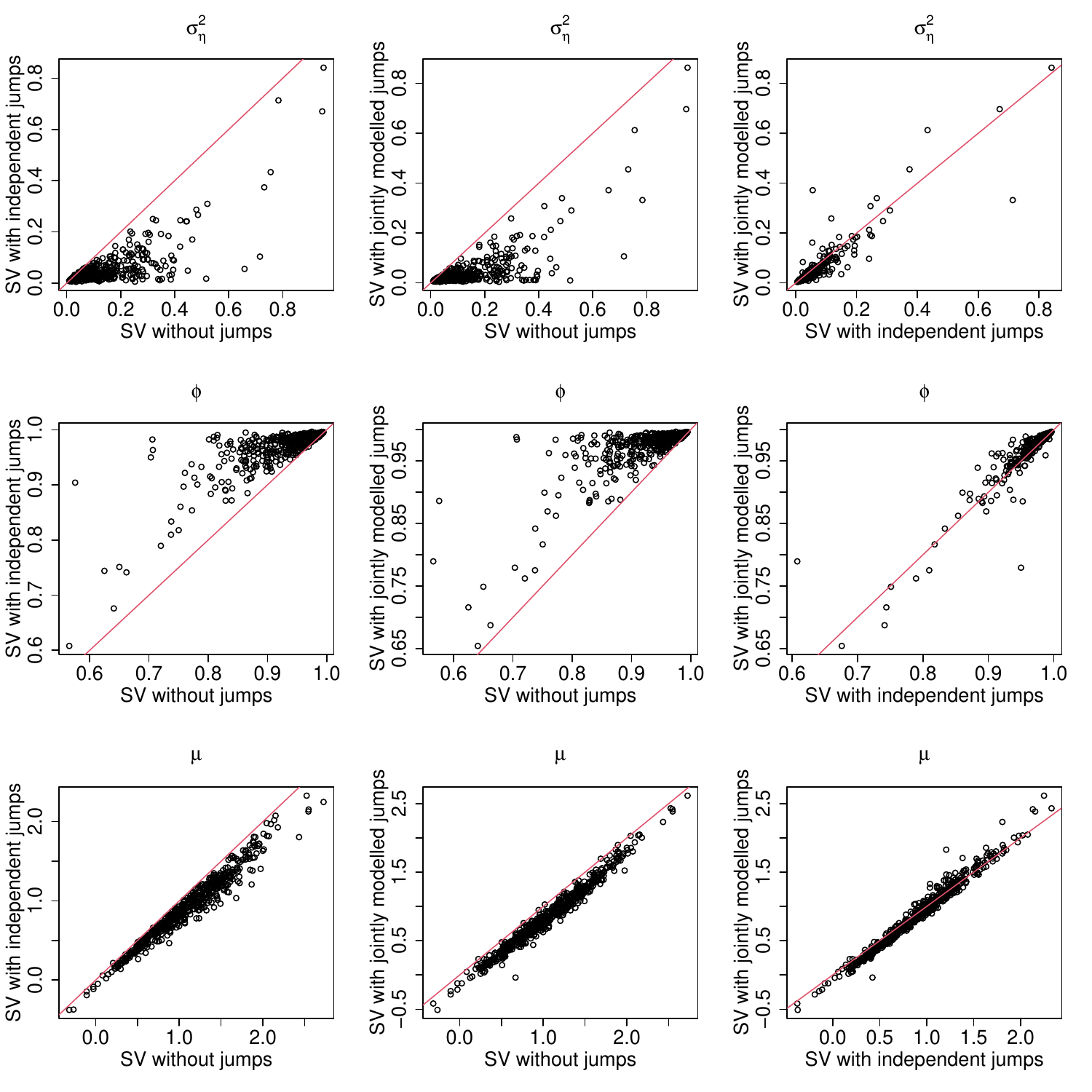}
\vspace*{-0.4cm}
\caption{Posterior means for the parameters of the log-volatility processes of the SV model without jumps, with independent jumps and with jointly modelled jump intensities. The red solid line has zero intercept and slope one.}
\label{fig:vol_pars}
\end{center}
\end{figure}

\begin{figure}[H]
\begin{center}
\includegraphics[scale=0.5]{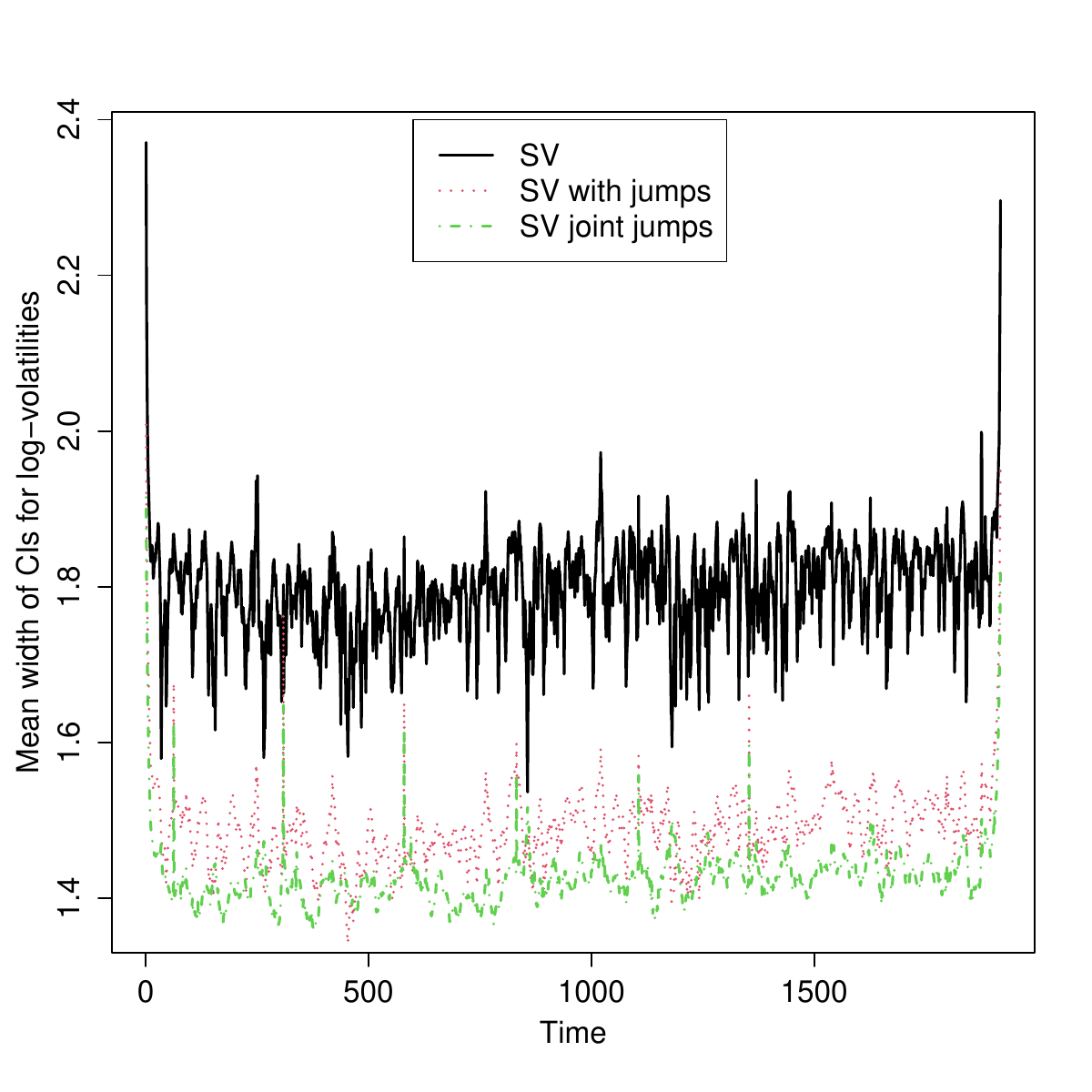}
\vspace*{-0.3cm}
\caption{Mean width of the $571$ $95\%$ credible intervals (CIs) of the log-volatility processes that correspond to SV models without jumps, with independent jumps and with jointly modelled jump intensities. The x-axis corresponds to the $T=1917$ time points with observations.}
\label{fig:vol_CIs}
\end{center}
\end{figure}

\begin{figure}[H]
\begin{center}
\includegraphics[scale=0.5]{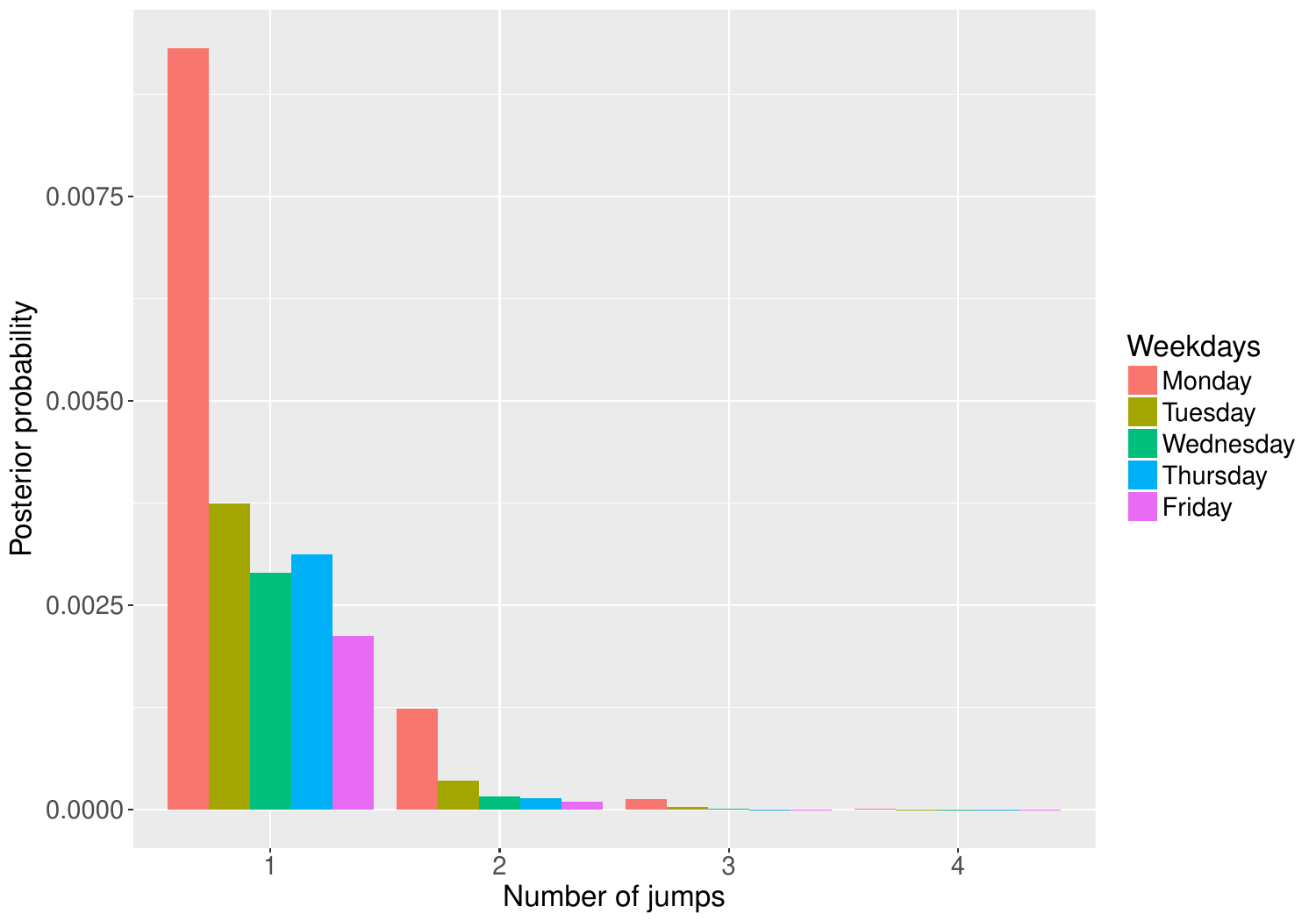}
\vspace*{-0.5cm}
\caption{Posterior distributions for the number of jumps in $571$ stocks from the STOXX $600$ Europe Index conditional on the weekday. The order in which the bars are appearing from left to right is with respect to the order of the weekdays. The probabilities of zero jumps are not displayed.}
\label{fig:weekdays}
\end{center}
\end{figure}

\subsection*{Root mean squared errors of the predictions}

\begin{figure}[H]
       \centerline{\includegraphics[scale=0.5]{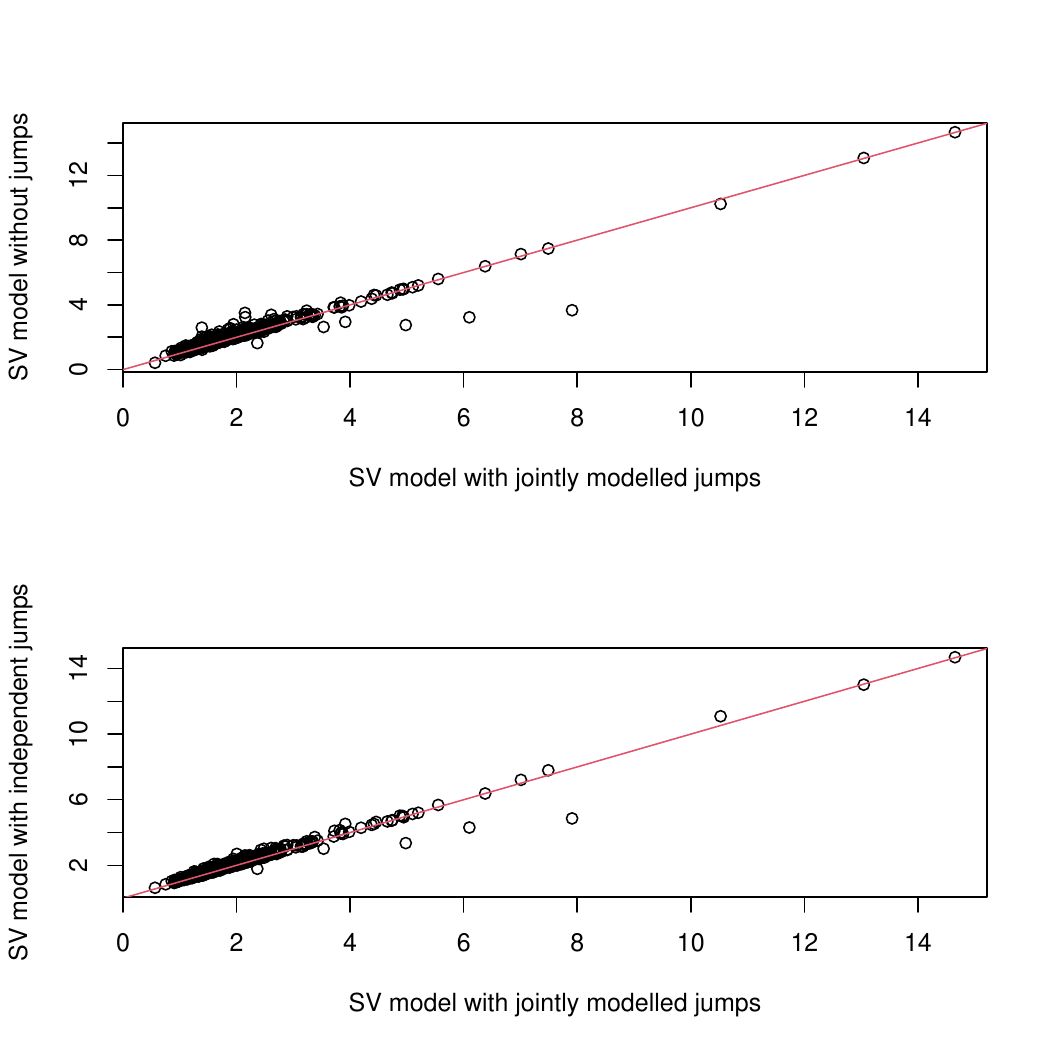}}
       \vspace*{-0.5cm}
       \caption{Root mean squared errors for the SV model with jointly modelled jumps against SV models without jumps (top) and SV models with independent jump intensities (bottom) for each one the $571$ stocks. The red line has zero intercept and slope one.}
       \label{fig:rmses}
   \end{figure}

\subsection*{Predictions by using different prior-hyperparameters for the jump intensities}

In the case of modelling independently, over time and across stocks, the jump intensities we have assumed that $\lambda_{it} \sim Gam(\delta,c)$ and by following the related literature, see for example \cite{chib2002markov}, we set $\delta=1$ and $c=50$ in order the prior expectation of $\lambda_{it}$ to be $0.02$ with prior variance $0.0004$. To model jointly the jump intensities we used the dynamic factor model specified by equations $(3.5)$--$(3.8)$ in the main paper. We specified hyperparameters $\sigma_w^2$, $\sigma_b^2$ and $\mu_b$ such that the induced prior on $\lambda_{it}$ to be similar to the gamma distributions used in the case of independent modelling. Table \ref{tab:intens_priors_supp} displays the mean, the variance and the mode of the prior induced, by the dynamic factor model, on each $\lambda_{it}$, for three different choices of the hyperparameters $\mu_b$, $\sigma^2_b$ and $\sigma^2_w$. Since our aim with joint modelling the jump intensities was to improve the predictions of future observations we perform a sensitivity analysis of the predictions obtained from independent and joint modelled jump intensities with different choices for the prior-hypeparameters of the jump intensities.   

Figure \ref{sens_BF_uni} compares the predictive performance of univariate SV models without jumps with the
performance of SV models with jumps in which the jump intensities are modelled independently
over time and across stocks by using the gamma prior distributions presented in Table \ref{tab:intens_priors_uni_supp}. The figure is constructed by considering the log Bayes factors defined by equation $(5.1)$ of the main paper for the $\ell=30$ out-of-sample observations of our real dataset presented in Section $6$. Figure \ref{sens_BF_mv} presents the log Bayes factors in $(5.5)$ and the 
logarithm of the Bayes factors in $(5.6)$ for the $30$ out-of-sample observations by using the priors on the jump intensities that described in Table \ref{tab:intens_priors_supp}.

\begin{figure}[H]
	\centerline{\includegraphics[scale=0.7]{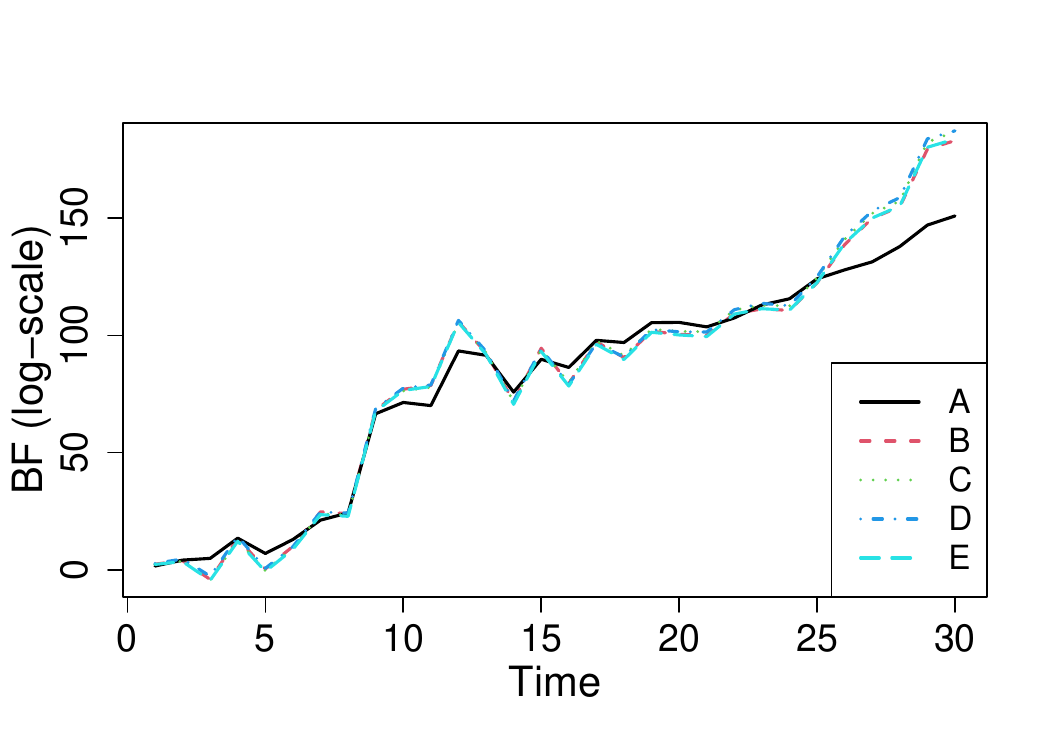}}
	\caption{Log Bayes factors (BF), defined by equation $(5.1)$ of the main paper, in favour of univariate SV models with jumps against univariate SV without jumps models for each one of the $\ell=30$ out-of-sample observations. The letters A, B, C, D and E in the legend correspond to the prior distributions displayed in Table \ref{tab:intens_priors_uni_supp}.}
	\label{sens_BF_uni}
\end{figure}

\begin{table}[H]
\centering
\begin{tabular}{|c c c c|}
\hline
Prior for $\lambda_{it}$ & Mean & Variance & Mode\\
\hline
A: $Gam(1,250)$  &  $0.004$  & $0.00016$ & $0$\\
B: $Gam(1,50)$ & $0.02$ & $0.0004$ & 0\\
C: $Gam(1.05,52.63)$ & $0.02$ & $0.00038$ & $0.001$\\
D: $Gam(1.33,66.67)$ & $0.02$ & $0.00029$ & $0.005$\\
E: $Gam(2,100)$ & $0.02$ & $0.0002$ & $0.01$\\
\hline
\end{tabular}
\caption{Mean, variance and mode of different gamma priors assumed for the jump intensities $\lambda_{it}$ in the case of independent models.}
\label{tab:intens_priors_uni_supp}
\end{table}

\begin{figure}[H]
	\centerline{\includegraphics[scale=0.7]{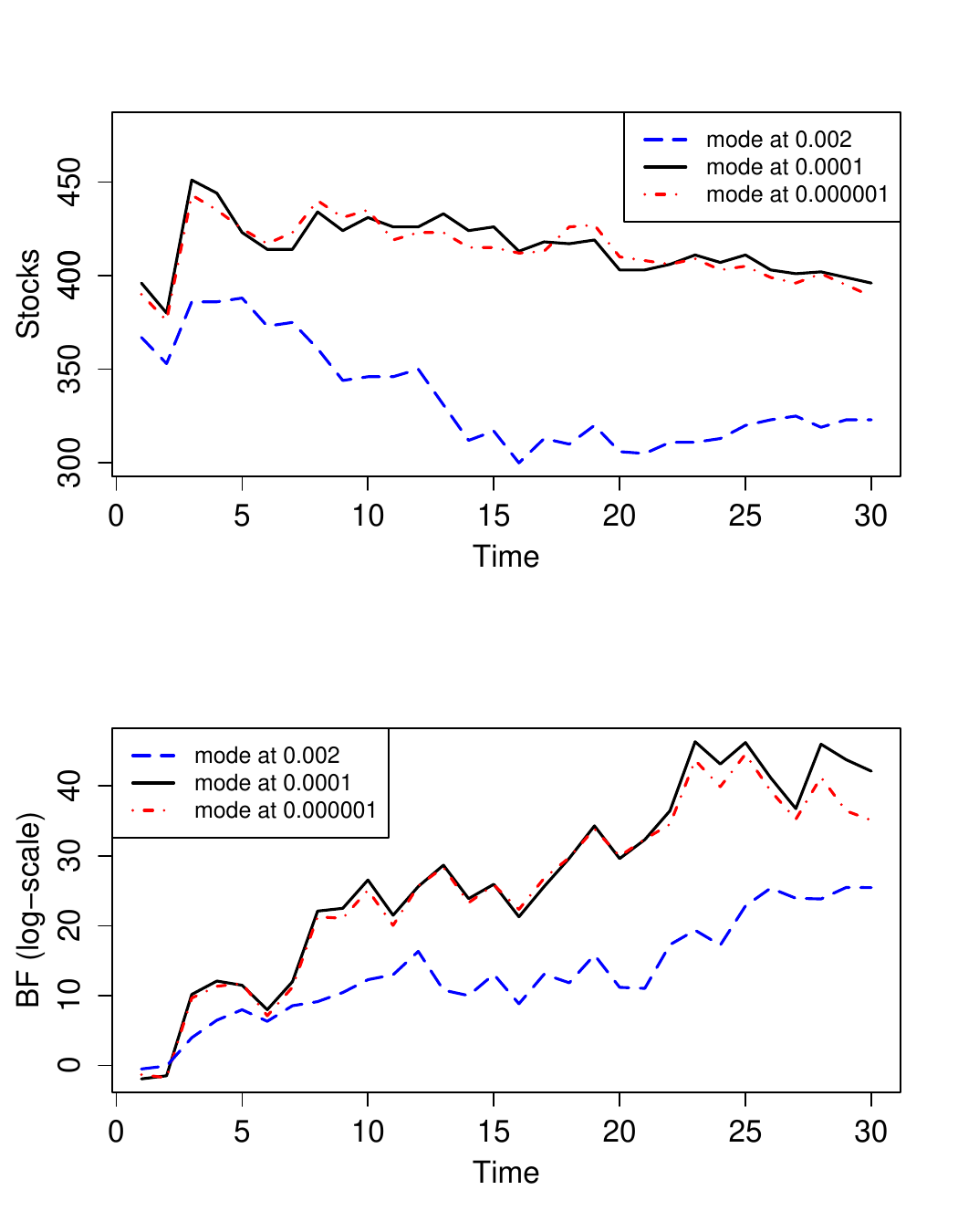}}
	\caption{Top: number of stocks with positive log Bayes factors, defined in equation $(5.5)$ of the main paper, in favour of the proposed SV model with jointly modelled jump intensities against SV models with independent intensities. Bottom: logarithms of the Bayes factors, defined in $(5.6)$, in favour of the proposed SV model with jointly modelled jump intensities against SV models with independent intensities. The quantities have been calculated for each one of the $\ell=30$ out-of-sample observations. The different lines in each plot correspond to induced, by the dynamic factor model, priors on the jump intensities with different modes.}
	\label{sens_BF_mv}
\end{figure}

\subsection*{Evaluation of the sequential Monte Carlo methods}

To evaluate the accuracy of the estimators calculated by employing sequential Monte Carlo methods such as particle filters and the annealed importance sampling we employ the effective sample size (ESS) of the samples drawn from the posterior predictive distributions of interest. The formula of the ESS is given in the $8$th line of Algorithm \ref{alg:pf}. We note that a small ESS indicates large variance for the importance weights calculated during the sequential Monte Carlo algorithms which in turn implies that only few of the weighted samples will dominate the estimates; see for example \cite{chopin2020introduction} for more details.

Figure \ref{pf_ESS} presents for each out-of-sample time point a boxplot with the $571$ ESSs that correspond to particle filters algorithms while Figure \ref{ais_ESS} shows the same quantity for the samples drawn by using the annealed importance sampling algorithm. From the visual inspection of the Figures we concluded that for the vast majority of the stocks the ESS of the $1,000$ posterior predictive indicates an accurate estimation of the quantities of interest: predictive Bayes factors, continuous probability ranked scores of the predictive cumulative distribution functions and interval scores of the prediction credible intervals.

\begin{figure}[H]
	\centerline{\includegraphics[scale=0.65]{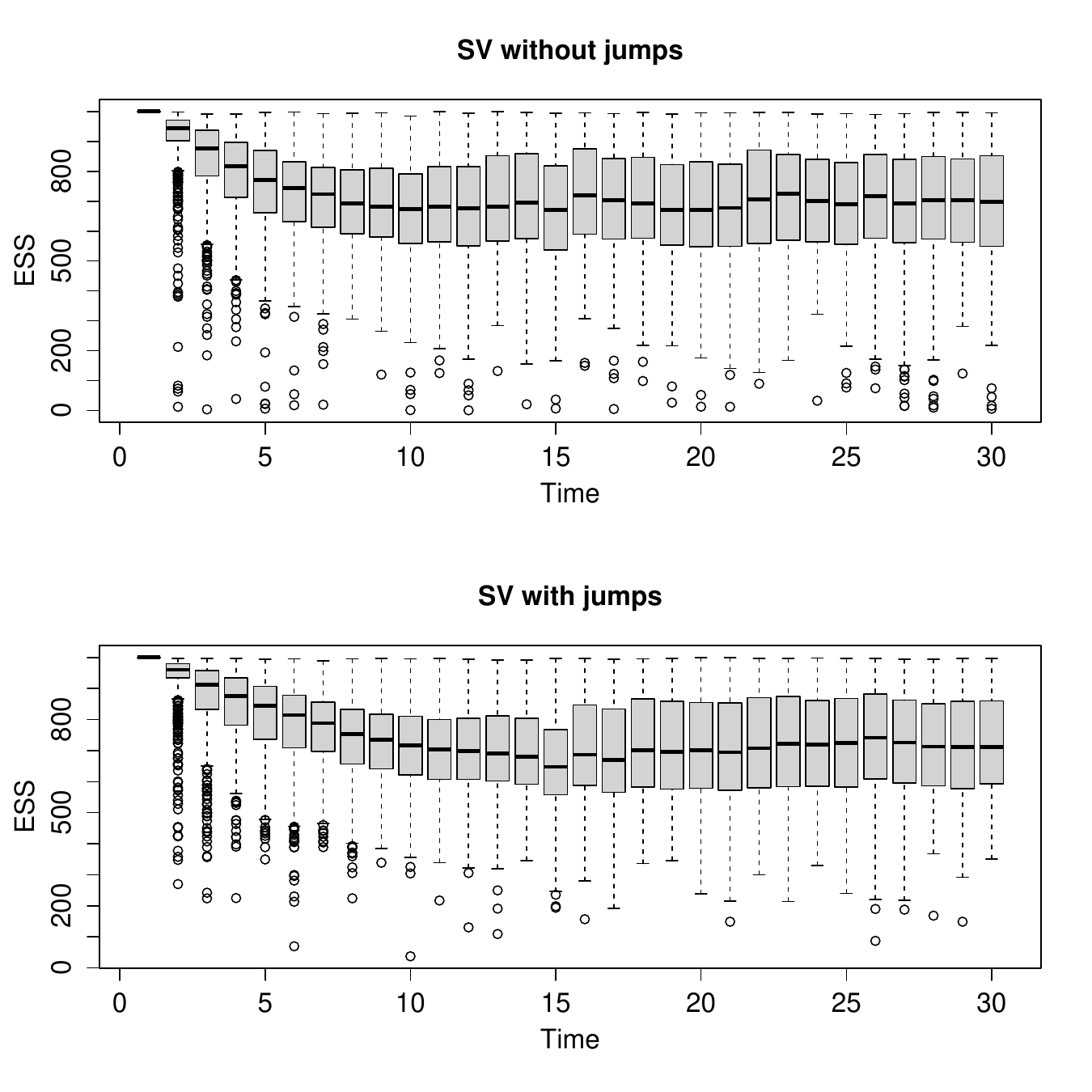}}
	\caption{Boxplots for the effective sample sizes (ESS) of the samples drawn from the posterior predictive distributions of the $571$ stocks by using the particle filter algorithm; y-axis: out-of-sample time points.}
	\label{pf_ESS}
\end{figure}

\begin{figure}[H]
	\centerline{\includegraphics[scale=0.65]{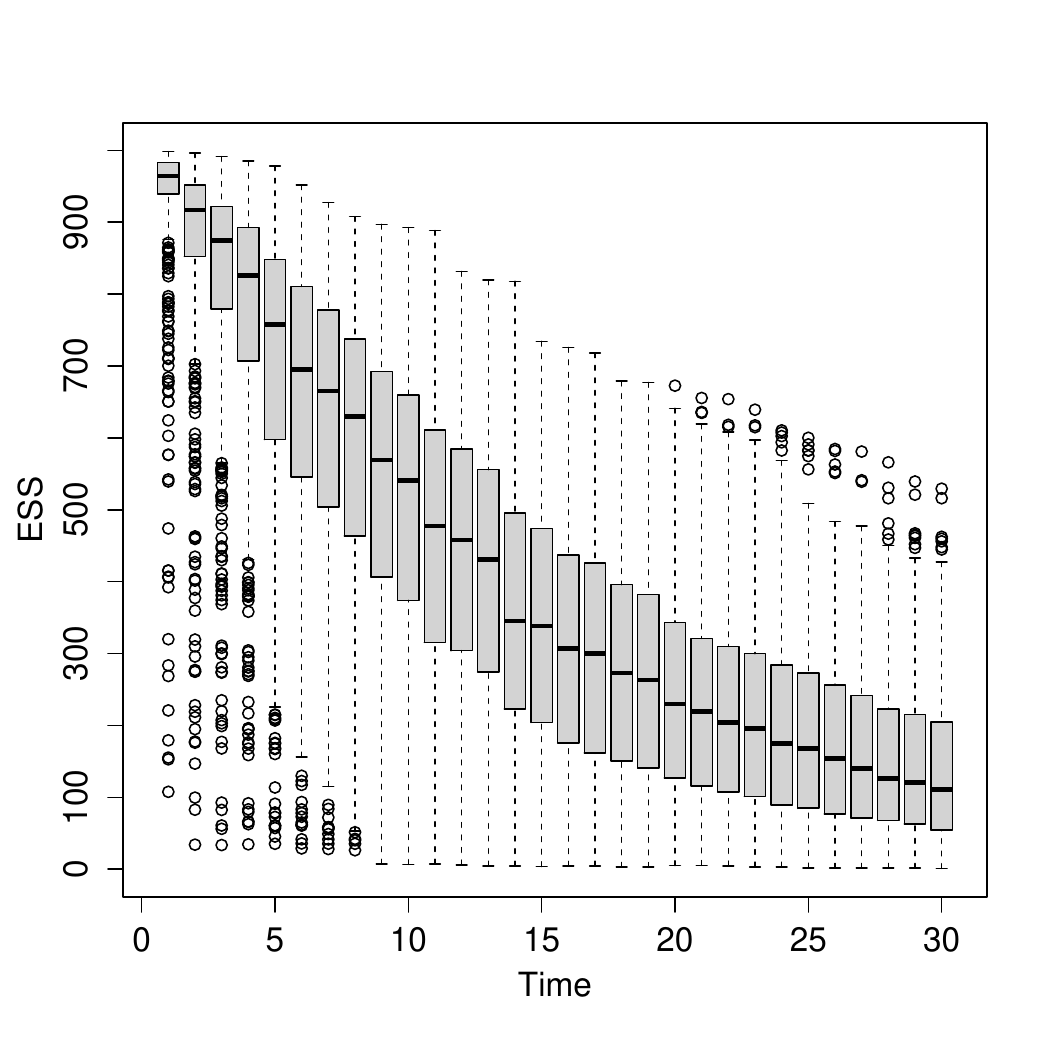}}
	\caption{Boxplots for the effective sample sizes (ESS) of the samples drawn from the posterior predictive distributions of the $571$ stocks by using the annealed importance sampling algorithm; y-axis: out-of-sample time points.}
	\label{ais_ESS}
\end{figure} 

\subsection*{Comparison with SV with $t$ errors}

An alternative formulation that accounts for the heavy tails of the distribution of the observed returns is the SV model in which the observation errors follow the Student's $t$ distribution \cite{harvey1994multivariate}. Bayesian inference for the parameters of SV models with Student's $t$ errors has been examined by \cite{chib2002markov}. Figure \ref{svt} compares the predictive performance of univariate SV models with Student's $t$ errors with the performance of the proposed SV model with jumps and jointly modelled jump intensities. The Figure evaluates the predictive performance of the two models with respect to their interval scores, CRPS as well as root mean squared errors. From the visual inspection of the Figure we conclude that by modelling jointly the jump intensities of SV models we obtain more accurate forecasts rather than using Student's $t$ instead of Gaussian errors. To conduct Bayesian inference for the parameters and the latent states of the SV models with Student's $t$ errors we utilized the r-package \texttt{stochvol} \citep{kastner2014dealing} to collect $1,000$ thinned MCMC samples by storing the outcome of every $100$th iteration after the first $50,000$ iterations. To draw  samples from the posterior predictive distributions of interest we employed the particle filter algorithm summarized by Algorithm \ref{alg:pf} by setting $p(r_{it}|h_{it})$ in \eqref{eq:uniwts} to be the density of the Student's $t$ distribution scaled by $e^{\tfrac{1}{2}h_{it}}$.

\begin{figure}[H]
	\centerline{\includegraphics[scale=0.65]{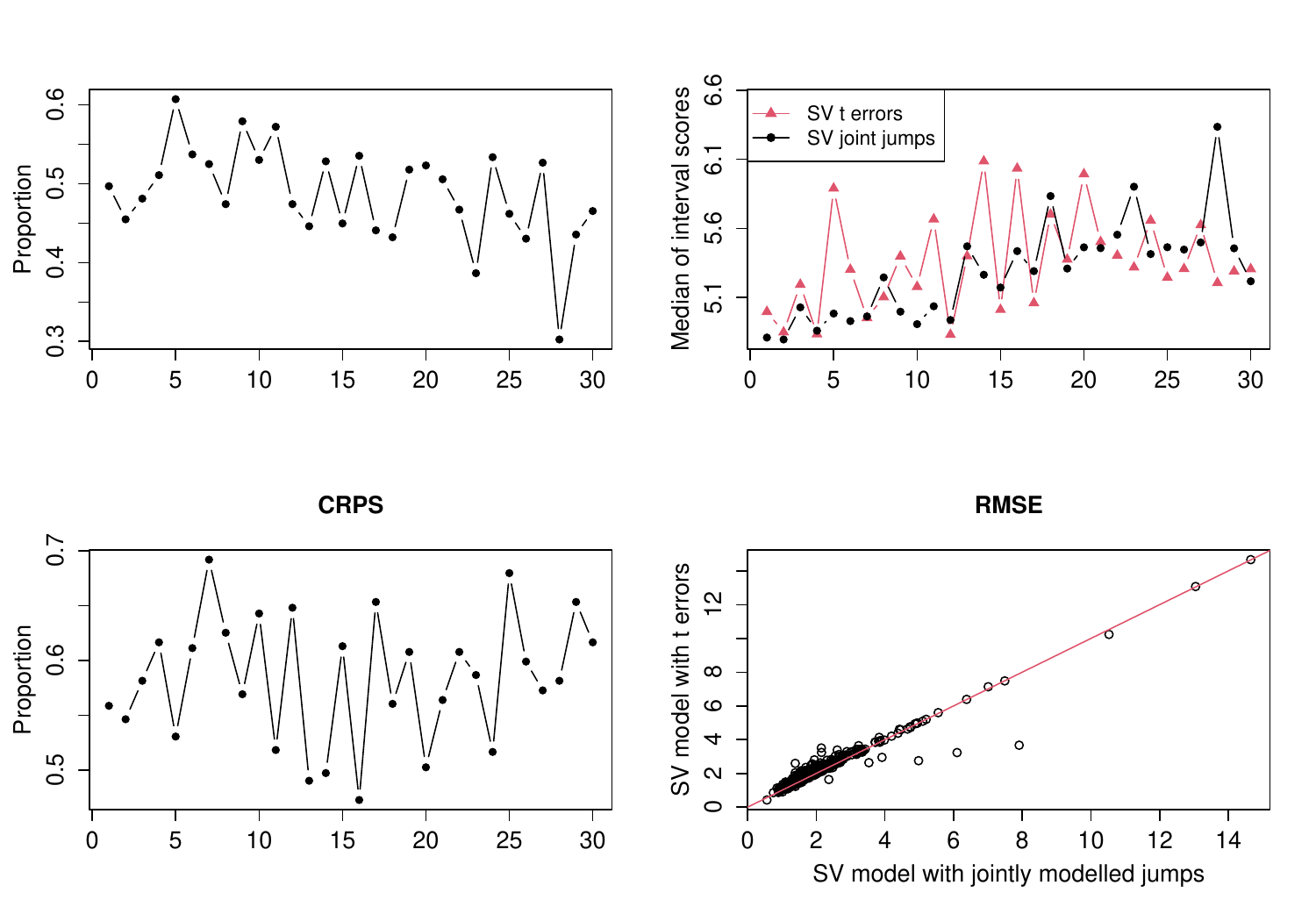}}
	\caption{Top left: proportion of stocks for which the interval score of the $95\%$ predicted intervals of the SV model with $t$ errors is greater than the corresponding score obtained from SV with jointly modelled jumps. Top right: medians of interval scores across the $571$ stocks. Bottom left: proportion of stocks for which the continuous ranked probability score (crps) for the SV model with $t$ errors is greater than the corresponding score obtained from SV with jointly modelled jumps. Bottom right: root mean squared errors for the SV model with jointly modelled jumps against SV models with $t$ errors for each one the $571$ stocks. The red line has zero intercept and slope one.}
	\label{svt}
\end{figure} 

\subsection*{Computing times}

\begin{table}[H]
\centering
\begin{tabular}{|c c |}
\hline
$p$ & Time in seconds \\
\hline
$50$ & $0.25$\\
$100$& $0.52$\\
$200$ & $0.98$\\
$600$& $2.1$\\
\hline
$571$ (real data)& $2.44$\\
\hline
\end{tabular}
\caption{Seconds needed in order the proposed MCMC algorithm to complete one iteration for simulated datasets with different number $p$ of time series each consisted of $1,500$ log-returns. The last line correspond to the time needed in the case of the real dataset in which the length of the time series with stock-returns varies from $1,000$ to $1,917$. The timing conducted on a Laptop with a 1.6 GHz Dual-Core Intel Core i5 CPU running R $4.0.0$ \citep{citeR}.}
\label{tab:sim_comp_times}
\end{table}

\bibliographystyle{Chicago}
\bibliography{refs}

\end{document}